\documentclass[nofootinbib,floatfix,onecolumn,notitlepage,amsfonts,amsmath,amssymb,superscriptaddress]{revtex4-1}

\pdfoutput=1

\usepackage[utf8]{inputenc}
\usepackage[T1]{fontenc}
\usepackage{accents}
\usepackage[colorlinks]{hyperref}

\usepackage[normalem]{ulem}
\usepackage{amsthm}
\usepackage{mathtools}
\usepackage{bbm}
\usepackage{bm}
\usepackage{verbatim}
\usepackage{graphicx}
\usepackage{tikz}
\usepackage{float}
\usepackage{soul}
\usepackage{cancel}

\usetikzlibrary{arrows}

\usepackage{lineno}

\newtheorem{theorem}{Theorem}
\newtheorem{proposition}[theorem]{Proposition}

\newtheorem{observation}[theorem]{Observation}

\def\bra#1{\mathinner{\langle{#1}|}}
\def\ket#1{\mathinner{|{#1}\rangle}}
\def\braket#1#2{\mathinner{\langle{#1|#2}\rangle}}
\def\ketbra#1#2{\mathinner{|{#1}\rangle\!\langle{#2}|}}

{\catcode`\|=\active 
  \gdef\Braket#1{\left<\mathcode`\|"8000\let|\BraVert {#1}\right>}}
\def\BraVert{\egroup\,\mid@vertical\,\bgroup}

\def\dbra#1{\mathinner{\langle\!\langle{#1}|}}
\def\dket#1{\mathinner{|{#1}\rangle\!\rangle}}
\def\dketbra#1{\mathinner{|{#1}\rangle\!\rangle\!\langle\!\langle{#1}|}}

\DeclareMathOperator{\Tr}{Tr}

\DeclareMathOperator{\Span}{span}

\renewcommand\L{\mathcal{L}}
\renewcommand\P{\mathcal{P}}
\newcommand{\M}{\mathcal{M}}
\newcommand{\HS}{\mathcal{H}}

\newcommand{\K}{\mathcal{K}}
\newcommand{\I}{\mathcal{I}}
\newcommand{\E}{\mathcal{E}}
\newcommand{\X}{\mathcal{X}}
\newcommand{\Y}{\mathcal{Y}}
\newcommand{\U}{\mathcal{U}}
\newcommand{\G}{\mathcal{G}}
\newcommand{\W}{\mathcal{W}}
\newcommand{\C}{\mathcal{C}}

\newcommand{\Z}{\mathcal{Z}}
\newcommand{\F}{\mathcal{F}}
\newcommand{\id}{\mathbbm{1}}

\newcommand{\jwchanges}[1]{{\color{black} #1}}

\newcommand\blfootnote[1]{%
  \begingroup
  \renewcommand\thefootnote{}\footnote{#1}%
  \addtocounter{footnote}{-1}%
  \endgroup
}

\newcommand*\patchAmsMathEnvironmentForLineno[1]{%
  \expandafter\let\csname old#1\expandafter\endcsname\csname #1\endcsname
  \expandafter\let\csname oldend#1\expandafter\endcsname\csname end#1\endcsname
  \renewenvironment{#1}%
     {\linenomath\csname old#1\endcsname}%
     {\csname oldend#1\endcsname\endlinenomath}}%
\newcommand*\patchBothAmsMathEnvironmentsForLineno[1]{%
  \patchAmsMathEnvironmentForLineno{#1}%
  \patchAmsMathEnvironmentForLineno{#1*}}%
\AtBeginDocument{%
\patchBothAmsMathEnvironmentsForLineno{equation}%
\patchBothAmsMathEnvironmentsForLineno{align}%
\patchBothAmsMathEnvironmentsForLineno{flalign}%
\patchBothAmsMathEnvironmentsForLineno{alignat}%
\patchBothAmsMathEnvironmentsForLineno{gather}%
\patchBothAmsMathEnvironmentsForLineno{multline}%
}

\begin{document}
\title{Existence of processes violating causal inequalities on time-delocalised subsystems}

\author{Julian Wechs}
\affiliation{QuIC, Ecole Polytechnique de Bruxelles, C.P. 165, Universit\'e Libre de Bruxelles, 1050 Brussels, Belgium\looseness=-1}
\affiliation{Univ.\ Grenoble Alpes, CNRS, Grenoble INP, Institut N\'eel, 38000 Grenoble, France}

\author{Cyril Branciard}
\affiliation{Univ.\ Grenoble Alpes, CNRS, Grenoble INP, Institut N\'eel, 38000 Grenoble, France}

\author{Ognyan Oreshkov}
\affiliation{QuIC, Ecole Polytechnique de Bruxelles, C.P. 165, Universit\'e Libre de Bruxelles, 1050 Brussels, Belgium\looseness=-1}

\date{\today}

\begin{abstract}
It has been shown that it is theoretically possible for there to exist quantum and classical processes in which the operations performed by separate parties do not occur in a well-defined causal order. 
A central question is whether and how such processes can be realised in practice. 
In order to provide a rigorous framework for the notion that certain such processes have a realisation in standard quantum theory, the concept of \emph{time-delocalised quantum subsystem} has been introduced.
In this paper, we show that realisations on time-delocalised subsystems exist for all unitary extensions of tripartite processes. 
This class contains processes that violate \emph{causal inequalities}, i.e., that can generate correlations that witness the incompatibility with definite causal order in a device-independent manner, and whose realisability has been a central open problem. We consider a known example of such a tripartite classical process that has a unitary extension, and study its realisation on time-delocalised subsystems.
We then discuss this finding with regard to the assumptions that underlie causal inequalities, and argue that they are indeed a meaningful concept to show the absence of a definite causal order between the variables of interest. 
\end{abstract}

\maketitle

\section*{Introduction}

\blfootnote{This is a post-peer-review, pre-copyedit version of an article
published in \emph{Nature Communications} 14, 1471 (2023). The final
authenticated version is available online at: \url{https://doi.org/10.1038/s41467-023-36893-3}.}
The concept of causality is essential for physics and for our perception of the world in general.
Our usual understanding is that events take place in a definite causal order, with past events influencing future events, but not vice versa.
One may however wonder whether this notion is really fundamental, or whether scenarios without such an underlying order can exist.
In particular, the questions of what quantum theory implies for our understanding of causality, and what new types of causal relations arise in the presence of quantum effects, have recently attracted substantial interest.
This investigation is motivated both by foundational and by applied questions.
On the one hand, it is expected to lead to new conceptual insights into the tension between general relativity and quantum theory~\cite{hardy05,oreshkov12,zych19}. 
On the other hand, it also opens up new possibilities in quantum information processing~\cite{chiribella13}.

A particular model for the study of quantum causal relations is the \emph{process matrix framework}~\cite{oreshkov12}, where one considers multiple parties which perform operations that locally abide by the laws of quantum theory, but that are not embedded into any a priori causal order.
As it turns out, this framework indeed allows for situations where the causal order between the parties is not well-defined (see e.g. Refs.~\cite{oreshkov12,araujo15,oreshkov16,branciard16,abbott16,wechs19}). Moreover, some of these processes, called \textit{noncausal}, can produce correlations that violate \textit{causal inequalities}~\cite{oreshkov12,oreshkov16,branciard16,abbott16,baumeler14a,baumeler16}, which witnesses the incompatibility with definite causal order in a device-independent manner, similarly to the way a violation of a Bell inequalities witnesses the incompatibility with local hidden variables~\cite{bell87}. 
A central question is which of these processes with indefinite causal order have a practical realisation, and in what physical situations they can occur.
It has been speculated that indefinite causal order could arise in exotic physical regimes, such as at the interface of quantum theory and gravity~\cite{hardy05,oreshkov12,zych19}.
However, there are also processes with indefinite causal order that have an interpretation in terms of standard quantum theoretical concepts. A paradigmatic example is the \emph{quantum switch}~\cite{chiribella13}, a process in which the order between two operations is controlled coherently by a two-dimensional quantum system. This control qubit may be prepared in a superposition state, which leads to a \emph{superposition of causal orders}. Although the quantum switch cannot violate causal inequalities~\cite{araujo15,oreshkov16,wechs18,purves21} (however, see recent results in the presence of additional causal assumptions~\cite{gogioso22,vanderlugt22}), it can be proven incompatible with a definite causal order in a device-dependent sense~\cite{araujo15,oreshkov16}.

In order to demonstrate indefinite causal order in practice, a number of experiments that realise such coherent control of orders have been implemented in the laboratory~\cite{procopio15,rubino17,rubino17a,goswami18,goswami18a,wei19,guo18,taddei20,rubino21,cao22,Nie_2020}, however their interpretation as genuine \emph{realisations} of indefinite causal order has remained controversial~\cite{maclean17,oreshkov18,vilasini22,ormrod22}. 
Indeed, the claim that indefinite causal order can be realised in standard quantum scenarios seems contradictory at first sight---after all, such experiments admit a description in terms of standard quantum theory, where physical systems by assumption evolve with respect to a fixed background time, and it is therefore not manifest how the causal order between operations could possibly be indefinite.
A resolution of this apparent contradiction was proposed in Ref.~\cite{oreshkov18}, where it was shown that certain processes with indefinite causal order can be seen to take place as part of standard quantum mechanical evolutions if the latter are described in terms of suitable systems. 
The twist is to consider a more general type of system than usually studied, namely so-called \emph{time-delocalised subsystems}, which are nontrivial subsystems of composite systems whose constituents are associated with different times. 
This concept provides a rigorous underpinning for the interpretation of previous laboratory experiments as realisations of processes with indefinite causal order---when the experiment is described with respect to such an alternative, operationally equally meaningful factorisation of the Hilbert space, it acquires precisely the form of the process with indefinite causal order.
It was then shown in Ref.~\cite{oreshkov18} that this argument extends to an entire class of quantum processes, namely \emph{unitary extensions of bipartite processes}, as well as a class of isometric extensions, whose relation to the unitary class is not yet fully understood. 
The generalisation of these constructions to more parties has however remained an open question. In particular, it has remained an open question whether processes violating causal inequalities can be realised in a similar way. It is in fact generally believed that such processes could not be realised deterministically within the known physics~\cite{purves21}.

In this paper, we extend the proof of realisability on time-delocalised subsystems to all unitary extensions of tripartite processes. This class contains examples of processes that can violate causal inequalities, showing that they have realisations with the tools of known physics in a well-defined sense.

This work is structured as follows. We set the stage by reviewing the process matrix framework, as well as the notion of time-delocalised subsystems. 
We present the general tripartite construction, and we study an example of a tripartite noncausal process on time-delocalised subsystems. 
We then analyse our finding with regard to the assumptions that underlie causal inequalities, and argue that their violation witnesses the absence of a definite causal order in a meaningful way.

\section*{Results}

\textbf{Notations} \label{sec:main_notations} \quad We start by introducing some notations.
We denote the Hilbert space of some quantum system $Y$ by $\HS^Y$, the dimension of $\HS^Y$ by $d_Y$ and the space of linear operators over $\HS^Y$ by $\L(\HS^Y)$. Each such Hilbert space comes with a preferred, \emph{computational} basis generally denoted $\{\ket{i}^Y\}_{i}$.
The identity operator on $\HS^{Y}$ is denoted by $\id^Y$.
We also use the notation $\HS^{YZ} \coloneqq \HS^Y \otimes \HS^Z$ for the tensor product of two Hilbert spaces $\HS^Y$ and $\HS^Z$ (whose computational basis is built as the tensor product of the two subsystems' computational bases).
For two isomorphic Hilbert spaces $\HS^Y$ and $\HS^Z$, we denote the \emph{identity operator} between these spaces (i.e. the canonical isomorphism, which maps each computational basis state $\ket{i}^Y$ of $\HS^Y$ to the corresponding computational basis state $\ket{i}^Z$ of $\HS^Z$) by $\id^{Y \to Z} \coloneqq \sum_i \ket{i}^Z\bra{i}^Y$, and its pure Choi representation (see Methods section \hyperref[app:choi]{``The Choi isomorphism and the link product''}) by $\dket{\id}^{YZ} \coloneqq \sum_i \ket{i}^Y \otimes \ket{i}^Z$.
(Generally, superscripts on vectors indicate the Hilbert space they belong to, which may be omitted when clear from the context).
Moreover, we will often abbreviate $X_IX_O$ to $X_{IO}$ for the incoming and outgoing systems of the party $X$ (see below).\\

\textbf{The process matrix framework} \label{sec:main_pm_general} \quad In the following, we briefly outline the process matrix framework, originally introduced in Ref.~\cite{oreshkov12}. 
We consider multiple parties $X = A, B, C, \ldots$ performing operations that are locally described by quantum theory. 
That is, each party has an incoming quantum system $X_I$ with Hilbert space $\HS^{X_I}$ and an outgoing quantum system $X_O$ with Hilbert space $\HS^{X_O}$, and can perform arbitrary quantum operations from $X_I$ to $X_O$.
A quantum operation is most generally described by a \textit{quantum instrument}, that is, a collection of completely positive (CP) maps $\{\M_X^{[o_X]}\}_{o_X}$, with each $\M_X^{[o_X]}: \L(\HS^{X_I}) \to \L(\HS^{X_O})$ associated to a classical outcome $o_X$, and with the sum over the classical outcomes yielding a completely positive and trace-preserving (CPTP) map.

The process matrix framework was conceived to study the most general correlations that can arise between such parties, without making any a priori assumption about the way they are connected. 
In Ref.~\cite{oreshkov12}, it was shown that these correlations can most generally be expressed as
\begin{align}
\label{eq:generalisedBornRule}
P(o_A,o_B,o_C,\ldots) =  W * (M_A^{[o_A]} \otimes M_B^{[o_B]} \otimes M_C^{[o_C]} \otimes \cdots).
\end{align}
Here, $M_X^{[o_X]} \in \L(\HS^{X_{IO}})$ are the Choi representations of the local CP maps $\M_X^{[o_X]}$ and $``*"$ denotes the \emph{link product}~\cite{chiribella08,chiribella09}, a mathematical operation that describes the composition of quantum operations in terms of their Choi representation (see Methods section \hyperref[app:choi]{``The Choi isomorphism and the link product''}).
$W \in \L(\HS^{A_{IO}B_{IO}C_{IO}\ldots})$ is a Hermitian operator called the \emph{process matrix}. 
The requirement that the probabilities in Eq.~\eqref{eq:generalisedBornRule} should be non-negative, even when the operations of the parties are extended so as to act on additional, possibly entangled ancillary input systems, is equivalent to $W \ge 0$.
The requirement that the probabilities should be normalised (i.e., they should sum up to $1$ for any choice of local operations) is equivalent to $W$ satisfying certain linear constraints~\cite{oreshkov12,araujo15,oreshkov16,araujo17,wechs19}, and having the trace $\Tr(W) = d_{A_O} d_{B_O} d_{C_O} \ldots \ $. 

The process matrix is the central object of the formalism, which describes the physical resource or environment through which the parties are connected.
Mathematically, the process matrix defines (i.e., it is the Choi representation of) a quantum channel $\W: \L(\HS^{A_O B_O C_O \ldots}) \to \L(\HS^{A_I B_I C_I \ldots})$ from all output systems of the parties to their input systems.
Eq.~\eqref{eq:generalisedBornRule} then describes the composition of that channel with the local operations, which can be interpreted as a \emph{circuit with a cycle} as represented graphically (for the bipartite case) in Fig.~\hyperref[fig:process_matrices]{1(a)}. 

\begin{figure}[h]
      \includegraphics[width=0.89\textwidth]{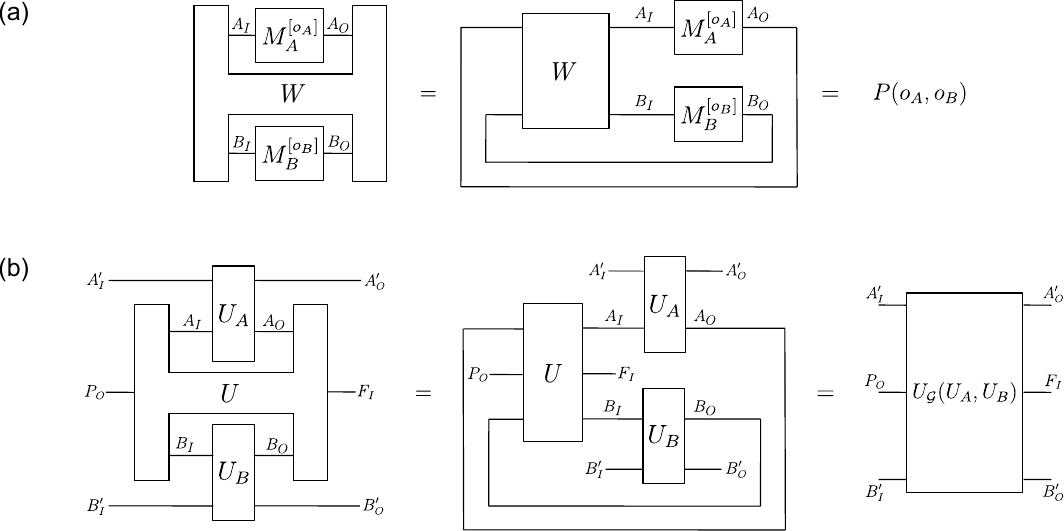}
    \caption{\textbf{Process matrix scenarios as cyclic circuits.} 
(a) In the process matrix framework, the operations performed by the parties (here, Alice and Bob) are composed with the \emph{process matrix}, which defines a channel from the output systems $A_O B_O$ of the parties back to their input systems $A_I B_I$. This composition can be seen as a cyclic circuit, and provides  the probabilities for the classical outcomes $o_A$ and $o_B$.\\
(b) Composing a unitarily extended process matrix with unitary operations performed by the parties gives rise to a unitary operation from the outgoing system $P_O$ of the \emph{global past} party $P$ and the incoming ancillas of the parties $A_I'$, $B_I'$ to the incoming system $F_I$ of the \emph{global future} party $F$ and the outgoing ancillas of the parties $A_O'$, $B_O'$~\cite{araujo17}.}
        \label{fig:process_matrices}
\end{figure}

Through the top-down approach outlined here, one recovers standard quantum scenarios, such as joint measurements on multipartite quantum states, or, more generally, quantum circuits in which the parties apply their operations in a fixed causal order (and the process matrix corresponds to the acyclic circuit fragment consisting of the operations in between the parties~\cite{chiribella09,gutoski06}). 
However, one also finds processes that are incompatible with any definite causal order between the local operations. Such processes are said to be \emph{causally nonseparable}~\cite{oreshkov12,araujo15,oreshkov16,wechs19}. Furthermore, some causally nonseparable processes can generate correlations $P(o_A,o_B,o_C,\ldots|i_A,i_B,i_C,\ldots)$, where $i_X$ are local classical inputs based on which the local operations are chosen, that violate so-called causal inequalities~\cite{oreshkov12,oreshkov16,branciard16,abbott16,baumeler14a,baumeler16}, which certifies their incompatibility with a definite causal order in a device-independent way. Such processes are referred to as \emph{noncausal}.

A class of processes that is of particular interest in this paper is that of \emph{unitarily extendible processes}, which were first discussed in Ref.~\cite{araujo17}.
A \emph{unitary extension} of a process matrix $W$ is a process matrix which involves an additional party $P$ with a trivial, one-dimensional input Hilbert space, as well as an additional party $F$ with a trivial, one-dimensional outgoing Hilbert space, such that the corresponding channel from $P_O A_O B_O C_O \ldots$ to $F_I A_I B_I C_I \ldots$ is unitary, and such that the original process matrix $W$ is recovered when $P$ prepares some fixed state and $F$ is traced out. 
That is, the extended process matrix is a rank-one projector $\dketbra{U}$, where $\dket{U}$ is the pure Choi representation (see \hyperref[app:choi]{Methods}) of a unitary $U: \HS^{P_O A_O B_O C_O \ldots} \to  \HS^{F_I A_I B_I C_I \ldots}$, which satisfies
\begin{equation}
    W = \dket{U}\dbra{U} * (\ket{0}\bra{0}^{P_O} \otimes \id^{F_I}).
\end{equation}
The additional parties $P$ and $F$ can be interpreted as being in the \emph{global past}, respectively \emph{global future}, of all other parties, since they do not receive, respectively send out, a quantum system.

Note that the unitary extension also needs to be a valid process matrix, i.e., it needs to satisfy the above-mentioned constraints which ensure that it yields valid probabilities when the parties (including $P$ and $F$) perform arbitrary local operations. 
In Ref.~\cite{araujo17}, it was found that some process matrices do not admit such a unitary extension, and unitary extendibility was postulated as a necessary condition for a process matrix to describe a physically realisable scenario.
It was also shown that unitary extensions are equivalent to processes that preserve the reversibility of quantum operations.
That is, when the slots of $P$ and $F$ are left open, and all other parties perform unitary operations $\U_X: \L(\HS^{X_I X_I'}) \to \L(\HS^{X_O X_O'})$, which act on $X_I$ and $X_O$ as well as some (possibly trivial) additional ancillary incoming and outgoing systems $X_I'$ and $X_O'$, the resulting \emph{global operation}, which takes the initial systems $P_O A_I' B_I' C_I' \ldots$ to the final systems $F_I A_O' B_O' C_O' \ldots$, is again unitary (see Fig.~\hyperref[fig:process_matrices]{1(b)}).  

In this case, the Choi representations of the local operations, as well as the unitarily extended process matrix, are rank-one projectors, and we can describe their composition in terms of their pure Choi representations.
The global unitary operation $\U_\G(\U_A,\U_B,\U_C,\cdots): \L(\HS^{P_O A_I' B_I' C_I' \ldots}) \to \L(\HS^{F_I A_O' B_O' C_O' \ldots})$, in its pure Choi representation, is given by 
\begin{equation}
\dket{U_\G(U_A,U_B,U_C,\cdots)} = \dket{U} * \big( \dket{U_A} \otimes \dket{U_B} \otimes \dket{U_C}  \otimes \ldots \big) \quad \in \HS^{P_O A_{I}' B_{I}' C_{I}' \ldots F_I A_{O}' B_{O}' C_{O}' \ldots}
\label{eq:globalTrafo_pure}
\end{equation}
where $\dket{U_X} \in \HS^{X_{IO}X_{IO}'}$ are the pure Choi representations of the local unitary operations $\U_X$, and $``*"$ denotes here the so-called \emph{vector link product}~\cite{wechs18} (cf. \hyperref[app:choi]{Methods}). 
In the following, we are going to refer to $\dket{U}$ as the \emph{process vector} of the unitary process under consideration. 

The process matrices that we are concerned with in this work are unitary extensions of bipartite or tripartite process matrices. 
Moreover, any local operation can be dilated to a unitary channel acting on the original incoming and outgoing systems together with an additional incoming and outgoing ancilla, followed by a measurement of the outgoing ancilla.
Throughout the paper, we will therefore not consider the actions of the global past and global future parties explicitly, but rather work with the description as per Eq.~\eqref{eq:globalTrafo_pure} in terms of pure Choi representations, which is convenient. 
We will also take the incoming and outgoing Hilbert spaces of all parties, except for $P$ and $F$, to be of equal dimension $d_{X_I} = d_{X_O} \eqqcolon d$. 
This simplification saves us some technicalities, and it does not entail any loss of generality. 
Namely, if these dimensions do not match, one can treat the process under consideration as a part of an extended process with process vector $\dket{U} \otimes \dket{\id}^{P_A\tilde{A}_I} \otimes \dket{\id}^{\tilde{A}_O F_A} \otimes \dket{\id}^{P_B\tilde{B}_I} \otimes \dket{\id}^{\tilde{B}_OF_B} \otimes \dket{\id}^{P_C\tilde{C}_I} \otimes \dket{\id}^{\tilde{C}_OF_C} \otimes \ldots$, which involves additional identity channels between additional outgoing (incoming) Hilbert spaces $\HS^{P_A}$, $\HS^{P_B}$, $\HS^{P_C}, \ldots$ ($\HS^{F_A}$, $\HS^{F_B}$, $\HS^{F_C}, \ldots$) of the global past (future) party, and additional incoming (outgoing) Hilbert spaces $\HS^{\tilde{A}_I}$, $\HS^{\tilde{B}_I}$, $\HS^{\tilde{C}_I}, \ldots$ ($\HS^{\tilde{A}_O}$, $\HS^{\tilde{B}_O}$, $\HS^{\tilde{C}_O}, \ldots$) of the remaining parties, whose dimensions are chosen such that $d_{X_I \tilde X_I} = d_{X_O \tilde X_O} = d$ for all parties (except $P$ and $F$).\\

\textbf{Time-delocalised subsystems and operations} \label{sec:td_general} \quad In this section, we discuss the concept of time-delocalised subsystem, first introduced in Ref.~\cite{oreshkov18}.
Briefly summarised, the idea is that a quantum circuit, consisting of operations that act at definite times on specific input and output systems, can be described in terms of a different choice of systems, corresponding to an alternative factorisation of the joint Hilbert spaces of the input and output systems of operations at different times. In general, the new systems may be delocalised relative to the old ones and thus spread over different times.
When described in terms of such alternative time-delocalised subsystems, the circuit generally contains cycles as considered in the process matrix framework (see Fig.~\ref{fig:process_matrices}).
We first discuss the general formalisation of this idea, then we recall how it applies to the case of the quantum switch, as well as general unitary extensions of bipartite processes, for which it was shown in Ref.~\cite{oreshkov18} that realisations on such time-delocalised subsystems always exist.

The concept of time-delocalised subsystem arises from combining two notions from standard quantum theory, namely the definition of quantum subsystem decompositions in terms of tensor product structures, and the fact that a fragment of a quantum circuit containing multiple operations implements itself a quantum operation from all its incoming to all its outgoing systems. 

In quantum theory, the division of a physical system into subsystems is formally described through the choice of a \emph{tensor product structure}. 
Equipping a given Hilbert space $\HS^Y$, corresponding to some quantum system $Y$, with a tensor product structure means choosing an isomorphism (i.e., a unitary transformation) $J: \HS^Y \to \bigotimes_{i=1}^n \HS^{Y_n}$, where $\HS^{Y_1},\ldots,\HS^{Y_n}$ are Hilbert spaces of dimensions $d_{Y_1},\ldots,d_{Y_n}$, with $\Pi_{i=1}^n d_{Y_i} = d_Y$. Such a choice establishes a notion of locality on $\HS^Y$, and defines a decomposition of the system $Y$ into subsystems $Y_1, \ldots, Y_n$. 
For instance, the operators in $\L(\HS^Y)$ that are local on the subsystem $Y_i$ are those of the form $J^\dagger (O^{Y_i} \otimes \id^{Y_1,\ldots,Y_{i-1}Y_{i+1}\ldots Y_n})J$
with $O^{Y_i} \in \L(\HS^{Y_i})$. 
(Equivalently, the tensor product structure can also be defined in terms of the algebras of operators that are local on the different subsystems~\cite{zanardi04}.)
Since the choice of such a tensor product structure is not unique, there are many different ways to view $\HS^Y$ as the state space of some quantum system with multiple subsystems.

In standard quantum theory, physical systems evolve with respect to a fixed background time.
At an abstract level, such standard quantum mechanical time evolution can be described in terms of a \emph{quantum circuit}, that is, a collection of quantum operations (pictorially represented by boxes) that are composed through quantum systems (pictorially represented by wires) in an acyclic network. 
The operations in such a quantum circuit thus act on their incoming and outgoing quantum systems (which may consist of several subsystems) at definite times.
One may however also consider quantum operations that act on several subsystems associated with different times.
In fact, this possibility arises naturally within the quantum circuit framework. Namely, if one considers a generic fragment of a quantum circuit
containing many operations, that fragment implements a quantum operation from the joint system of all wires that enter into it, to the joint system of all wires that go out of it \cite{chiribella09}, where the incoming and outgoing wires are generally associated with Hilbert spaces at different times (see Fig.~\hyperref[fig:circuit_generic]{2(a)} for an example).

\begin{figure}[h]
     \includegraphics[width=0.89\textwidth]{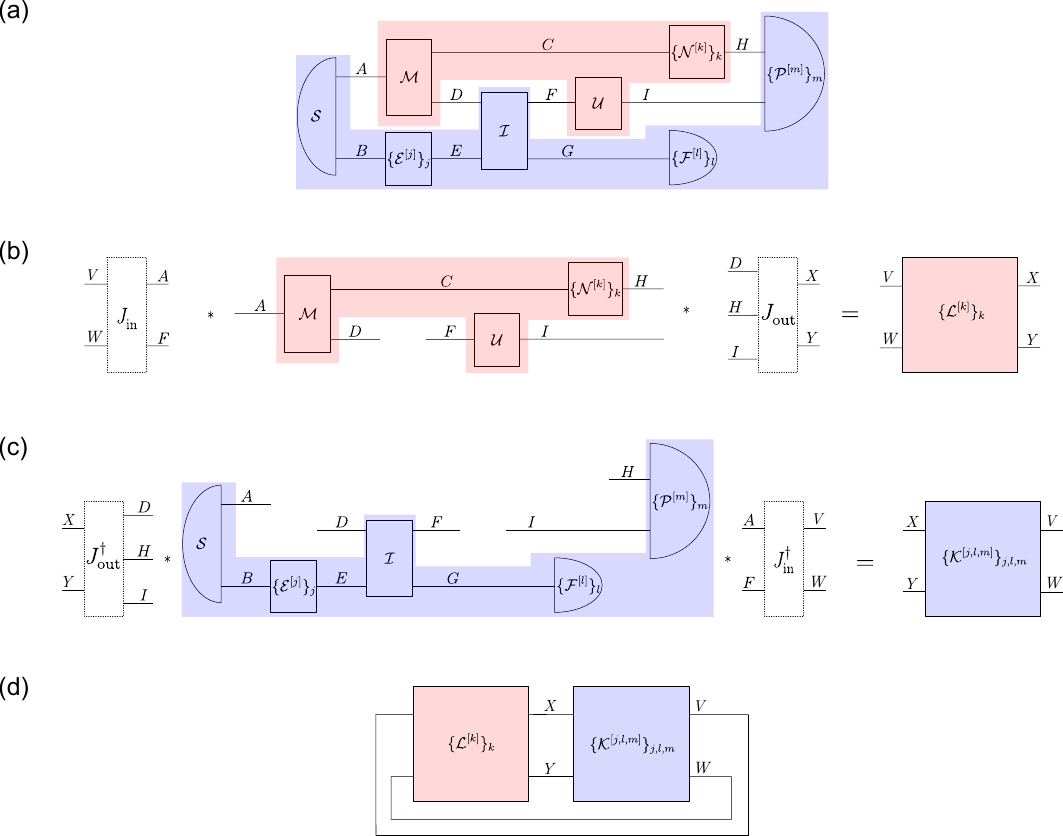}
     \caption{\textbf{Description of a quantum circuit in terms of time-delocalised subsystems.} (a) Example of a quantum circuit, consisting of quantum operations $\mathcal{S}$, $\M$, $\{\E^{[j]}\}_j$, $\I$, $\U$, $\{\mathcal{N}^{[k]}\}_k$, $\{\mathcal{F}^{[l]}\}_l$, $\{\P^{[m]}\}_m$, which are composed through the systems $A$, $B$, $C$, $D$, $E$, $F$, $G$, $H$, $I$, and a decomposition thereof into fragments, corresponding to the red and blue boxes. 
The red fragment implements itself a quantum operation from the incoming systems $A$ and $F$ to the outgoing systems $D$, $H$  and $I$, which are each associated with different times.
It is composed with its complement, the blue fragment, which implements a quantum operation from the systems $D$, $H$, $I$ to the systems $A$, $F$.\\
(b) Description of the red circuit fragment in terms of time-delocalised subsystems $V$, $W$, $X$, $Y$, which are defined by isomorphisms $J_{\text{in}}: \HS^{VW} \to \HS^{AF}$ and $J_{\text{out}}: \HS^{DHI} \to \HS^{XY}$. We obtain a new operation $\{ \L^{[k]} \}_k$ from $V$, $W$ to $X$, $Y$.\\
(c) Description of the blue circuit fragment in terms of the time-delocalised subsystems $V$, $W$, $X$, $Y$. We obtain a new operation $\{ \K^{[j,l,m]} \}_{j,l,m}$ from $X$, $Y$ to $V$, $W$.\\
(d) In the new subsystem description in terms of the time-delocalised subsystems $V$, $W$, $X$, $Y$, we obtain a cyclic circuit composed of $\{ \L^{[k]} \}_k$ and $\{ \K^{[j,l,m]} \}_{j,l,m}$.}
    \label{fig:circuit_generic}
\end{figure}

One may choose to describe such a quantum operation implemented by a fragment with respect to a different subsystem decomposition.
Formally, this is achieved by composing its incoming, respectively outgoing, wires with some isomorphisms that define a different tensor product structure on the corresponding joint Hilbert spaces (see Fig.~\hyperref[fig:circuit_generic]{2(b)}).  
The resulting subsystems are then in general not associated with a definite time. 
This is what one understands by \emph{time-delocalised subsystems}.

To describe the full circuit in terms of these newly chosen time-delocalised subsystems, the operation implemented by the complement of the fragment under consideration needs to be composed with precisely the inverse of the chosen isomorphisms (see Fig.~\hyperref[fig:circuit_generic]{2(c)}).
The composition of the two fragments (which, for a circuit with no open wires, corresponds to the joint probability of the measurement outcomes of the different operations in the circuit~\cite{hardy09,chiribella10}, see Fig.~\hyperref[fig:process_matrices]{1(a)}) then indeed remains the same in the old and new descriptions.
This follows from the properties of the link product (see Methods, Eqs.~\eqref{eq:unitary_canceling_pure} and~\eqref{eq:unitary_canceling_mixed}), which provides a formal tool to connect the different fragments that a quantum circuit is decomposed into~\cite{chiribella08,chiribella09}.  

Importantly, the structure of a given circuit with respect to such a choice of time-delocalised subsystems can also be tested operationally~\cite{oreshkov18}. In particular, the circuit can be disconnected at the chosen subsystems and each of the time-delocalised operations that occur on these subsystems can be experimentally addressed and verified, similarly to the way one would test the circuit description with respect to the standard time-local factorisation. 
In this sense, such an alternative description of the experiment is operationally just as meaningful. 
This is discussed in more detail in Supplementary Note~\hyperref[app:testing]{1}.\\ 

\textbf{Processes with indefinite causal order on time-delocalised subsystems} \label{sec:td_processes} \quad In the laboratory experiments that have been proposed as implementations of the quantum switch, one considers a \emph{target} quantum system at two possible times. 
The operation $U_A$ is applied to the target system $T_1$ at the earlier time, or to the target system $T_2$ at the later time, depending on  whether another two-dimensional quantum system, the \emph{control qubit}, is in the computational basis state $\ket{0}$ or $\ket{1}$, and conversely for the operation $U_B$. 
There has been much debate (see e.g. Refs.~\cite{maclean17,oreshkov18,vilasini22,ormrod22}) about whether experiments of that type can be interpreted as valid realisations of the quantum switch, understood as an abstractly defined scenario in the process matrix formalism~\cite{araujo15}.
Indeed, the relation between the above outlined experimental procedure, and the situation considered in the process matrix framework, where one instance of each $U_A$ and $U_B$ is composed with the process matrix in a circuit with a cycle, is a priori unclear.
A heuristic argument that is sometimes invoked to justify that each of the two operations is indeed applied once and only once is that each operation occurs precisely once in each of the two superposed coherent branches, and is therefore used once overall. To further corroborate this, one could introduce a \emph{flag} or \emph{counter} system~\cite{araujo14,purves21} that keeps track of the usage of the operations.
To really understand the sense in which the quantum switch is realised in these experiments, it is however desirable to rigorously formalise the link between the standard quantum description of the experiments, and the process matrix scenario.
This question was addressed in Ref.~\cite{oreshkov18}. It was shown that the temporally ordered quantum circuit that describes the experimental situation outlined above indeed takes the structure of a circuit with a cycle as in the process matrix framework (i.e., as in Fig.~\ref{fig:process_matrices}), when one changes to a description in terms of specific time-delocalised subsystems---whose choice, broadly speaking, formalises the intuition that the input system is $T_1$ when the control system is in state $\ket{0}$ and $T_2$ when the control system is in state $\ket{1}$, and similarly for the output systems~\cite{oreshkov18}.
In other words, when these experiments are realised physically, what happens on these alternative systems is precisely the structure described in the process matrix framework. 
It is in that sense that these experiments can be considered \emph{realisations} of the abstract mathematical concept.

It was then shown that this argument can be generalised, and that there exist other types of processes which have a realisation in this sense. 
Notably, this is the case for the entire class of unitary extensions of bipartite processes, of which the quantum switch is a particular example.
It was subsequently shown in Refs.~\cite{barrett20,yokojima20} that all such processes are variations of the quantum switch, but the proof of Ref.~\cite{oreshkov18} did not rely on this knowledge. 
It is the idea behind the original proof from Ref.~\cite{oreshkov18}, together with the subsequent result of Refs.~\cite{barrett20,yokojima20}, that will allow us to generalise the proof to the tripartite case.
We therefore recall the bipartite result from Ref.~\cite{oreshkov18}, in the language and conventions we use in this paper (notably employing the Choi representation and the link product), in \hyperref[app:bipartite_methods]{Methods}, and the corresponding proofs in Supplementary Note~\hyperref[app:bipartite_proof]{2}.\\

\textbf{Unitary extensions of tripartite processes on time-delocalised subsystems} \label{sec:td_tripartite} \quad For unitary extensions of processes with more than two parties, it is a priori unclear whether and how a realisation on time-delocalised subsystems can be found.
In the following, we will establish the result for unitary extensions of tripartite processes.
Briefly summarised, we show that for any unitarily extended tripartite process, there exists a standard, temporally ordered quantum circuit,
with operations that depend on the local operations $U_A$, $U_B$ and $U_C$ applied in the process, which precisely corresponds to the situation considered in the process matrix framework, with one instance of each $U_A$, $U_B$ and $U_C$ composed with the process matrix in a circuit with a cycle, when described in terms of a suitable choice of time-delocalised subsystems.

Formally, we prove the following proposition.\\

\textit{Proposition 1.} \ \ Consider a unitary extension of a tripartite process, described by a process vector $\dket{U} \in \HS^{P_O A_{IO}B_{IO} C_{IO} F_I}$, composed with unitary local operations $U_A:\HS^{A_I A_I'} \to \HS^{A_O A_O'}$, $U_B:\HS^{B_I B_I'} \to \HS^{B_O B_O'}$ and $U_C:\HS^{C_I C_I'} \to \HS^{C_O C_O'}$.  
For any such process, the following exist.
\begin{enumerate}
\item A temporal circuit of the form shown in Fig.~\ref{fig:tripartite_implementation_main}, in which $U_A$ and $U_B$ are applied on the target input and output systems $T_1^{(\prime)}$ or $T_2^{(\prime)}$, coherently conditioned on the state of the control systems $Q_1^{(\prime)}$ and $Q_2^{(\prime)}$, and which is composed of circuit operations that depend on $U_C$.
\item Isomorphisms $J_{\textup{in}}: \HS^{A_I B_I C_I Y Z} \to \HS^{T_1 T_2 \bar T_1' \bar T_2' Q_1 P_O}$ and $J_{\textup{out}}: \HS^{ T_1' T_2' \bar T_1 \bar T_2 Q_2' F_I} \to \HS^{A_O B_O C_O \bar Y \bar Z}$, such that, with respect to the subsystems $A_I$, $B_I$ and $C_I$ of $T_1 T_2 \bar T_1' \bar T_2' Q_1 P_O$ and the subsystems $A_O$, $B_O$ and $C_O$ of $T_1' T_2' \bar T_1 \bar T_2 Q_2' F_I$ that these isomorphisms define, the circuit in Fig.~\ref{fig:tripartite_implementation_main} takes the form of a cyclic circuit composed of $U$, $U_A$, $U_B$ and $U_C$\, as in the process matrix framework (see Fig.~\ref{fig:tripartite_proof}).
\end{enumerate}

\begin{figure}[h]
    \centering
    \includegraphics[width=0.8\textwidth]{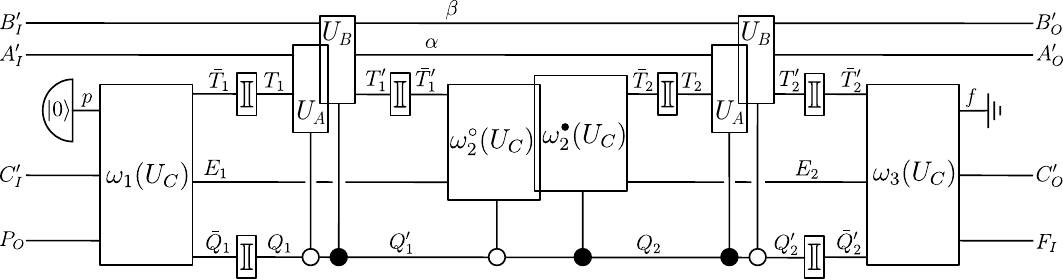}
    \caption{\textbf{Temporal circuit for a general tripartite unitary process.} 
$U_A$ and $U_B$ are applied either on the time-local target system $T_1^{(\prime)}$ or $T_2^{(\prime)}$ (and the ancillary systems), depending coherently on the state of the two-dimensional control systems $Q_1^{(\prime)}$ and $Q_2^{(\prime)}$.
These two applications of the coherently controlled operations $U_A$ and $U_B$ are surrounded by circuit operations $\omega_1(U_C): \HS^{p C_I' P_O} \to \HS^{\bar{T}_1 E_1 \bar Q_1}$, $\omega_2^\circ(U_C): \HS^{\bar T_1' E_1} \to \HS^{\bar T_2 E_2}$, $\omega_2^\bullet(U_C): \HS^{\bar T_1' E_1} \to \HS^{\bar T_2 E_2}$ (these two also being coherently controlled), and $\omega_3(U_C): \HS^{\bar T_2' E_2 \bar Q_2'} \to \HS^{f C_O' F_I}$, which can (together with the therein introduced ancillary systems $E_1, E_2$) in general all depend on $U_C$, the third party's (Charlie's) operation. 
The boxes $\mathbb{I}$ stand for identity channels that relate the systems with and without the bars.
The ancillary system $p$ is prepared in the state $\ket{0}^p$ in the beginning, and the final ancillary system $f$ is discarded. (Note that, with a slight abuse of notation, we use the ground symbol for this discarding of $f$, which is commonly used for mixed circuits where the boxes represent CP maps, rather than for circuits consisting of pure operations, as we have here. The system $f$ however always ends up in the state $\ket{0}^f$ (see Supplementary Note~\hyperref[app:tripartite_circuit]{3~A}), so that taking the partial trace over $f$ is equivalent to projecting onto $\ket{0}^f$, and does not introduce any decoherence or loss of purity. 
The coherently controlled applications of $U_A$ and $U_B$, as well as of $\omega_2^\circ(U_C)$ and $\omega_2^\bullet(U_C)$, are displayed with a slight shift for graphical clarity, but they can be taken to act at the same time.}
    \label{fig:tripartite_implementation_main}
\end{figure}

In the following, we outline the proof.
All technical proofs and calculations for this tripartite construction are given in Supplementary Note~\hyperref[app:tripartite_proof]{3}.\\

\textit{Outline of proof.} 
The existence of a temporal circuit as in Fig.~\ref{fig:tripartite_implementation_main} is shown in Supplementary Note~\hyperref[app:tripartite_circuit]{3~A}.
It follows from the result that all unitary extensions of bipartite processes can be implemented as \emph{variations of the quantum switch}~\cite{barrett20,yokojima20}, in which the time of the two local operations is controlled coherently.
Any unitary extension of a tripartite process can thus be implemented as a variation of the quantum switch, with two local operations whose time is controlled coherently, and which is composed of circuit operations that depend on the third local operation. 
The isomorphisms $J_{\textup{in}}: \HS^{A_I B_I C_I Y Z} \to \HS^{T_1 T_2 \bar T_1' \bar T_2' Q_1 P_O}$ and $J_{\textup{out}}: \HS^{ T_1' T_2' \bar T_1 \bar T_2 Q_2' F_I} \to \HS^{A_O B_O C_O \bar Y \bar Z}$ (where $Y, Z, \bar Y$ and $\bar Z$ are appropriate complementary subsystems) are defined in Supplementary Note~\hyperref[app:iso_tripartite]{3~B}, based on a specific decomposition of unitarily extended process vectors which plays a central role in the bipartite proof (cf. Eq.~\eqref{eq:decomp_bipartite}), and which generalises to the multipartite case (cf. Eq.~\eqref{eq:decomp_tripartite}).

In Supplementary Note~\hyperref[app:tripartite_systems]{3~C}, we change to the description of the circuit in terms of the corresponding time-delocalised subsystems. 
For that purpose, we decompose the circuit into the red and blue circuit fragment shown in Fig.~\ref{fig:tripartite_proof}. By construction, when composed with $J_\text{in}$ and $J_\text{out}$, the red circuit fragment shown in Fig.~\hyperref[fig:tripartite_proof]{4(a)} consists of precisely one application of $U_A$ and $U_B$, in parallel to a unitary operation $R(U_C): \HS^{C_I' C_I Y Z \bar Q_2'} \to \HS^{C_O' C_O \bar Y \bar Z \bar Q_1}$.
Under that change of subsystems, the complementary blue fragment needs to be composed with the inverse isomorphisms $J_\text{in}^\dagger$ and $J_\text{out}^\dagger$, which results in an operation $R': \HS^{P_O A_O B_O C_O \bar Y \bar Z \bar Q_1} \to \HS^{F_I A_I B_I C_I Y Z \bar Q_2'}$ (see Fig.~\hyperref[fig:tripartite_proof]{4(b)}).
$R(U_C)$ and $R'$ cannot be further decomposed for now.

At this point, we thus have a cyclic circuit which consists of the four boxes $U_A$, $U_B$, $R(U_C)$ and $R'$, and which involves the systems $P_O$, $A_{IO}^{(\prime)}$, $B_{IO}^{(\prime)}$, $C_{IO}^{(\prime)}$, $F_I$, as well as $Y$, $\bar Y$, $Z$, $\bar Z$, $\bar Q_1$, $\bar Q_2'$ (see the left-hand side of Fig.~\hyperref[fig:tripartite_proof]{4(c)}).
In order to obtain a description with respect to only the systems $P_O$, $A_{IO}^{(\prime)}$, $B_{IO}^{(\prime)}$, $C_{IO}^{(\prime)}$, $F_I$, we need to evaluate the composition of $R(U_C)$ and $R'$ over the systems $Y$, $\bar Y$, $Z$, $\bar Z$, $\bar Q_1$, $\bar Q_2'$ (but not over the systems $C_I$ and $C_O$, which we wish to maintain in the description).
The isomorphisms $J_{\textup{in}}$ and $J_{\textup{out}}$ are constructed in precisely such a way (based on the abstract relation between the systems in the process that is also used in the bipartite proof) that, when this composition of $R(U_C)$ and $R'$ over $Y$, $\bar Y$, $Z$, $\bar Z$, $\bar Q_1$, $\bar Q_2'$ is evaluated, the result is a cyclic circuit fragment consisting of the unitary operation $U: \HS^{P_O A_O B_O C_O} \to \HS^{F_I A_I B_I C_I}$ that defines the process, composed with the operation $U_C: \HS^{C_I C_I'} \to \HS^{C_O C_O'}$ (see the middle of Fig.~\hyperref[fig:tripartite_proof]{4(c)}). 
(Note the particularity that $U_C$ only appears as an explicit part of the cyclic circuit after this composition of $R(U_C)$ with $R'$, and is not a tensor product factor of $R(U_C)$).

Therefore, in its description with respect to the systems $P_O$, $A_{IO}^{(\prime)}$, $B_{IO}^{(\prime)}$, $C_{IO}^{(\prime)}$, $F_I$, the circuit in Fig.~\ref{fig:tripartite_implementation_main} indeed consists of the four operations $U_A: \HS^{A_I A_I'} \to \HS^{A_O A_O'}$, $U_B: \HS^{B_I B_I'} \to \HS^{B_O B_O'}$, $U_C: \HS^{C_I C_I'} \to \HS^{C_O C_O'}$ and $U: \HS^{P_O A_O B_O C_O} \to \HS^{F_I A_I B_I C_I}$, connected in a cyclic circuit as in the process matrix framework (see the right-hand side of Fig.~\hyperref[fig:tripartite_proof]{4(c)}).
This establishes the tripartite result.

Note that a similar construction is possible when one considers an asymmetric tripartite temporal circuit where $U_A$ is applied at a given, well-defined time, and $U_B$ either before or after it, coherently depending on the control systems (or vice versa, with the roles of $A$ and $B$ exchanged).

\begin{figure}[h]
   \includegraphics[width=0.9\textwidth]{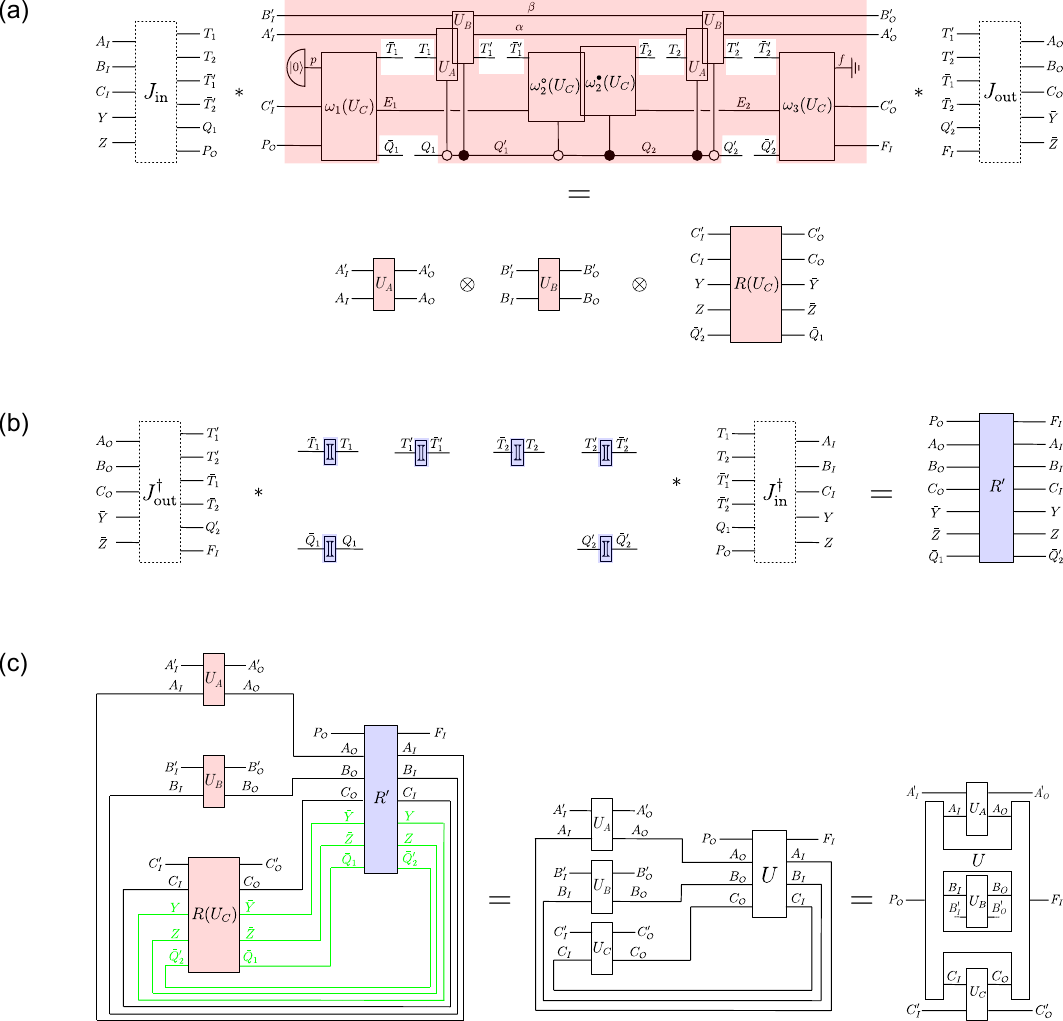}
    \caption{\textbf{Description of the tripartite temporal circuit in terms of time-delocalised subsystems.}\\ (a) Description of the red circuit fragment in terms of the time-delocalised subsystems $A_I$, $B_I$, $C_I$, $Y$, $Z$ of the joint system $T_1 T_2 \bar T_1' \bar T_2' Q_1 P_O$, and $A_O$, $B_O$, $C_O$, $\bar Y$, $\bar Z$ of the joint system $T_1' T_2' \bar T_1 \bar T_2  Q_2' F_I$. \\ 
(b) Description of the blue circuit fragment in terms of the time-delocalised subsystems $A_I$, $B_I$, $C_I$, $Y$, $Z$ of the joint system $T_1 T_2 \bar T_1' \bar T_2' Q_1 P_O$, and $A_O$, $B_O$, $C_O$, $\bar Y$, $\bar Z$ of the joint system $T_1' T_2' \bar T_1 \bar T_2  Q_2' F_I$.\\
(c) The composition of the operations $R(U_C)$ and $R'$ over the systems $Y$, $\bar Y$, $Z$, $\bar Z$, $\bar Q_1$, $\bar Q_2'$ gives rise to a cyclic circuit fragment consisting of the operation $U_C$ and the unitary $U$ that defines the process. That is, when evaluating the composition of $R(U_C)$ and $R$ over the wires shown in green (but not over $C_I$ and $C_O$), we obtain the cyclic circuit in the middle, consisting of the operations $U_A$, $U_B$, $U_C$ and $U$. With respect to the systems $P_O$, $A_{IO}'$, $B_{IO}'$, $C_{IO}'$, $F_I$, as well as the time-delocalised systems $A_{IO}$, $B_{IO}$, $C_{IO}$, the circuit therefore consists of $U_A$, $U_B$, $U_C$ and $U$, composed in a cyclic manner as in the process matrix framework.}
    \label{fig:tripartite_proof}
\end{figure}

\clearpage

\textbf{A process that violates causal inequalities on time-delocalised subsystems} \label{sec:td_bw} \quad In Ref.~\cite{baumeler14a}, it was shown that, for three and more parties, there exist process matrices that violate causal inequalities and that can be interpreted as \emph{classical} process matrices, since they are diagonal in the computational basis.  
An example, first found by Ara\'ujo and Feix and further studied by Baumeler and Wolf in Refs.~\cite{baumeler16,baumeler17}, is the process matrix
\begin{align}
\label{eq:pm_bw_trivialPF}
W_{\text{AF}} &= \sum_{a_O b_O c_O} \ket{\neg b_O \!\land\! c_O,\neg c_O \!\land\! a_O,\neg a_O \!\land\! b_O}\bra{\neg b_O \!\land\! c_O,\neg c_O \!\land\! a_O,\neg a_O \!\land\! b_O}^{A_I B_I C_I} \otimes \ket{a_O,b_O,c_O}\bra{a_O,b_O,c_O}^{A_OB_OC_O},
\end{align}
where $a_O,b_O,c_O \in \{0,1\}$ and where $\neg$ is the negation.
It was then shown by Baumeler and Wolf~\cite{baumeler17} (cf. also Refs.~\cite{araujo17,araujo17a}) that $W_{\text{AF}}$ has a unitary extension $W_{\text{BW}} = \dketbra{U_{\text{BW}}}$, with
\begin{align}
\label{eq:pm_bw}
\dket{U_{\text{BW}}} &= \sum_{\substack{a_O b_O c_O\\p_1 p_2 p_3}} \ket{p_1,p_2,p_3}^{P_1 P_2 P_3} \otimes \ket{p_1 \oplus \neg b_O \land c_O, p_2 \oplus \neg c_O \land a_O, p_3 \oplus \neg a_O \land b_O}^{A_I B_I C_I} \notag \\[-4mm] &\hspace{75mm}\otimes \ket{a_O,b_O,c_O}^{A_OB_OC_O} \otimes \ket{a_O,b_O,c_O}^{F_1 F_2 F_3} 
\end{align}
(with $p_1,p_2,p_3 \in \{0,1\}$, i.e., $\HS^{P_O} = \HS^{P_1 P_2 P_3}$ and $\HS^{F_I} = \HS^{F_1 F_2 F_3}$ consisting of three qubits each, and with $\oplus$ denoting addition modulo 2). 
$W_{\text{AF}}$ is recovered from $\dketbra{U_{\text{BW}}}$ when the global past party prepares the state $\ketbra{0,0,0}{0,0,0}^{P_1 P_2 P_3}$, and the global future party is traced out.
What kind of temporal circuit do we obtain when we apply the general tripartite considerations from the previous section to this particular example? 
A possible such realisation of this process on time-delocalised subsystems is given by the circuit shown in Fig.~\ref{fig:tripartite_example} (similar circuits corresponding to this process have also been studied in other contexts in Refs.~\cite{guerin18b,araujo17a,baumann21}).
 \begin{figure}[h]
     \centering
     \includegraphics[width=0.7\textwidth]{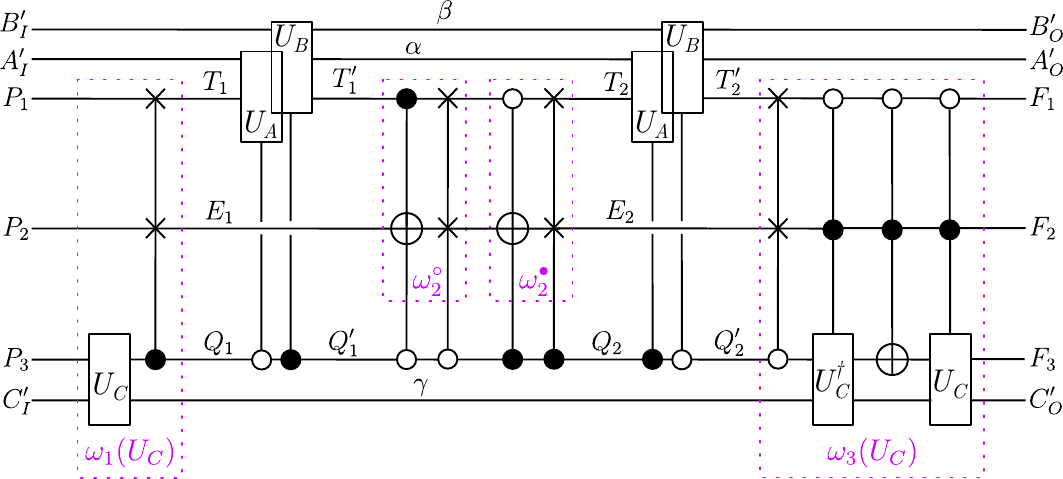}
     \caption{\textbf{Realisation on time-delocalised subsystems of} $\dket{U_{\text{BW}}}$. 
Note that for this particular process, the two circuit operations $\omega_2^\circ$ and $\omega_2^\bullet$ do not depend on $U_C$.
For simplicity of the representation, the identity channels $\id^{\bar T_1 \to T_1}$, $\id^{T_1' \to \bar T_1'}$, $\id^{\bar T_2 \to T_2}$, $\id^{T_2' \to \bar T_2'}$, $\id^{\bar Q_1 \to Q_1}$ and $\id^{Q_2' \to \bar Q_2'}$ that constitute the blue circuit fragment in Fig.~\ref{fig:tripartite_proof} are omitted in the figure here.
Note that, with respect to the general tripartite circuit of Fig.~\ref{fig:tripartite_implementation_main}, we can make a few simplifications for this particular process.
In order to match the general form, the ancilla $\gamma$ would need to be incorporated into the circuit ancillas $E_1$, and $E_2$. But since it is just transmitted identically from $\omega_1(U_C)$ to $\omega_3(U_C)$, we may keep it as a separate wire.
We can also omit the additional systems $p$ and $f$, which we introduce in Supplementary Note~\hyperref[app:tripartite_circuit]{3~A} in order to derive an alternative temporal circuit (Fig.~\ref{fig:bipartite_symm}) for general unitarily extended bipartite processes (and from which we then obtain the circuit of Fig.~\ref{fig:tripartite_implementation_main} for general unitarily extended tripartite processes).
The four circuit operations can be further broken down into several temporal steps, as shown within the purple boxes.
This allows one to get a descriptive understanding of how the time-delocalisation of Charlie's operation comes about in this realisation.
Namely, a time-local instance of $U_C$ is applied once as part of the first circuit operation, and determines the state of the control systems that determine coherently whether $U_A$ is applied on $T_1^{(\prime)}$ and $U_B$ on $T_2^{(\prime)}$ or vice versa (i.e., their order). 
After they have both been applied, a reversal and reapplication of $U_C$ may occur, with a NOT gate in between, and whether this happens or not is determined jointly (and coherently, again) by the operations of Alice and Bob. 
However, we emphasize again that the occurence of several time-local operations that depend on $U_C$ should not be interpreted as $U_C$ being applied multiple times. 
Instead, just like $U_A$ and $U_B$, it is applied once and only once, on time-delocalised input and output systems.}
     \label{fig:tripartite_example}
 \end{figure}

In Supplementary Note~\hyperref[app:tripartite_example_quantum]{4~A}, we give the explicit expressions of the circuit operations in Fig.~\ref{fig:tripartite_example}, as well as for the isomorphisms that define the description in terms of time-delocalised subsystems for this particular case, and we sketch how to apply the general tripartite proof to this example.

The abstract process $W_{\text{AF}}$ in Eq.~\eqref{eq:pm_bw_trivialPF} violates causal inequalities when each party performs a computational basis measurement on its incoming Hilbert space (and outputs the measurement result $o_X$), and prepares the computational basis state $\ket{i_X}$ (corresponding to its classical input $i_X$) on its outgoing Hilbert space.
The corresponding unitary operations that need to be applied in the pure description of the process (and therefore in the circuit of Fig.~\ref{fig:tripartite_example}) are $U_X = \id^{X_I \to X_O'} \otimes \id^{X_I' \to X_O}$, with each incoming ancillary system being prepared in the state $\ket{i_X}^{X_I'}$ and the outgoing ancillary systems being measured in the computational bases. 
One obtains the deterministic correlation $P(o_A,o_B,o_C|i_A,i_B,i_C) = \delta_{o_A,\neg i_B \land i_C}\delta_{o_B,\neg i_C \land i_A}\delta_{o_C,\neg i_A \land i_B}$, which was shown to violate causal inequalities in Ref.~\cite{baumeler16}.

An example of a causal inequality that is violated by this correlation is 
\begin{align}
    &P(0,0,0|0,0,1) + P(0,0,1|0,0,1) + P(0,0,0|1,0,0) + P(1,0,0|1,0,0)\notag \\ &\hspace{35mm}+ P(0,0,0|0,1,0) + P(0,1,0|0,1,0) - P(0,0,0|0,0,0) \eqqcolon I_1 \ge 0, \label{eq:causal_ineq}
\end{align}
which was derived in Ref.~\cite{abbott16}. 
(It corresponds to Eq.~(26) given there, with $0$ and $1$ exchanged for all inputs and outputs).
Here, we find that $I_1 = -1$.

Interestingly, for that particular process with these particular local operations, all operations involved in the tripartite construction simply take computational basis states to computational basis states. 
These can be understood as deterministic operations between \emph{classical random variables}, rather than unitary operations between quantum systems.
In Supplementary Note~\hyperref[app:tripartite_example_classical]{4~B}, we explain this in more detail. 

All things considered, our main result is thus that there exist classical, deterministic circuits, composed of operations between \emph{time-local variables}, which, when described in terms of suitable \emph{time-delocalised variables}, correspond to classical, deterministic processes that violate causal inequalities. \\

\textbf{Noncausal correlations between time-delocalised variables} \label{sec:correlations_main} \quad After having established that this realisation of a noncausal process exists, we now turn to the question of what we should conclude from the fact that a causal inequality can be violated in such a situation. 
The general reasoning behind causal inequalities is similar to that behind Bell inequalities---one considers certain assumptions which restrict the correlations that can arise from some experiment, and their violation then implies that not all of these assumptions are satisfied. 
To determine whether a causal inequality violation is a meaningful device-independent witness of causal indefiniteness, one must therefore clarify whether the assumptions underlying causal inequalities are plausible or compelling in the setting under consideration---a question that is subtle, notably in regimes of relativistic quantum information and quantum gravity~\cite{ho18,debski22}, but, as it will turn out, also in the standard quantum situations we consider here.
In the following, we will therefore analyse our result in this regard, and argue that causal inequalities are indeed a meaningful concept to show the absence
of a definite causal order between the time-delocalised variables we identified.

In the original approach developed in Ref.~\cite{oreshkov12}, one firstly assumes that the events involved in the experiment take place in a \emph{causal order} (which, in general, can be dynamical and subject to randomness~\cite{oreshkov16,abbott16}). 
With respect to this causal order, there are two further assumptions that enter the derivation of causal inequalities.
Firstly, the classical inputs which the parties receive are subject to \emph{free choice}.
Technically speaking, this means that they cannot be correlated with any properties pertaining to their causal past or elsewhere (see \hyperref[app:methods_CI]{Methods}).
Secondly, the parties operate in \emph{closed laboratories}.
That is, intuitively speaking, they open their laboratory once to let a physical system enter, interact with it and open their laboratory once again to send out a physical system, which provides the sole means of information exchange between the local variables and the rest of the experiment.
More formally, the closed laboratory assumption says that, for each party $X$, any causal influence from the \emph{setting} variable $I_X$, which describes its classical input, to any other variable, except the variable $O_X$ that describes its classical outcome, has to pass through the outgoing variable $X_O$. Similarly, any causal influence to $O_X$ from any other variable except $I_X$ has to pass through $X_I$. Furthermore, $X_I$ is in the causal past of $X_O$ (see \hyperref[app:methods_CI]{Methods}).
In order to clarify whether the violation of a causal inequality discovered here is meaningful and interesting, the question that we need to address is whether one would naturally expect the free choice and closed laboratory assumptions to be satisfied in our scenario with time-delocalised (classical) variables, or whether one of them is manifestly violated.

In the Methods section \hyperref[app:methods_CI]{``Causal inequality assumptions''}, we formulate these assumptions, for the multipartite case, in a way that is suitable for our time-delocalised setting, namely directly in terms of the variables involved (rather than in terms of \emph{events} as in~\cite{oreshkov12}), and show that they indeed imply that causal inequalities must be respected.
Our formulation provides a strengthening of the original derivation in Ref.~\cite{oreshkov12} by relaxing the closed laboratories assumption---rather than imposing that the incoming variable $X_I$ is always in the causal past of the outgoing variable $X_O$, we only require this constraint to hold for at least one particular value of the corresponding setting variable $I_X$ (see \hyperref[app:methods_CI]{Methods}). As we discuss in the following, this formulation of the assumptions is directly motivated by the observable causal relations between the variables of interest. Thus, the violation of a causal inequality in the experiment can be seen as a compelling, device-independent demonstration of the nonexistence of a possibly dynamical and random causal order between the variables. 

The causal relations between the incoming and outgoing variables $X_I$ and $X_O$, as well as the setting and outcome variables $I_X$ and $O_X$, $X = A,B,C$, can be graphically represented by a directed graph as in Fig.~\ref{fig:CSM}, where the arrows describe direct causal influences. 

\begin{figure}[h]
    \centering
    \includegraphics[width=0.6\textwidth]{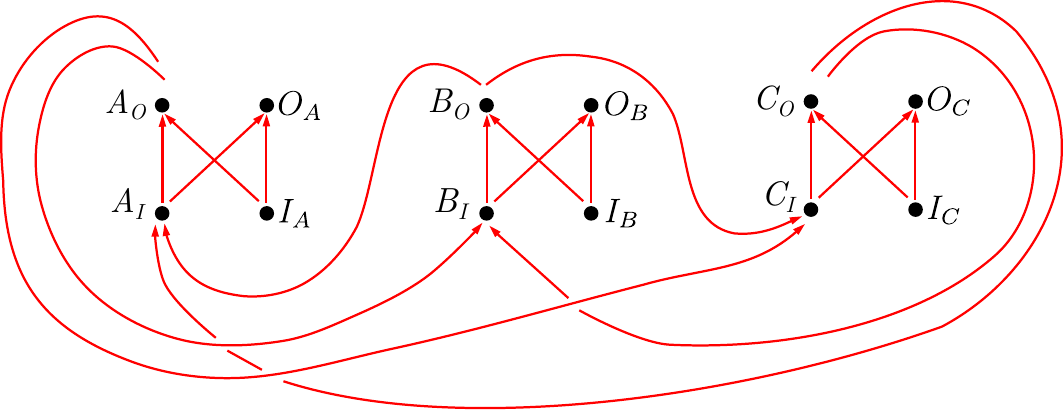}
    \caption{\textbf{Causal structure of the cyclic causal model corresponding to the process} $W_{\text{AF}}$. The causal influences represented by the arrows can be rigorously defined in the framework of \emph{cyclic split-node causal models}~\cite{barrett20} (or, in the case of more general quantum processes, \emph{cyclic quantum causal models}~\cite{barrett20}). 
In particular, if we regard each single variable as a split-node (or in the more general case of quantum processes, a quantum node where the output Hilbert space is the dual of the input Hilbert space), the experiment can be viewed as a process on a larger number of nodes, which is given by the (tensor) product of the original process and the local operations of the parties. 
The causal relations between the new nodes form the cyclic causal structure in Fig.~\ref{fig:CSM}. 
This follows from the known cyclic causal structure of the process $W_{\text{AF}}$~\cite{barrett20} and the most general causal structure that each local operation from $X_I$, $I_X$ to $X_O$, $O_X$, $X = A,B,C$, can have. 
(Here, we are imagining an experiment in which each party could choose over a finite set of local operations that could instantiate all these different causal relations. 
The set of operations over which the party can choose can be embedded within a single deterministic operation with the maximally connected acyclic causal structure displayed in the figure by choosing $I_X$ and $O_X$ of sufficiently large cardinality.)
This causal structure can be operationally verified: by applying time-delocalised SWAP operations on the time-delocalised variables (or quantum systems) so as to disconnect the operations of the parties (see Supplementary Note~\hyperref[app:testing]{1}), one could intervene on the variables (or quantum systems) and verify which ones are directly influenced by which other ones. 
Note that the process $W_{\text{AF}}$ was first studied in the framework of cyclic causal models in Ref.~\cite{barrett20}, but from the perspective of the coarse-grained split nodes defined by the pairs of variables $(X_I, X_O) \equiv X$, for $X=A, B, C$.}
    \label{fig:CSM}
\end{figure}

In the causal structure in Fig.~\ref{fig:CSM}, the variables $I_X$ are \emph{root variables} and hence they can only be correlated with other variables as a result of causal influence from them to these other variables. 
It is thus natural to assume the same would be true if there existed an explanation of the correlations in terms of a definite causal order, which legitimates the free choice assumption.    

Regarding the closed laboratory assumption, in the graph of Fig.~\ref{fig:CSM}, any causal influence from $I_X$ to variables other than $O_X$ and $X_O$ is mediated, or \emph{screened off}, by $X_O$.
Similarly, any influence onto $O_X$ by variables other than $I_X$ and $X_I$ is mediated by $X_I$. 
It is natural to assume that these constraints would also hold in any potential explanation of the correlations in terms of a definite causal order.
Finally, the causal diagram displays causal influence from $X_I$ to $X_O$.
Note that this causal influence from $X_I$ to $X_O$ can be turned on or off depending on the value of the setting variable $I_X$. 
This is precisely the reason why we introduced the weakened form of the closed laboratory assumption described above, which indeed allows for $X_O$ to be inside or outside of the causal future of $X_I$, depending on the value of $I_X$. 

To summarise, we have shown that there is a set of natural assumptions about the possible underlying causal orders between the variables of interest in our experiment, which are directly motivated by the observable causal relations between these variables, and which imply that the correlations in the experiment would need to respect causal inequalities. 
The observable violation of a causal inequality in the experiment thus implies that an underlying causal order compatible with these assumptions cannot exist.

Are there any considerations that would lead us to drop one assumption over another in this type of experiment? 
In particular, could it be that, in spite of the outlined considerations about the observable causal relations, a more careful inspection of the temporal description of the experiment would reveal that it is in fact the free choice or closed laboratory assumptions that is violated, as opposed to the existence of a causal order per se? In the discussion below and in Supplementary Note~\hyperref[app:CI_assumptions]{6}, we analyse this question and argue that if the hypothetical causal order is expected to be imposed by spatiotemporal relations, it is the existence of causal order per se that seems violated, since the variables of interest do not admit an effective localisation in spacetime.

\section*{Discussion}
\label{sec:discussion}

A central question in the study of quantum causality is which processes with indefinite causal order have a realisation within standard quantum theory.
In order to address this question, it is first of all necessary to clarify what it means for a causally indefinite process to have a standard quantum theoretical realisation, a question that is subtle and has led to a lot of controversy. 
An answer to this question is provided by the concept of \emph{time-delocalised subsystems}, which establishes a bridge between the standard quantum theoretical description of the scenarios under consideration and their description in the process matrix framework, in which the notion of indefinite causal order is formalised.
Prior to our work, it had been known that indefinite causal order can be realised on systems that are time-delocalised in a coherently controlled manner---that is, intuitively speaking, the input and output systems of each party effectively reduce to one or another time-local system, conditionally on the state of a control quantum system.
Here, we showed that this paradigm does not encompass all possibilities, and that standard quantum theory also allows for more radical ways to realise indefinite causal order processes.
Notably, there exist processes that have realisations on time-delocalised subsystems and that violate causal inequalities, a feature that is generally believed to be impossible within standard (quantum) physics~\cite{purves21}.
We analysed a concrete tripartite example, for which it turned out that the situation can entirely be understood in terms of classical variables, rather than quantum systems.
There, Alice's and Bob's input and output variables are time-delocalised in a classically controlled way, while the situation for Charlie is quite different.
From the point of view of the temporal description of the experiment, one time-local instance of Charlie's operation is applied in the beginning of the circuit, which may be reversed and reapplied at the end of the circuit, conditionally on the output of Alice and Bob.
We then analysed this causal inequality violation with regard to the assumptions that underlie the derivation of causal inequalities, and found that the free choice and closed laboratory assumptions are not manifestly violated, which makes causal inequalities a meaningful device-independent concept to qualify these realisations as incompatible with a definite causal order.

Let us further elaborate on the subtleties that this analysis involves, in particular with respect to the closed laboratory assumption (see a more detailed discussion in Supplementary Note~\hyperref[app:CI_assumptions]{6}). 
From an intuitive reading of the circuit in Fig.~\ref{fig:tripartite_example}, one may be tempted to say that Charlie acts multiple times or receives several inputs, and sends out several outputs.
At first sight, this seems to violate the closed laboratory assumption, which essentially stipulates that each party is involved in a single round of information exchange, where they receive information about the past through the input variable $X_I$ and subsequently send out information into the future through the output variable $X_O$.
However, it is crucial to realise that the causal inequality assumptions concern concrete variables (or quantum systems), which in our case we have explicitly specified, and which are not the same as what one might intuitively assume if one thinks of this experiment as involving three laboratories existing through time that exchange information with each other. 
In particular, the parties Alice, Bob and Charlie must be understood abstractly as agents who control the parameters that describe the operations taking place on the time-delocalised variables.
As such, they indeed apply their operations once and only once on the pairs of input and output variables we have identified.
To say that the closed laboratory assumption is violated, one would need to come up with an account for the process in terms of variables which are embedded into a causal order, but for which the closed laboratory assumption fails.
We are not aware of any explanation in terms of the time-local variables in the temporal circuit and the causal order defined by their spatiotemporal relations (or any other operationally meaningful variables) where this is the case.
In particular, the above-outlined intuitive reading of the circuit, with the operations being effectively localised in time, conditioned on other variables in the process, while meaningful for quantumly controlled time-delocalised operations, does not make operational sense in our case (as it would mean that some future parties can influence what has happened in the past, see Supplementary Note~\hyperref[app:CI_assumptions]{6}).
In Supplementary Note~\hyperref[app:CI_assumptions]{6}, we show that, for some of the time-delocalised variables we identified, there do not exist time-local variables that take their value, meaning that they do not admit any effective localisation in time. 

The further implications of this finding are yet to be unravelled, and raise various open questions.
In a more general sense, there is a causal explanation for how these correlations in our process come about—namely, precisely the tripartite circuit realisation we found.
This raises the question of whether and how the concept of causal inequalities in itself could be revised or modified. For instance, could there be a notion of \emph{causal process} which is more relaxed, and which includes such possibilities?

What other processes beyond the classes considered here have a realisation on time-delocalised subsystems, and what other types of time-delocalisation would this involve?
Could it be that any indefinite causal order process admits such a realisation, or are there counterexamples?
The proof for unitarily extended tripartite processes is crucially based on the fact that the bipartite unitarily extended process resulting from fixing one of the operations has a particular standard form---namely, a variation of the quantum switch~\cite{barrett20,yokojima20}. 
Establishing whether a similar standard form exists for unitarily extended processes with more than two parties could give insight into whether the constructions presented here can be generalised to more parties.

Note that there are also unitary extensions of bipartite processes---i.e., variations of the quantum switch---that have realisations of the kind considered here, with one of the operations being reversed and reapplied (for instance, one obtains such a realisation when one fixes Alice's or Bob's operation in the circuit of Fig.~\ref{fig:tripartite_example}). 
This raises the question of whether, conversely, the process considered in this work could have an alternative, more intuitive interpretation as a superposition of processes with different definite causal orders in some sense (although it cannot be achieved by direct multipartite generalisations of the quantum switch~\cite{wechs18}). 
The decomposition of this process into a direct sum of causal unitary processes shown in~\cite{barrett20} may offer insights into this question. 
 
Finally, in the way the process framework was originally conceived, the operations performed by the parties were imagined to be local from the point of view of some local notion of time for each party. 
Can we conceive of a notion of a quantum temporal reference frame with respect to which the time-delocalised variables considered here would look local, and what implications would this have for our understanding of the spacetime causal structure in which these experiments are embedded? 
In view of the fact that the example considered here is purely classical, the question arises of which part of a noncausal process is actually related to the quantumness of causal relations. 
On the practical side, an obvious question is whether our finding could unveil new applications.
For instance, could we use such time-delocalised variables for new cryptographic or other information-processing protocols?

\section*{Acknowledgements}
This publication was made possible through the support of the ID\# 61466 grant and ID\# 62312 grant from the John Templeton Foundation, as part of the \href{https://www.templeton.org/grant/the-quantum-information-structure-of-spacetime-qiss-second-phase}{``The Quantum Information Structure of Spacetime'' Project (QISS)}. The opinions expressed in this project/publication are those of the author(s) and do not necessarily reflect the views of the John Templeton Foundation. This work was supported by the Program of Concerted Research Actions (ARC) of the Universit\'{e} libre de Bruxelles and by the French National Research Agency through its \emph{``Investissements d'avenir''} (ANR-15-IDEX-02) program and the ANR-22-CE47-0012 project. J. W. is supported by the Chargé de Recherche fellowship of the Fonds de la Recherche Scientifique FNRS
(F.R.S.-FNRS). O. O. is a Research Associate of the Fonds de la Recherche Scientifique (F.R.S.–FNRS). Published with the support of the University Foundation of Belgium.

\section*{Methods}

\textbf{The Choi isomorphism and the link product} \quad \label{app:choi} The Choi isomorphism~\cite{choi75} is a convenient way to represent linear maps between vector spaces as vectors themselves, and linear maps between spaces of operators as operators themselves. In order to define it, we choose for each Hilbert space $\HS^Y$ a fixed orthonormal, so-called \emph{computational} basis $\{\ket{i}^Y\}_i$. For a Hilbert space $\HS^{YZ} = \HS^Y \otimes \HS^Z$, with computational bases $\{\ket{i}^Y\}_i$ of $\HS^Y$ and $\{\ket{j}^Z\}_j$ of $\HS^Z$, respectively, the computational basis is taken to be $\{\ket{i,j}^{YZ} \coloneqq \ket{i}^Y \otimes \ket{j}^Z\}_{i,j}$. 
We then define the \emph{pure Choi representation} of a linear operator $V: \HS^Y \to \HS^Z$ as
\begin{align}\label{eq:pure_cj}
\dket{V} \coloneqq \id \otimes V \dket{\id}^{YY} = \sum_i \ket{i}^Y \otimes V\ket{i}^Y \ \in \HS^Y \otimes \HS^Z,
\end{align}
with $\dket{\id}^{YY} \coloneqq \sum_i \ket{i}^Y \otimes \ket{i}^Y \in \HS^Y \otimes \HS^Y$.
Similarly, we define the \emph{(mixed) Choi representation} of a linear map $\M: \L(\HS^{Y}) \to \L(\HS^{Z})$ as
\begin{align}
M & \coloneqq \, \big({\cal I}^Y \otimes \M \big) \big( \dketbra{\id}^{YY} \big)  = \, \sum_{i,i'} \ketbra{i}{i'}^Y \otimes \M\big( \ketbra{i}{i'}^Y \big) \quad \in \L\big(\HS^{YZ}\big) \label{eq:def_Choi_matrix}
\end{align}
where ${\cal I}^Y$ denotes the identity map on $\L(\HS^Y)$.

The \emph{link product}~\cite{chiribella08,chiribella09} is a tool which allows one to compute the Choi representation of a composition of maps in terms of the Choi representation of the individual maps.
Consider two tensor product Hilbert spaces $\HS^{XY} = \HS^X \otimes \HS^Y$ and $\HS^{YZ} = \HS^Y \otimes \HS^Z$ which share the same (possibly trivial) space factor $\HS^Y$, and with non-overlapping $\HS^X, \HS^Z$. 
The link product of any two vectors $\ket{a} \in \HS^{XY}$ and $\ket{b} \in \HS^{YZ}$ is defined (with respect to the computational basis $\{\ket{i}^Y\}_i$ of $\HS^Y$) as~\cite{wechs18}
\begin{align}
\ket{a} * \ket{b} \coloneqq \, &\big( \id^{XZ} \otimes \dbra{\id}^{YY} \big) (\ket{a} \otimes \ket{b}) = \, \sum_i \ket{a_i}^X \otimes \ket{b_i}^Z \quad \in \HS^{XZ} \label{eq:def_pure_link_product}
\end{align}
with $\ket{a_i}^X \coloneqq (\id^X \otimes \bra{i}^Y) \ket{a} \in \HS^X$ and $\ket{b_i}^Z \coloneqq (\bra{i}^Y \otimes \id^Z) \ket{b} \in \HS^Z$. 
Similarly, the link product of any two operators $A \in \L(\HS^{XY})$ and $B \in \L(\HS^{YZ})$ is defined as~\cite{chiribella08,chiribella09}%
\begin{align}
A * B &\coloneqq  \, \big( \id^{XZ} \otimes \dbra{\id}^{YY} \big) (A \otimes B) \big( \id^{XZ} \otimes \dket{\id}^{YY} \big) =  \, \sum_{ii'} A_{ii'}^X \otimes B_{ii'}^Z \quad \in \L\big(\HS^{XZ}\big) \label{eq:def_mixed_link_product}
\end{align}
with $A_{ii'}^X \coloneqq (\id^X \otimes \bra{i}^Y) A (\id^X \otimes \ket{i'}^Y) \in \L(\HS^X)$ and $B_{ii'}^Z \coloneqq (\bra{i}^Y \otimes \id^Z) A (\ket{i'}^Y \otimes \id^Z) \in \L(\HS^Z)$. 

The link products thus defined are commutative (up to a re-ordering of the tensor products), and associative provided that each constituent Hilbert space appears at most twice~\cite{chiribella09,wechs18}. 
For $\ket{a} \in \HS^X$ and $\ket{b} \in \HS^Z$, or $A \in \L(\HS^X)$ and $B \in \L(\HS^Z)$ in distinct, non-overlapping spaces, they reduce to tensor products ($\ket{a} * \ket{b} = \ket{a} \otimes \ket{b}$ or $A * B = A \otimes B$). 
For $\ket{a}, \ket{b} \in \HS^Y$, or $A , B \in \L(\HS^Y)$ in the same spaces, they reduce to scalar products ($\ket{a} * \ket{b} = \sum_i \braket{i}{a} \braket{i}{b} = \ket{a}^T \ket{b}$ or $A * B = \Tr[A^T B]$). 

For two linear operators $V_1: \HS^{X} \to \HS^{X'Y}$ and $V_2: \HS^{YZ} \to \HS^{Z'}$, the pure Choi representation of the composition $V \coloneqq (\id^{X'} \otimes V_2)(V_1 \otimes \id^Z): \HS^{XZ} \to \HS^{X'Z'}$ is obtained, in terms of the pure Choi representations $\dket{V_1} \in \HS^{XX'Y}$ and $\dket{V_2} \in \HS^{YZZ'}$ of the individual operators $V_1$ and $V_2$, as
\begin{align}
\dket{V} = \dket{V_1} * \dket{V_2} \quad \in \HS^{XX'ZZ'}. \label{eq:compose_link_prod_pure}
\end{align}
Similarly, for two linear maps $\M_1: \L(\HS^{X}) \to \L(\HS^{X'Y})$ and $\M_2: \L(\HS^{YZ}) \to \L(\HS^{Z'})$
the Choi representation of the composition $\M \coloneqq ({\cal I}^{X'} \otimes \M_2) \circ (\M_1 \otimes {\cal I}^Z): \L(\HS^{XZ}) \to \L(\HS^{X'Z'})$ is obtained, in terms of the Choi representations of the individual maps $M_1 \in \L(\HS^{XX'Y})$ and $M_2 \in \L(\HS^{YZZ'})$ of $\M_1$ and $\M_2$, as
\begin{align}
M = M_1 * M_2 \quad \in \L\big(\HS^{XX'ZZ'}\big). \label{eq:compose_link_prod_mixed}
\end{align}

Another property of the link product, which can easily be verified from its definition, is that for any $\ket{a} \in \HS^{XY}$, $\ket{b} \in \HS^{YZ}$ and any unitary $U: \HS^Y \to \HS^{Y'}$, it holds that 
\begin{equation}
\label{eq:unitary_canceling_pure}
(\ket{a} * \dket{U}) * (\dket{U^\dagger} * \ket{b}) = \ket{a} * \ket{b}.
\end{equation}
Similarly, for any $A \in \L(\HS^{XY})$, $B \in \L(\HS^{YZ})$ and any unitary $U: \HS^Y \to \HS^{Y'}$, it holds that 
\begin{equation}
\label{eq:unitary_canceling_mixed}
(A * \dket{U}\dbra{U}) * (\dket{U^\dagger}\dbra{U^\dagger} * B) = A * B.
\end{equation}
This is precisely the property we use in the main text when changing the subsystem description of a circuit.
Namely, it is due to this property that the overall composition of two circuit fragments remains the same when we compose one fragment with certain isomorphisms (i.e., unitary transformations) defining new subsystems, and the complementary fragment with the inverses of these isomorphisms.\\

\textbf{Unitary extensions of bipartite processes on time-delocalised subsystems} \label{app:bipartite_methods} \quad
In summary, the bipartite result says that for any unitarily extended bipartite process, there exists a temporally ordered quantum circuit, with operations that depend on the local operations $U_A$ and $U_B$ applied in the process, which precisely corresponds to the situation considered in the process matrix framework, with one instance of each $U_A$ and $U_B$ composed with the process matrix in a cyclic circuit, when described in terms of a suitable choice of time-delocalised subsystems.

Formally, the bipartite result can be stated as follows.\\

\textit{Proposition 2.} \ \ Consider a unitary extension of a bipartite process, described by a process vector $\dket{U} \in \HS^{P_O A_{IO}B_{IO} F_I}$, composed with unitary local operations $U_A:\HS^{A_I A_I'} \to \HS^{A_O A_O'}$ and $U_B:\HS^{B_I B_I'} \to \HS^{B_O B_O'}$. 
For any such process, the following exist.
\begin{enumerate}
\item A temporal circuit as in Fig.~\ref{fig:bipartite_implementation_main}, in which $U_A$ is applied on some systems $A_I$ and $A_O$ at a definite time, preceded and succeded respectively by two unitary circuit operations $\omega_1(U_B): \HS^{B_I' P_O} \to \HS^{A_I E}$ and $\omega_2(U_B): \HS^{A_O E} \to \HS^{B_O' F_I}$ that depend on $U_B$.
\item Isomorphisms $J_\text{in}: \HS^{B_I Z} \to \HS^{A_O P_O}$ and $J_\text{out}: \HS^{A_I F_I} \to \HS^{B_O \bar Z}$, such that, with respect to the subsystem $B_I$ of $A_OP_O$ and the subsystem $B_O$ of $A_IF_I$ that these isomorphisms define, the circuit in Fig.~\ref{fig:bipartite_implementation_main} takes the form of a cyclic circuit composed of $U$, $U_A$ and $U_B$, as in the process matrix framework (see Fig.~\ref{fig:bipartite_proof}). 
\end{enumerate}

\begin{figure}[h]
\centering
\includegraphics[width=0.5\textwidth]{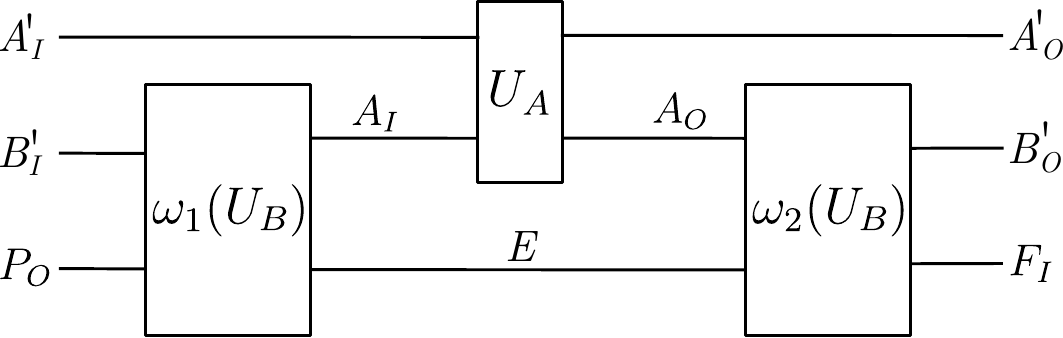}
\caption{\textbf{Temporal circuit in the bipartite case.} Temporal circuit for a bipartite unitary process, with 
$U_A$ being applied on time-local systems $A_I$ and $A_O$, and with circuit operations  $\omega_1(U_B): \HS^{B_I' P_O} \to \HS^{A_I E}$ and $\omega_2(U_B): \HS^{A_O E} \to \HS^{B_O' F_I}$ that depend on $U_B$, and that are connected by an ancillary system $E$.}
\label{fig:bipartite_implementation_main}
\end{figure}

Here, we outline the main points of the proof. 
All technical details and calculations are given in Supplementary Note~\hyperref[app:bipartite_proof]{2}.

\textit{Outline of proof.} 
The existence of a temporal circuit with the form of Fig.~\ref{fig:bipartite_implementation_main} is shown in Supplementary Note~\hyperref[sec:bipartite_circuit]{2~A}. 
It follows from the fact that any unitary extension of a one-party process can be implemented as a fixed-order circuit or \emph{quantum comb}~\cite{chiribella08,chiribella09}, in which the party applies its operation at a definite time.
For a unitary extension of a bipartite process, one can therefore find a fixed-order circuit in which one of the parties acts at a definite time, and which is composed of circuit operations that depend on the operation of the other party.

In Supplementary Note~\hyperref[app:iso_bipartite]{2~B}, we show that the unitary $U$ which defines the process isomorphically maps some subsystem of $A_O P_O$ to $B_I$, and $B_O$ to some subsystem of $A_I F_I$.
The corresponding isomorphisms $J_\text{in}: \HS^{B_I Z} \to \HS^{A_O P_O}$ and $J_\text{out}: \HS^{A_I F_I} \to \HS^{B_O \bar Z}$ (where $Z$ and $\bar Z$ are appropriate complementary subsystems) can be taken to define an alternative description of the circuit in Fig.~\ref{fig:bipartite_implementation_main} in terms of time-delocalised subsystems, since there, $P_O$, $A_I$, $A_O$ and $F_I$ are time-local wires.

In Supplementary Note~\hyperref[sec:bipartite_subsystems]{2~C}, we change to the description of the circuit in terms of these time-delocalised subsystems.
For that purpose, we decompose the circuit into the red and blue circuit fragment shown in Fig.~\ref{fig:bipartite_proof}. 
By construction, when composed with $J_\text{in}$ and $J_\text{out}$, the red fragment consists of precisely one application of $U_B: \HS^{B_I B_I'} \to \HS^{B_O B_O'}$, in parallel to an identity channel from $Z$ to $\bar Z$ (see Fig.~\hyperref[fig:bipartite_proof]{8(a)}).
The blue fragment, which is just the operation $U_A$, needs to be composed with the inverse isomorphisms $J_\text{in}^\dagger$ and $J_\text{out}^\dagger$ so that the overall, global transformation implemented by the circuit remains the same (see Fig.~\hyperref[fig:bipartite_proof]{8(b)}).
In the new description of the circuit of Fig.~\ref{fig:bipartite_implementation_main} in terms of these subsystems, one thus obtains a cyclic circuit as on the left-hand side of Fig.~\hyperref[fig:bipartite_proof]{8(c)}.

The final step is to note that the composition of the inverse isomorphisms $J_{\text{in}}^\dagger$ and $J_{\text{out}}^\dagger$ with the identity channel $\id^{Z \to \bar Z}$ over the systems $Z$ and $\bar Z$ is precisely the unitary operation $U$ that defines the process.
Therefore, in this coarse-grained description with respect to the systems $P_O$, $A_{IO}^{(\prime)}$, $B_{IO}^{(\prime)}$, and $F_I$, the circuit indeed consists of three transformations $U_A: \HS^{A_I A_I'} \to \HS^{A_O A_O'}$, $U_B: \HS^{B_I B_I'} \to \HS^{B_O B_O'}$ and $U: \HS^{P_O A_O B_O} \to \HS^{F_I A_I B_I}$ that are composed in a cyclic circuit as in the process matrix picture (see the right-hand side of Fig.~\hyperref[fig:bipartite_proof]{8(c)}).
In other words, it is precisely that structure that happens on the subsystems with respect to which we chose to describe the circuit. 
This establishes the bipartite result.

Applying the bipartite constructions presented here to the particular case of the quantum switch leads to an asymmetric implementation with Alice performing a time-local operation and Bob's operation being time-delocalised through coherent control of the times at which it is applied. For symmetric implementations in which both Alice's and Bob's operation are time-delocalised, a similar argument can be made~\cite{oreshkov18}.\\

\begin{figure}[h]
    \includegraphics[width=0.9\textwidth]{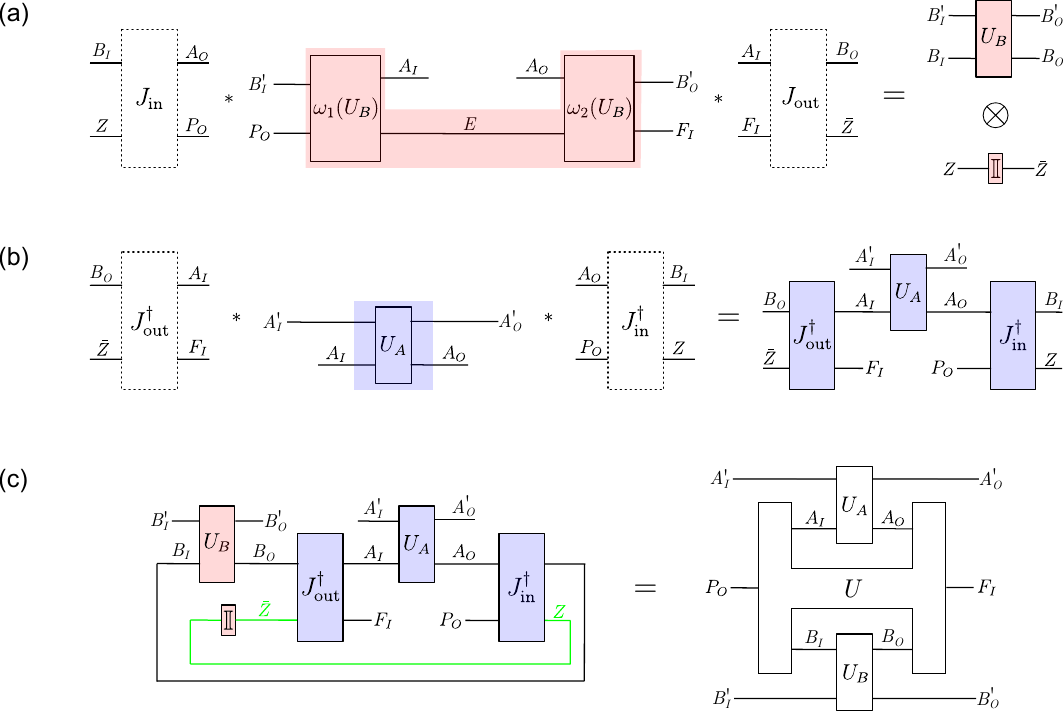}
    \caption{ \textbf{Description of the bipartite temporal circuit in terms of time-delocalised subsystems.}\\
(a) Description of the red circuit fragment, which implements an operation from $\HS^{B_I' P_O A_O}$ to $\HS^{B_O' A_I F_I}$, in terms of the time-delocalised subsystems $B_I$, $Z$ of the joint system $A_O P_O$ and $B_O$, $\bar Z$ of $A_I F_I$. 
\\
(b) Description of the blue circuit fragment, which is simply the operation $U_A$, in terms of the time-delocalised subsystems $B_I$, $Z$ of $A_O P_O$ and $B_O$, $\bar Z$ of $A_I F_I$. 
\\
(c) In the new subsystem description, one obtains a cyclic circuit, as considered in the process matrix framework, where the unitary operation $U$ that defines the process is obtained by composing the inverse isomorphisms $J_{\text{in}}^\dagger$ and $J_{\text{out}}^\dagger$ and the identity channel $\id^{Z \to \bar Z}$ over the subsystems $Z$ and $\bar Z$ (i.e., over the wires shown in green).}
    \label{fig:bipartite_proof}
\end{figure}
\newpage
\textbf{Causal inequality assumptions} \label{app:methods_CI} \quad A causal order between the elements of some set $\mathcal{S}$ is formally described by a \emph{strict partial order} (SPO) on $\mathcal{S}$~\cite{oreshkov12,oreshkov16}.
A SPO is a binary relation $\prec$, which, for all $X$, $Y$, $Z \in \mathcal{S}$, satisfies \emph{irreflexivity} (not $X \prec X$) and \emph{transitivity} (if $X \prec Y$ and $Y \prec Z$, then $X \prec Z$). (Note that irreflexivity and transitivity together imply \emph{asymmetry}, i.e., if $X \prec Y$, then not $Y \prec X$.)
If $X \prec Y$, we will say that $X$ \emph{is in the causal past of} $Y$ (equivalently, $Y$ \emph{is in the causal future of} $X$).
For $X \neq Y$ and not $X \prec Y$, we will use the notation $X \npreceq Y$, and the terminology $X$ \emph{is not in the causal past of} $Y$ (equivalently, $Y$ \emph{is not in the causal future of} $X$).
If $X \npreceq Y$ and $Y \npreceq X$, we will say that $X$ \emph{is in the causal elsewhere of} $Y$~\cite{eddington28} (sometimes also termed $X$ \emph{is not causally connected to} $Y$, or $X$ \emph{is causally disconnected from} $Y$).
For subsets $\mathcal{S}' \subset \mathcal{S}$, we will use the short-hand notation $X \npreceq \mathcal{S}'$ to denote that $\forall \ Y \in \mathcal{S}', X \npreceq Y$.
We furthermore define the \emph{causal past} of $X$ as the set $\P_X \coloneqq \{Y \in \mathcal{S}|Y \prec X \}$, the \emph{causal future} of $X$ as $\F_X \coloneqq \{Y \in \mathcal{S}|X \prec Y\}$ and the \emph{causal elsewhere} of $X$ as $\E_X \coloneqq \{Y \in \mathcal{S}|Y \npreceq X \ \text{and} \ X \npreceq Y\}$. 
Also, note that a SPO on $\mathcal{S}$ naturally induces a SPO on any subset of $\mathcal{S}$.

The variables involved in the process under consideration are the time-delocalised incoming and outgoing variables $A_I$, $A_O$, $B_I$, $B_O$, $C_I$, $C_O$, as well as the settings and outcomes, which can be described by random variables $I_A$, $I_B$, $I_C$ (with values $i_A$, $i_B$, $i_C$, respectively) and $O_A$, $O_B$, $O_C$ (with values $o_A$, $o_B$, $o_C$, respectively).
We will abbreviate the set of all these variables to $\Gamma \coloneqq \{A_I,A_O,B_I,B_O,C_I,C_O,I_A,O_A,I_B,O_B,I_C,O_C\}$.
The assumption that the correlations $P(o_A,o_B,o_C|i_A,i_B,i_C)$ arise from a situation in which these variables occur in a (generally probabilistic and dynamical) causal order can be formalised as follows.
\\ \\
\textit{Causal order assumption.} There exists a random variable which takes values $\kappa(\Gamma)$ in the possible strict partial orders on the set $\Gamma$, and a joint probability distribution $P(o_A, o_B, o_C, \kappa(\Gamma)|i_A,i_B,i_C)$, which, when marginalised over that variable, yields the correlations $P(o_A, o_B, o_C, |i_A,i_B,i_C)$ observable in the process, i.e.,
\begin{equation}
\label{eq:sum_causalStructure}
\sum_{\kappa(\Gamma)} P(o_A, o_B, o_C, \kappa(\Gamma) |i_A,i_B,i_C) = P(o_A, o_B, o_C, |i_A,i_B,i_C).
\end{equation}
This probability distribution satisfies the following two conditions.

\emph{1) Free choice.} The settings $I_A$, $I_B$ and $I_C$ are assumed to be \emph{freely chosen}, which means that they cannot be correlated with any properties pertaining to their causal past or elsewhere.
That is, the probability for their causal past and elsewhere to consist of certain variables, for the variables in these sets to have a certain causal order, and for the outcome variables in these sets to take certain values, cannot depend on the respective setting.
Formally, with respect to $I_A$, for any (disjoint) subsets $\Y$ and $\Z$ of $\Gamma \backslash \{I_A\}$, and any causal order $\kappa(\Y \cup \Z)$ on the variables in $\Y \cup \Z$, the following must hold:
\begin{align}
\label{eq:free_choice}
&P(o^\Y,o^\Z, \P_{I_A} = \Y, \E_{I_A} = \Z, \kappa(\Y \cup \Z) | i_A, i_B, i_C) = P(o^\Y,o^\Z, \P_{I_A} = \Y, \E_{I_A} = \Z, \kappa(\Y \cup \Z) | i_B, i_C) .
\end{align}
Here, by $P(o^\Y,o^\Z, \P_{I_A} = \Y, \E_{I_A} = \Z, \kappa(\Y \cup \Z) | i_A, i_B, i_C)$, we denote the probability that is obtained from $P(o_A, o_B, o_C, \kappa(\Gamma)|i_A,i_B,i_C)$ by marginalising over all $O_X \notin \Y \cup \Z$, and by summing over all $\kappa(\Gamma)$ that satisfy the specified constraints---that is, all $\kappa(\Gamma)$ for which the causal past $\P_{I_A}$ of $I_A$ is $\Y$, the causal elsewhere $\E_{I_A}$ of $I_A$ is $\Z$, and the causal order on the subset $\Y \cup \Z$ is $\kappa(\Y \cup \Z)$. 
The free choice assumption is that this probability is independent of the value of $I_A$.
The analogous conditions must hold with respect to $I_B$ and $I_C$.   

\emph{2) Closed laboratories.} The second constraint is the \emph{closed laboratory} assumption, which says, intuitively speaking, that causal influence from $I_A$ to any other variable except $O_A$ has to pass through $A_O$; that, similarly, any causal influence to $O_A$ from any other variable except $I_A$ has to pass through $A_I$; and that $A_I$ is in the causal past of $A_O$ (and analogously for $B$ and $C$).
Note that, in the original derivation of causal inequalities~\cite{oreshkov12}, it was assumed that $X_I \prec X_O$ always holds. 
Here, we weaken this assumption by requiring that this constraint only holds for at least one particular value of the corresponding setting variable $I_X$.
The reason is that this weakened form of the assumption (unlike the stronger assumption of $X_I \prec X_O$ regardless of the value of $I_X$) is directly motivated by the observable causal relations in our situation with time-delocalised variables (see the discussion in the main text).

This closed laboratory assumption can be formalised as a constraint on the possible causal orders as follows.
\begin{align}
& P(o_A,o_B,o_C,\kappa(\Gamma)|i_A,i_B,i_C) > 0 \text{ only if $\kappa(\Gamma)$ satisfies the following properties for all $Y \in \Gamma$:} \notag \\
\label{eq:CL1}
& \qquad i)\ I_A \prec Y, \ \text{iff} \ Y = O_A \ \text{or} \ Y = A_O \ \text{or} \ A_O \prec Y.
\quad ii) \ Y \prec O_A, \ \text{iff} \ Y = I_A \ \text{or} \ Y = A_I \ \text{or} \ Y \ \prec A_I.
\end{align}
Furthermore, there exists at least one value $i_A^*$ of $I_A$ for which $A_I \prec A_O$ with certainty, that is
\begin{equation}
\label{eq:closedLab3}
    P(o_A,o_B,o_C,\kappa(\Gamma)|i_A^*,i_B,i_C) > 0 \quad \text{only if} \ \kappa(\Gamma) \ \text{satisfies} \ A_I \prec A_O.
\end{equation}
The analogous conditions must be satisfied for $B$ and $C$.\\

We show in Supplementary Note~\hyperref[app:causal_correlations]{5} that this causal order assumption---notably, even with the weakened form of the closed laboratory condition we introduced---implies that the correlations $P(o_A, o_B, o_C|i_A,i_B,i_C)$ that are established in the process must be \emph{causal}~\cite{oreshkov16,abbott16,branciard16}. 
Such correlations form a polytope, whose facets precisely define causal inequalities~\cite{oreshkov16,abbott16,branciard16}.  

(Note furthermore that we could similarly weaken the assumption that $O_A$ is always in the causal future of $A_I$. 
This would however change nothing about the argument, and the proof from Supplementary Note~\hyperref[app:causal_correlations]{5} would go through in the same way.)

Here, we presented the argument in the classical case for concreteness, but it can be readily extended to a quantum process, or even an abstract process~\cite{oreshkov16} possibly compatible with more general operational probabilistic theories (OPTs)~\cite{hardy09,chiribella10}, where there is no analogue of the classical variables $X_I$ and $X_O$. 
Indeed, in the general case all elements of the argument remain the same, except that the objects $X_I$ and $X_O$ over which the partial order is assumed would be general systems rather than classical variables ($I_X$ and $O_X$ will remain classical). 
Moreover, the argument applies analogously for any number of parties, so we have assumptions applicable to the most general case of a process. 


\newpage

\section*{Supplementary Note 1---Testing operations on time-delocalised subsystems}
\label{app:testing}

\subsection{General circuits on time-delocalised subsystems}
\label{sec:testing_general}

The structure of a circuit with respect to a particular choice of time-delocalised subsystems, as described for a generic circuit in Sec.~\hyperref[sec:td_general]{``Time-delocalised subsystems and operations''} of the main text, can be tested operationally~\cite{oreshkov18}.
Figuratively speaking, to achieve this, one ``disconnects'' the circuit fragment under consideration from its complement by ``cutting through'' its incoming and outgoing wires, and by ``pulling all incoming, respectively outgoing, wires to the same time''.
Formally, this means that one performs additional SWAP operations, which send some additional incoming ancillary systems to the incoming wires of the fragment, and its outgoing wires to some additional outgoing ancillary systems (see Fig.~\ref{fig:fragment_generic_intervention} for the example of a fragment considered in \jwchanges{Fig.~\ref{fig:circuit_generic} of the main text}).
This results in the operations implemented by the fragment and its complement effectively taking place on these additional ancillary systems.
By performing suitable time-local preparations (respectively, measurements) on the additional incoming (respectively, outgoing) ancillas, one can then perform tomography on these operations.
Through such a procedure, one can thus in particular probe the operations that happen on the time-delocalised subsystems arising from a given choice of tensor product structure on the incoming and outgoing wires of the fragments, and test operationally that the circuit has a particular (generally cyclic) form when described with respect to these subsystems.
In doing so, we make the assumption that the operations still act on these systems in the ``non-disconnected'' circuit, which is however completely reasonable---in fact, it is standard also for quantum circuits on time-local systems.
Namely, probing the operations that make up a given standard circuit requires intervening around each operation (e.g. with suitable preparations and measurements so as to make tomography of the operation). 
It is an assumption that the practical procedure we employ for doing this does not alter the original operation, that is, the procedure can be described by a modified circuit that contains the original operation acting on the original systems but now connected to the probing operations via these systems, instead of to the operations from the original circuit. 
What we do here is the same, except that the systems we consider are time-delocalised.

\begin{figure}[h]
    \centering
    \includegraphics[width=0.7\textwidth]{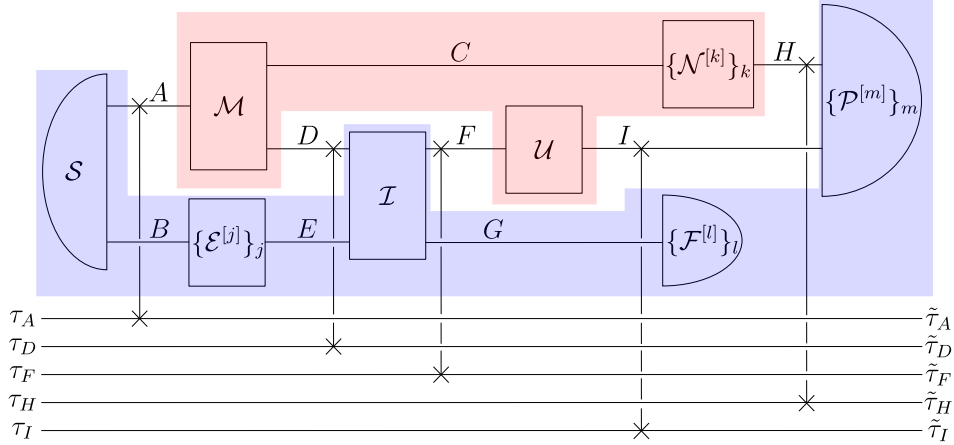}
    \caption{\jwchanges{Disconnecting circuit fragments with SWAP operations.} The circuit \jwchanges{in Fig.~\hyperref[fig:circuit_generic]{2(a)}} can effectively be decomposed into the red fragment and its blue complement, by “cutting” and “pulling to the same time” the corresponding “wires”.
That is, one performs SWAP operations which send some incoming ancillas $\tau_A$, $\tau_F$ to the incoming systems $A$ and $F$ of the fragment, and which send the corresponding output systems of the complementary blue fragment to some outgoing ancillas $\tilde \tau_A$, $\tilde \tau_F$.
Similarly, one inserts SWAP operations which send the outgoing wires $D$, $H$, $I$ of the fragment to outgoing ancillas $\tilde \tau_D$, $\tilde \tau_H$, $\tilde \tau_I$, and some incoming ancillas $\tau_D$, $\tau_H$, $\tau_I$ to the corresponding input systems of the complementary blue fragment.
This allows to test the operations implemented by the circuit fragments operationally by preparing suitable states and performing suitable measurements on the ancillas.
}
    \label{fig:fragment_generic_intervention}
\end{figure}

In the realisations of processes we considered in this work, some of the systems need to be composed in the cyclic circuit in order for the structure from the process matrix framework to emerge (namely, the systems $Z$, $\bar Z$ in the bipartite case, and the systems $Y$, $\bar Y$, $Z$, $\bar Z$, $\bar Q_1$ and $\bar Q_2'$ in the tripartite case). 
This raises notably the question of whether one could test the structure of the cyclic circuit in a way that leaves these systems connected, so as to probe precisely the constituents that appear in the process matrix picture.
In the following, we will outline how this can be achieved for the cases considered in our paper.

\newpage

\subsection{Unitary extensions of bipartite processes on time-delocalised subsystems}

Applying the general argument from Supplementary Note~\hyperref[sec:testing_general]{1~A} to the case of unitary extensions of bipartite processes studied in \hyperref[app:bipartite_methods]{Methods}, one could test operationally that, for any local operations $U_A$ and $U_B$, and with respect to the systems $P_O$, $A_I$, $A_O$, $A_I'$, $A_O'$ $B_I'$, $B_O'$, $F_I$ and the time-delocalised subsystems $B_I$, $B_O$, $Z$ and $\bar Z$ we identified, the bipartite circuit of Fig.~\jwchanges{7} consists of the five operations $U_A$, $U_B$, $J_{\text{in}}^\dagger$, $J_{\text{out}}^\dagger$ and $\id^{Z \to \bar Z}$, which are connected in a cyclic manner as shown on the left-hand side of Fig.~\hyperref[fig:bipartite_proof]{8(c)}.

Once this statement has been established as an operationally verifiable fact, there exists, in particular, an operationally verifiable way to apply SWAP operations only to the incoming and outgoing time-delocalised wires
$B_I$ and $B_O$ of $U_B$, so as to ``cut through'' and ``pull to the future, respectively past'' only these wires, while the time-delocalised wires $Z$ and $\bar Z$ remain connected in the cyclic circuit.
Namely, one can realise a ``modified'' temporal circuit with operations $\omega_1(U_B^{(s)})$ and $\omega_2(U_B^{(s)})$, where $U_B^{(s)}$ is related to the original $U_B$ by $U_B^{(s)}: \HS^{B_I B_I' \tau_{B_I} \tau_{B_O}} \to \HS^{B_O B_O' \tilde \tau_{B_I} \tilde \tau_{B_O}}$ 
with $U_B^{(s)} = (\id^{B_O \to \tilde \tau_{B_O}} \otimes \id^{\tau_{B_O} \to B_O} \otimes \id^{\tilde \tau_{B_I} B_O'}) \cdot (U_B \otimes \id^{\tilde \tau_{B_I} \tau_{B_O}}) \cdot (\id^{B_I \to \tilde \tau_{B_I}} \otimes \id^{\tau_{B_I} \to B_I} \otimes \id^{B_I' \tau_{B_O}})$.
(For the purpose of constructing the corresponding temporal circuit according to Supplementary Note~\hyperref[sec:bipartite_circuit]{2~A}, the ancillas $\tau_{B_I}$, $\tau_{B_O}$ ($\tilde \tau_{B_I}$, $\tilde \tau_{B_O}$) can be incorporated into the incoming (outgoing) ancillary systems of $U_B$).
By construction, in the time-delocalised description as on the left-hand side of Fig.~\hyperref[fig:bipartite_proof]{8(c)}, this modification translates to the SWAP operations which achieve the desired ``disconnecting'' of $U_B$ (see Fig.~\ref{fig:test_ua} for the description of the red fragment of the modified circuit in terms of the time-delocalised subsystems $B_I$, $Z$, $B_O$, $\bar Z$).
Disconnecting only $U_B$ from the cyclic circuit (and disconnecting $U_A$ through standard, time-local SWAP operations on $A_I$ and $A_O$) then in turn allows to operationally test the ``coarse-grained'' structure of the cyclic circuit on the right-hand side of Fig.~\hyperref[fig:bipartite_proof]{8(c)}, where $J_{\text{in}}^\dagger$, $J_{\text{out}}^\dagger$ and $\id^{Z \to \bar Z}$ are composed over the systems $Z$, $\bar Z$ so as to form the operation $U$ that defines the process.

\begin{figure}[h]
         \centering
         \includegraphics[width=0.8\textwidth]{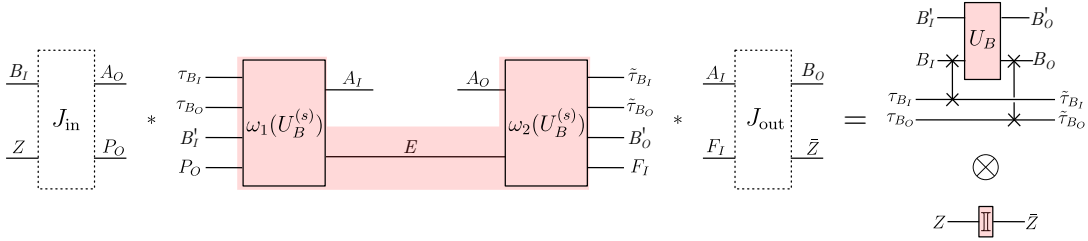}
    \caption{``Disconnecting'' the time-delocalised wires $B_I$ and $B_O$. 
    }
        \label{fig:test_ua}
\end{figure}

\subsection{Unitary extensions of tripartite processes on time-delocalised subsystems}

The tripartite case is a bit more involved, but by successively verifying a sequence of several statements that build on one another, one can similarly test that the circuit of Fig.~\ref{fig:tripartite_implementation_main} in the main text has the cyclic form on the right-hand side of Fig.~\hyperref[fig:tripartite_proof]{4(c)} when described with respect to the systems $P_O$, $A_{IO}^{(\prime)}$, $B_{IO}^{(\prime)}$, $C_{IO}^{(\prime)}$, $F_I$.
By disconnecting the red and blue fragments and applying the general procedure from Supplementary Note~\hyperref[sec:testing_general]{1~A}, one can test that, for any $U_A$, $U_B$, $U_C$, and with respect to the systems $P_O$, $A_{IO}^{(\prime)}$, $B_{IO}^{(\prime)}$, $C_{IO}^{(\prime)}$, $F_I$,  $Y$, $\bar Y$, $Z$, $\bar Z$, $\bar Q_1$, $\bar Q_2'$, we obtain a cyclic circuit consisting of the operations $U_A$, $U_B$, $R(U_C)$ and $R'$ \jwchanges{(see the left-hand side of Fig.~\hyperref[fig:tripartite_proof]{4(c)})}.
In this cyclic circuit, there then exists an operationally verifiable way to apply time-delocalised SWAP operations that ``disconnect'' precisely the time-delocalised wires $C_I$ and $Z$, while the other wires remain connected. 
Namely, in order to achieve this, one modifies the tripartite temporal circuit by inserting some additional operations,  shown in green in Fig.~\hyperref[fig:test_uc]{11(a)}, in the red fragment.
\begin{figure}[h]
     \flushleft{(a)}\\
        \centering
         \includegraphics[width=0.85\textwidth]{test_UC_1_1.pdf}
     \vspace{3mm}\\
      \flushleft{(b)}\\
         \centering
         \includegraphics[width=0.85\textwidth]{test_UC_2_2.pdf}
    \caption{\jwchanges{Disconnecting the time-delocalised systems $C_I$, $Z$, $C_O$ and $\bar Z$.}
    (a) Modifying the red fragment of the tripartite temporal circuit by inserting the additional operations shown in green (with $u_1 \coloneqq (\id^{C_I \to \tau_{C_I}} \otimes \id^{Z \to \tau_{Z}}) \cdot U_1 \cdot (\id^{\tau_{P_O} \to P_O} \otimes \id^{\tau_{A_O} \to A_O} \otimes \id^{\tau_{B_O} \to B_O})$ and $\tilde u_1 \coloneqq (\id^{C_I \to \tilde \tau_{C_I}} \otimes \id^{Z \to \tilde \tau_{Z}}) \cdot U_1 \cdot (\id^{\tilde \tau_{P_O} \to P_O} \otimes \id^{\tilde \tau_{A_O} \to A_O} \otimes \id^{\tilde \tau_{B_O} \to B_O})$, with $U_1$ from Eq.~\eqref{eq:decomp_tripartite})\jwchanges{,} corresponds precisely to performing a SWAP operation that ``cuts'' the time-delocalised wires $C_I$ and $Z$. (Note that, in order to achieve this, the second green control-SWAP operation needs to be inserted in between the identity channel relating $T_2'$ and $\bar T_2'$ and the identity channel relating $Q_2'$ and $\bar Q_2'$, so these identity channels need to be ``shifted'' against each other.
    But the exact placement of these identity channels in the temporal circuit is irrelevant, as long as they are between the second controlled application of $U_A$ and $U_B$ and the circuit operation $\omega_3(U_C)$, so that one can always choose them to be placed in that way.)  \\
    (b) The time-delocalised wires $C_O$ and $\bar Z$ can be ``cut'' in a similar way.
        }
        \label{fig:test_uc}
\end{figure}
In the description of the circuit with respect to $P_O$, $A_{IO}^{(\prime)}$, $B_{IO}^{(\prime)}$, $C_{IO}^{(\prime)}$, $F_I$,  $Y$, $\bar Y$, $Z$, $\bar Z$, $\bar Q_1$, $\bar Q_2'$, these additional operations correspond precisely to a SWAP operation that sends the incoming ancilla $\tau_{C_I}$ into the operation $R(U_C)$ and the time-delocalised subsystem $C_I$ to an outgoing ancilla $\tilde \tau_{C_I}$, and similarly for the system $Z$, as shown on the right-hand side of Fig.~\hyperref[fig:test_uc]{11(a)}.
In an analogous way, one can modify the temporal circuit so as ``disconnect'' the wires $C_O$ and $\bar Z$, which is shown in Fig.~\hyperref[fig:test_uc]{11(b)}.
That the described modification corresponds to such SWAP operations is again an operationally verifiable statement, which can be tested by disconnecting the red and blue fragment in the modified circuit and applying the general argument from Supplementary Note~\hyperref[sec:testing_general]{1~A}, in the same way as for the tripartite circuit without the additional operations.

Once this fact has been established, by disconnecting the systems $C_I$, $C_O$, $Z$ and $\bar Z$ in this way, while leaving $Y$, $\bar Y$, $\bar Q_1$, and $\bar Q_2'$ composed (and disconnecting $A_I$, $A_O$, $B_I$ and $B_O$, which can be done in the same way as for $B_I$ and $B_O$ in the bipartite case) one can then test the structure of the cyclic circuit with respect to the systems $P_O$, $A_I$, $A_O$, $B_I$, $B_O$, $C_I$, $C_O$, $F_I$, $Z$ and $\bar Z$. 
(With respect to these systems, the circuit consists of operations $U_1: \HS^{P_O A_O B_O} \to \HS^{C_I Z}$, $U_2: \HS^{C_O \bar Z} \to \HS^{A_I B_I F_I}$, $\id^{Z \to \bar Z}$, $U_A$, $U_B$ and $U_C$ (cf. Eq.~\eqref{eq:tripartite_final}).
Finally, once the structure of the circuit with respect to these systems has been established as an operationally verifiable fact (and, in particular, it has been established that $U_C$ acts on the systems $C_I$ and $C_O$), one can apply the analogous argument from the bipartite case to $U_C$ in order to only ``disconnect'' the wires $C_I$ and $C_O$ and leave $Z$ and $\bar Z$ connected.
This in turn allows to ``test'' also the structure of the circuit as \jwchanges{on the right-hand side of} Fig.~\hyperref[fig:tripartite_proof]{4(c)}, where $U_1$, $U_2$ and $\id^{Z \to \bar Z}$ are connected so as to form the operation $U$.

\newpage
\section*{Supplementary Note 2---Unitary extensions of bipartite processes on time-delocalised subsystems}
\label{app:bipartite_proof}
\setcounter{subsection}{0}

In this \jwchanges{Supplementary Note}, we will show all technical proofs and calculations pertaining to \hyperref[app:bipartite_methods]{``Unitary extensions of bipartite processes on time-delocalised subsystems''} in Methods.

\subsection{Construction of the temporal circuit}
\label{sec:bipartite_circuit}

First, we show how to construct the temporal circuit shown in Fig.~\ref{fig:bipartite_implementation_main}.
For any local operation $U_B$ performed by Bob, its composition with the process vector 
\begin{equation}
\label{eq:def_red_a}
\dket{U} * \dket{U_B} \eqqcolon \dket{U_\G (\cdot, U_B)} \qquad \in \HS^{P_O A_{IO} B_{IO}' F_I}
\end{equation}
is, mathematically, the process vector of a unitarily extended one-party process for the remaining party Alice, where the global past party has an output space $\HS^{B_I' P_O}$, and where the global future party has an input space $\HS^{B_O' F_I}$. (This follows from the fact that, for any unitary operation $U_A$ performed by Alice, $\dket{U_\G(\cdot, U_B)} * \dket{U_A} = \dket{U_\G(U_A,U_B)}$ is a unitary operation from $\HS^{A_I' B_I' P_O}$ to $\HS^{A_O' B_O' F_I}$).
It is well known~\cite{chiribella08a} that the most general process of this kind is a \textit{quantum comb}~\cite{chiribella09}, and thus has a realisation as a fixed-order quantum circuit. 
In general, this circuit realisation can be taken to consist of two isometric operations before and after Alice's operation, which are connected by a circuit ancilla or ``quantum memory'' system, and whose link product is the minimal Stinespring dilation of the process matrix, with the dilating system being traced out at the end of the circuit~\cite{chiribella08a,gutoski06,chiribella09}. 
Since the process here is unitary, the minimal Stinespring dilation is trivial, and the corresponding circuit isometries must be unitaries.
In other words, $\dket{U_\G(\cdot,U_B)}$ can be decomposed as 
\begin{equation}
\label{eq:decomp_red_a}
\dket{U_\G(\cdot,U_B)} = \dket{\omega_1(U_B)} * \dket{\omega_2(U_B)},
\end{equation}
with two unitary operations $\omega_1(U_B): \HS^{B_I'P_O} \to \HS^{A_I E}$ and $\omega_2(U_B): \HS^{A_O E} \to \HS^{B_O' F_I}$, which depend on $U_B$, and with an ancillary system $E$ (whose dimension also depends on $U_B$ in general). 

The composition $\dket{\omega_1(U_B)} * \dket{U_A} * \dket{\omega_2(U_B)}$ then describes a temporally ordered quantum circuit which by construction implements the global output operation $U_\G(U_A, U_B)$ of the process, and which consists of the three subsequent operations $\omega_1(U_B)$, $U_A: \HS^{A_I A_I'} \to \HS^{A_O A_O'}$ and $\omega_2(U_B)$, as shown in Fig.~\ref{fig:bipartite_implementation_main}.

\subsection{Derivation of the isomorphisms $J_{\text{in}}$ and $J_{\text{out}}$}
\label{app:iso_bipartite}

To derive the isomorphisms $J_{\text{in}}$ and $J_{\text{out}}$, we first derive a mathematical relation between the systems in the process. 
Namely, we note that the unitary $U$ isomorphically maps some subsystem of $A_O P_O$ to Bob's incoming system $B_I$, and Bob's outgoing system $B_O$ to a subsystem of $A_I F_I$. 
That is, formally, $U$ can be decomposed as $U = (\id^{B_I} \otimes U_2)\cdot (\id^{B_O} \otimes \id^{B_I} \otimes \id^{Z \to \bar Z}) \cdot (\id^{B_O} \otimes U_1)$, or, in the Choi representation,
\begin{equation}
\label{eq:decomp_bipartite}
    \dket{U} = \dket{U_1} * \dket{\id}^{Z \bar Z} * \dket{U_2}
\end{equation}
(see Fig.~\ref{fig:bipartite_decomposition}), where $U_1: \HS^{A_O P_O} \to \HS^{B_I Z}$ and $U_2: \HS^{B_O \bar Z} \to \HS^{A_I F_I}$ are unitary (and the identity channel from $Z$ to $\bar Z$ between the two isomorphic complementary subsystems $Z$ and $\bar Z$ is introduced for later convenience).
\begin{figure}[h]
    \centering
    \includegraphics[width=0.6\textwidth]{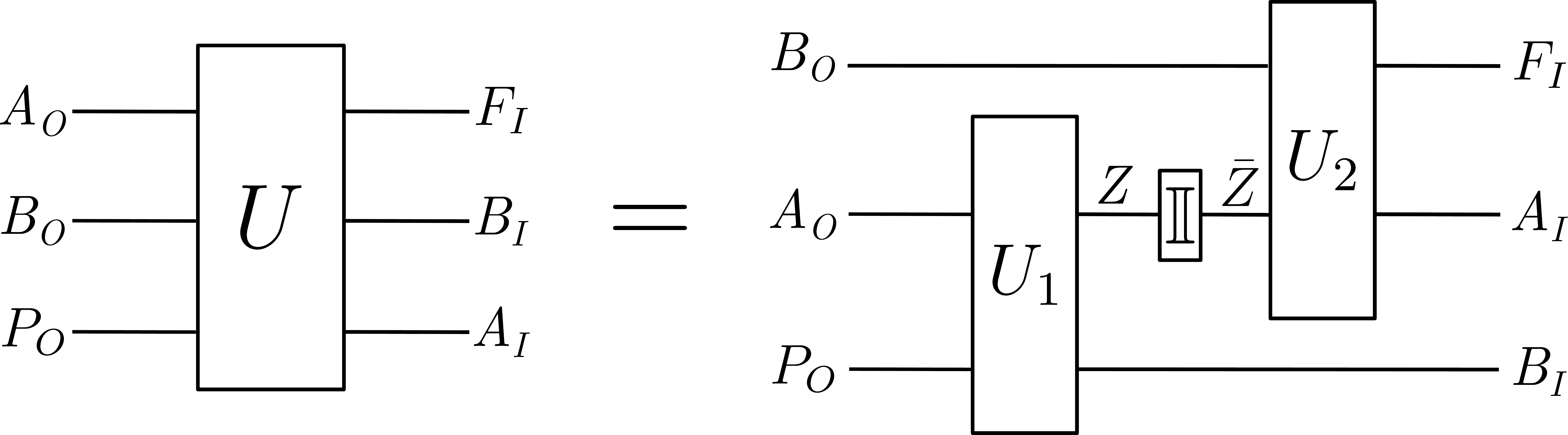}
    \caption{\jwchanges{Graphical illustration of Eq.~\eqref{eq:decomp_bipartite}.} The process vector describing a unitary extension of a bipartite process can be decomposed as $\dket{U} = \dket{U_1} * \dket{\id}^{Z \bar Z} * \dket{U_2}$, with unitaries $U_1: \HS^{A_O P_O} \to \HS^{B_I Z}$ and $U_2: \HS^{B_O \bar Z} \to \HS^{A_I F_I}$.
    The circuit on the right-hand side shows that the unitary $U$ maps some subsystem of $A_O P_O$ to $B_I$, and $B_O$ to some subsystem of $A_I F_I$. }
    \label{fig:bipartite_decomposition}
\end{figure}

Such a decomposition with unitary $U_1$ and $U_2$ exists for any unitarily extended bipartite process. 
This can be shown as follows. 
Consider the situation where Alice performs the specific SWAP unitary operation $S_A \coloneqq \id^{A_I' \to A_O} \otimes \id^{A_I \to A_O'}$, i.e., she performs identity channels from an ancillary system $A_I'$ to $A_O$ and from $A_I$ to an ancillary system $A_O'$ which effectively ``pull her output wire to the past'' and ``pull her input wire to the future''. 
When the process vector is composed with that particular local operation, the result $\dket{U_\G(S_A,\cdot)} \coloneqq \dket{U} * \dket{S_A} \in \HS^{P_O A_{IO}' B_{IO} F_I}$ is now the process vector of a unitarily extended one-party process for the remaining party Bob, where the global past party has an output space $\HS^{A_I' P_O}$, and where the global future party has an input space $\HS^{A_O' F_I}$, and which can again be realised as a fixed-order circuit.
That is, $\dket{U_\G(S_A,\cdot)}$ can be decomposed as $\dket{U_\G(S_A,\cdot)} = \dket{\tilde U_1} * \dket{\id}^{Z \bar Z} * \dket{\tilde U_2}$, with two circuit operations $\tilde U_1: \HS^{A_I'P_O} \to \HS^{B_I Z}$ and $\tilde U_2: \HS^{B_O \bar Z} \to \HS^{A_O'F_I}$ (which can be taken to be unitaries by the same argument as for $\omega_1(U_B)$ and $\omega_2(U_B)$ in Supplementary Note~\hyperref[sec:bipartite_circuit]{2~A} above, i.e., because the circuit realisation can be taken to achieve the minimal Stinespring dilation of the quantum comb). 
Since $\dket{U}$ can be recovered from $\dket{U_\G(S_A,\cdot)}$ through $\dket{U} = \dket{U_\G(S_A,\cdot)} * \dket{\id}^{A_O A_I'} * \dket{\id}^{A_O' A_I}$, one obtains the decomposition~\eqref{eq:decomp_bipartite} for $\dket{U}$, with $\dket{U_1} \coloneqq \dket{\tilde U_1} * \dket{\id}^{A_O A_I'}$ and $\dket{U_2} \coloneqq \dket{\tilde U_2} * \dket{\id}^{A_O' A_I}$. 

Up to this point, Eq.~\eqref{eq:decomp_bipartite} is just an abstract, mathematical relation between certain systems in the process.  
In the circuit of Fig.~\ref{fig:bipartite_implementation_main}, where $A_O$, $P_O$ as well as $A_I$, $F_I$ are time-local wires, this decomposition can now be taken to define an alternative description of the circuit in terms of  time-delocalised subsystems. 
Namely, we define the isomorphisms $J_\text{in}: \HS^{B_I Z} \to \HS^{A_O P_O}$ and $J_\text{out}: \HS^{A_I F_I} \to \HS^{B_O \bar Z}$ simply to be the inverses of $U_1$ and $U_2$, respectively, that is, $J_\text{in} \coloneqq U_1^\dagger$ and $J_\text{out} \coloneqq U_2^\dagger$.

\subsection{Changing to the description of the circuit in terms of time-delocalised subsystems}
\label{sec:bipartite_subsystems}

The gates of the circuit in Fig.~\ref{fig:bipartite_implementation_main} compose to the overall transformation 
\begin{equation}
\label{eq:global_trafo_bipartite}
    \dket{\omega_1(U_B)}^{P_O B_I' A_I E} * \dket{U_A}^{A_{IO} A_{IO}'} * \dket{\omega_2(U_B)}^{A_O E B_O' F_I} 
\end{equation}
from its initial systems $P_O A_I'B_I'$ to its final systems $F_I A_O' B_O'$. 
To change to the alternative description of the circuit in terms of time-delocalised subsystems, we start by decomposing the circuit into the blue and red circuit fragments shown in Fig.~\jwchanges{8}.
That is, formally, by using the properties of the link product (i.e., its commutativity, associativity and the fact that it reduces to a tensor product for non-overlapping Hilbert spaces), we rewrite~\eqref{eq:global_trafo_bipartite} as
\begin{equation}
\label{eq:fragments_1}
    \big[\dket{\omega_1(U_B)}^{P_O B_I' A_I E} * \dket{\omega_2(U_B)}^{A_O E B_O' F_I}\big] * \dket{U_A}^{A_{IO} A_{IO}'}.
\end{equation}
The red fragment, which corresponds to the term in the first pair of square brackets, implements a unitary operation from $\HS^{P_O B_I' A_O}$ to $\HS^{F_I B_O' A_I}$.  
This unitary operation satisfies
\begin{align}
\label{eq:first_fragment_rewritten}
    &\dket{\omega_1(U_B)}^{P_O B_I' A_I E} * \dket{\omega_2(U_B)}^{A_O E B_O' F_I} =  \dket{U_1}^{A_O P_O B_I Z} * \dket{\id}^{Z \bar Z} * \dket{U_2}^{B_O \bar Z A_I F_I} * \dket{U_B}^{B_{IO} B_{IO}'}  \notag \\
    &= \dket{J_{\text{in}}^\dagger}^{A_O P_O B_I Z} * (\dket{U_B} \otimes \dket{\id}^{Z \bar Z}) * \dket{J_{\text{out}}^\dagger}^{B_O \bar Z A_I F_I},
\end{align}
where the first equality follows from combining Eqs.~\eqref{eq:def_red_a}, \eqref{eq:decomp_red_a} and~\eqref{eq:decomp_bipartite} and the second equality follows from the definition of the isomorphisms  $J_\text{in}$ and $J_\text{out}$ as the inverses of $U_1$ and $U_2$, respectively.

Therefore, when we apply $J_{\text{in}}$ and $J_{\text{out}}$ on the incoming and outgoing systems of the fragment, we obtain the tensor product structure
\begin{align}
    \dket{J_{\text{in}}}^{A_O P_O B_I Z} * \big[\dket{\omega_1(U_B)}^{P_O B_I' A_I E} * \dket{\omega_2(U_B)}^{A_O E B_O' F_I}\big] * \dket{J_{\text{out}}}^{B_O \bar Z A_I F_I} = \dket{U_B}^{B_{IO} B_{IO}'} \otimes \dket{\id}^{Z \bar Z}.
\end{align}
That is, the circuit fragment indeed consists of precisely one instance of $U_B$, which is applied locally on the time-delocalised subsystems $B_I$ and $B_O$, in parallel to an identity channel from $Z$ to $\bar Z$, as is shown in Fig.~{\jwchanges{8}}(a).

In order to describe the entire circuit in terms of the newly chosen subsystem description, we need to rewrite also the blue circuit fragment, which simply consists of $U_A$, in terms of this new subsystem description.
To do this, we compose its incoming wire $A_I$ with the isomorphism $J_{\text{out}}^\dagger = U_2$, and its outgoing wire $A_O$ with $J_{\text{in}}^\dagger = U_1$ (see Fig.~{\jwchanges{8}}(b)). 

Recomposing the two fragments in the new subsystem decomposition over the systems $Z$, $\bar Z$ then yields
\begin{align}
    &[\dket{U_B}^{B_{IO} B_{IO}'} \otimes \dket{\id}^{Z \bar Z}]   *[\dket{U_2}^{B_O \bar Z A_I F_I} * \dket{U_A}^{A_{IO} A_{IO}'}* \dket{U_1}^{A_O P_O B_I Z}]  \notag \\
    = &[\dket{U_1}^{A_O P_O B_I Z} \otimes \dket{U_2}^{B_O \bar Z A_I F_I} * \dket{\id}^{Z \bar Z}] * \dket{U_A}^{A_{IO} A_{IO}'} * \dket{U_B}^{B_{IO} B_{IO}'}\notag \\
    = &\dket{U}^{P_O A_{IO} B_{IO} F_I} * \dket{U_A}^{A_{IO} A_{IO}'} \otimes \dket{U_B}^{B_{IO} B_{IO}'},
\end{align}
as shown graphically in Fig.~\hyperref[fig:bipartite_proof]{8(c)}.


\section*{Supplementary Note 3---Unitary extensions of tripartite processes on time-delocalised subsystems}
\label{app:tripartite_proof}
\setcounter{subsection}{0}

In this Supplementary Note, we will show all technical proofs and calculations pertaining to \hyperref[sec:td_tripartite]{``Unitary extensions of tripartite processes on time-delocalised subsystems''} of the main text. 

\subsection{Construction of the temporal circuit}
\label{app:tripartite_circuit}

The bipartite proof relies crucially on the fact that the unitarily extended one-party process that one obtains by fixing Bob's operation has a particular ``standard form''---namely, a circuit in which Alice's operation acts at a fixed time.
This provides us with some fixed, time-local(ised) physical systems relative to which we can define the time-delocalised subsystems on which Bob's operation acts based on the mapping effected by the unitary that defines the process.
For unitary extensions of bipartite processes, it was shown in Refs.~\cite{barrett20,yokojima20} that a similar standard form exists.
Namely, all unitary extensions of bipartite processes are \emph{variations of the quantum switch}.
This result forms the basis for the tripartite construction we derive in this paper.

In Refs.~\cite{barrett20,yokojima20}, the following characterisation was proven. For any process vector $\dket{U} \in \HS^{P_O A_{IO} B_{IO} F_I}$ of a unitary extension of a bipartite process, the output, respectively input Hilbert spaces of the global past and future parties can be decomposed into a direct sum $\HS^{P_O} = \HS^{P_O^{\ell}} \oplus \HS^{P_O^r}$ and $\HS^{F_I} = \HS^{F_I^{\ell}} \oplus \HS^{F_I^r}$, such that $\dket{U}$ has the form 
\begin{align}
\label{eq:BLO_decomp}
    \dket{U} = &\dket{\nu_1^{A \prec B}} * \dket{\nu_2^{A \prec B}} *  \dket{\nu_3^{A \prec B}} +  \dket{\nu_1^{B \prec A}} *  \dket{\nu_2^{B \prec A}} *  \dket{\nu_3^{B \prec A}}  
\end{align}
with unitary operations $\nu_1^{A \prec B}: \HS^{P_O^{\ell}} \to \HS^{A_I} \otimes \HS^{\lambda_1}$, $\nu_2^{A \prec B}: \HS^{A_O} \otimes \HS^{\lambda_1} \to \HS^{B_I} \otimes \HS^{\lambda_2}$, $\nu_3^{A \prec B}: \HS^{B_O} \otimes \HS^{\lambda_2} \to \HS^{F_I^{\ell}}$, as well as unitary operations  $\nu_1^{B \prec A}: \HS^{P_O^r} \to \HS^{B_I} \otimes \HS^{\rho_1}$, $\nu_2^{B \prec A}: \HS^{B_O} \otimes \HS^{\rho_1} \to \HS^{A_I} \otimes \HS^{\rho_2}$, $\nu_3^{B \prec A}: \HS^{A_O} \otimes \HS^{\rho_2} \to \HS^{F_I^{r}}$ (and conversely, any vector of the form as in Eq.~\eqref{eq:BLO_decomp} is the process vector of a valid unitarily extended bipartite process). In other words, any such $\dket{U}$ can be decomposed into a sum of two process vectors, the first of which describes a fixed-order circuit (i.e., a quantum comb) with the ``global past'' output space $\HS^{P_O^{\ell}}$ and the ``global future'' input space $\HS^{F_I^{\ell}}$, which consists of unitary circuit operations and in which $U_A$ is applied before $U_B$.
Similarly, the second summand in the decomposition~\eqref{eq:BLO_decomp} corresponds to a quantum comb with the ``global past'' output space $\HS^{P_O^r}$ and the ``global future'' input space $\HS^{F_I^r}$, in which $U_B$ is applied before $U_A$.

First, we address a technicality regarding the dimensions of the systems in such a decomposition.
Namely, the two combs in Eq.~\eqref{eq:BLO_decomp} generally involve ancillas of different dimensions $d_{\lambda_1} = d_{\lambda_2}$ and $d_{\rho_1} = d_{\rho_2}$ (or equivalently, global past and future spaces with different dimensions $d_{P_O^{\ell}} = d_{F_I^{\ell}}$ and $d_{P_O^r} = d_{F_I^r}$).
For the construction of an alternative temporal circuit for unitarily extended bipartite circuits that we will consider below (Fig.~\ref{fig:bipartite_symm}), it is however convenient to consider a process in which the dimensions in these two combs are the same.
We therefore show that any unitarily extended bipartite process as in Eq.~\eqref{eq:BLO_decomp} can be recovered as part of a suitable ``enlarged'' process for which this condition on the dimension is indeed satisfied.

To define this enlarged process, we consider arbitrary ``complementary'' unitary operations $\nu_{1\text{,comp.}}^{A \prec B}: \HS^{P_O^{r}} \to \HS^{A_I} \otimes \HS^{\rho_1}$, $\nu_{2\text{,comp.}}^{A \prec B}: \HS^{A_O} \otimes \HS^{\rho_1} \to \HS^{B_I} \otimes \HS^{\rho_2}$, $\nu_{3\text{,comp.}}^{A \prec B}: \HS^{B_O} \otimes \HS^{\rho_2} \to \HS^{F_I^{r}}$, as well as  $\nu_{1\text{,comp.}}^{B \prec A}: \HS^{P_O^\ell} \to \HS^{B_I} \otimes \HS^{\lambda_1}$, $\nu_{2\text{,comp.}}^{B \prec A}: \HS^{B_O} \otimes \HS^{\lambda_1} \to \HS^{A_I} \otimes \HS^{\lambda_2}$ and $\nu_{3\text{,comp.}}^{B \prec A}: \HS^{A_O} \otimes \HS^{\lambda_2} \to \HS^{F_I^{\ell}}$. 
We then define $\HS^{E_1} \coloneqq \HS^{\lambda_1} \oplus \HS^{\rho_1}$ and $\HS^{E_2} \coloneqq \HS^{\lambda_2} \oplus \HS^{\rho_2}$, and introduce two additional two-dimensional Hilbert spaces $\HS^p$ (with computational basis $\{\ket{0}^{p},\ket{1}^{p}\}$) and $\HS^f$ (with computational basis $\{\ket{0}^{f},\ket{1}^{f}\}$).
From that, we implicitly consider the various $\nu$ unitaries above to act on the extended ancillary spaces $\HS^{E_i}$ (by adding null contributions on the originally untouched subspaces $\HS^{\rho_i}$ or $\HS^{\lambda_i}$) and define the ``enlarged'' operations (denoted here with tildes) $\tilde \nu_{1}^{A \prec B}: \HS^{\tilde P_O^\ell} \to \HS^{A_I} \otimes \HS^{E_1}$, with $\HS^{\tilde P_O^\ell} \coloneqq \HS^{P_O^\ell} \otimes \Span\{ \ket{0}^p\} \oplus \HS^{P_O^r} \otimes \Span\{ \ket{1}^p\}$, as
\begin{equation}
\label{eq:nu1_elg_AB}
    \tilde \nu_1^{A \prec B} \coloneqq \nu_1^{A \prec B} \otimes \bra{0}^p + \nu_{1\text{,comp.}}^{A \prec B} \otimes \bra{1}^p,
\end{equation} 
$\tilde \nu_{2}^{A \prec B}: \HS^{A_O} \otimes \HS^{E_1} \to \HS^{B_I} \otimes \HS^{E_2}$ as
\begin{equation}
\label{eq:nu2_elg_AB}
\tilde \nu_{2}^{A \prec B} \coloneqq \nu_2^{A \prec B} + \nu_{2\text{,comp.}}^{A \prec B}
\end{equation}
and $\tilde \nu_3^{A \prec B}: \HS^{B_O} \otimes \HS^{E_2} \to \HS^{\tilde F_I^\ell}$, with $\HS^{\tilde F_I^\ell} \coloneqq \HS^{F_I^\ell} \otimes \Span\{ \ket{0}^f\} \oplus \HS^{F_I^r} \otimes \Span\{ \ket{1}^f\}$, as
\begin{equation}
\label{eq:nu3_elg_AB}
\tilde \nu_3^{A \prec B} \coloneqq \nu_3^{A \prec B} \otimes \ket{0}^f + \nu_{3\text{,comp.}}^{A \prec B} \otimes \ket{1}^f.
\end{equation}
Similarly, we define $\tilde \nu_1^{B \prec A}: \HS^{\tilde P_O^r} \to \HS^{B_I} \otimes \HS^{E_1}$, with $\HS^{\tilde P_O^r} \coloneqq \HS^{P_O^r} \otimes \Span\{ \ket{0}^p\} \oplus \HS^{P_O^\ell} \otimes \Span\{ \ket{1}^p\}$,
as
\begin{equation}
\label{eq:nu1_elg_BA}
\tilde \nu_1^{B \prec A} \coloneqq \nu_1^{B \prec A} \otimes \bra{0}^p + \nu_{1\text{,comp.}}^{B \prec A} \otimes \bra{1}^p,
\end{equation}
$\tilde \nu_2^{B \prec A}: \HS^{B_O} \otimes \HS^{E_1} \to \HS^{A_I} \otimes \HS^{E_2}$ as
\begin{equation}
\label{eq:nu2_elg_BA}
\tilde \nu_2^{B \prec A} \coloneqq \nu_2^{B \prec A} + \nu_{2\text{,comp.}}^{B \prec A}
\end{equation}
and $\tilde \nu_3^{B \prec A}: \HS^{A_O} \otimes \HS^{E_2} \to \HS^{\tilde F_I^r}$, with $\HS^{\tilde F_I^r} \coloneqq \HS^{F_I^r} \otimes \Span\{ \ket{0}^f\} \oplus \HS^{F_I^\ell} \otimes \Span\{ \ket{1}^f\}$, as
\begin{equation}
\label{eq:nu3_elg_BA}
\tilde \nu_3^{B \prec A} \coloneqq \nu_3^{B \prec A} \otimes \ket{0}^f + \nu_{3\text{,comp.}}^{B \prec A} \otimes \ket{1}^f.
\end{equation}
These thus defined enlarged operations act unitarily on input and output Hilbert spaces that all have the same dimension.
They can be combined to an enlarged process vector
\begin{align}
\label{eq:UComp}
    \dket{\tilde U} \coloneqq & \dket{\tilde \nu_1^{A \prec B}} * \dket{\tilde \nu_2^{A \prec B}} *  \dket{\tilde \nu_3^{A \prec B}} + \dket{\tilde \nu_1^{B \prec A}} *  \dket{\tilde \nu_2^{B \prec A}} *  \dket{\tilde \nu_3^{B \prec A}},
\end{align}
with a ``global past'' party that has two output systems $P_O$ and $p$ (with the corresponding Hilbert space $\HS^{P_O p}$ consisting of the two orthogonal subspaces of the same dimensions $\HS^{\tilde P_O^\ell}$, corresponding to the ``$A \prec B$ branch'', and $\HS^{\tilde P_O^r}$, corresponding to the ``$B \prec A$ branch''), and a ``global future'' party that has two input systems $F_I$ and $f$ (with the corresponding Hilbert space $\HS^{F_I f}$ consisting of the two orthogonal subspaces of the same dimensions $\HS^{\tilde F_I^\ell}$, corresponding to the ``$A \prec B$ branch'', and $\HS^{\tilde F_I^r}$, corresponding to the ``$B \prec A$ branch'').
The enlarged process vector has the desired property that the circuit ancillas are the same in each coherent branch (namely, $E_1$ and $E_2$).
The original process is recovered when the state $\ket{0}^p$ is prepared in the output subsystem $p$ of the global past party, and the subsystem $f$ of the global future party is discarded, i.e. $\dket{U}\dbra{U} = \Tr_f [\ket{0}\bra{0}^p * \dket{\tilde U}\dbra{\tilde U}]$ (note that all link products between the original and ``complementary'' operations evaluate to zero, since they are composed over orthogonal subspaces of the circuit ancillas $E_1$ and $E_2$; when inputting the initial state $\ket{0}^p$ the subsystem $f$ is also guaranteed to end up in the state $\ket{0}^f$, so that discarding $f$ above is then in fact equivalent to projecting it onto $\ket{0}^f$, and does not introduce any decoherence).

With this in place, we now construct an alternative temporal circuit for unitary extensions of bipartite processes (different from the one in Fig.~\ref{fig:bipartite_implementation_main}), which will provide the basis for the tripartite generalisation. 
This temporal circuit is shown in Fig.~\ref{fig:bipartite_symm}.

\begin{figure*}[t]
    \centering
    \includegraphics[width=0.7\textwidth]{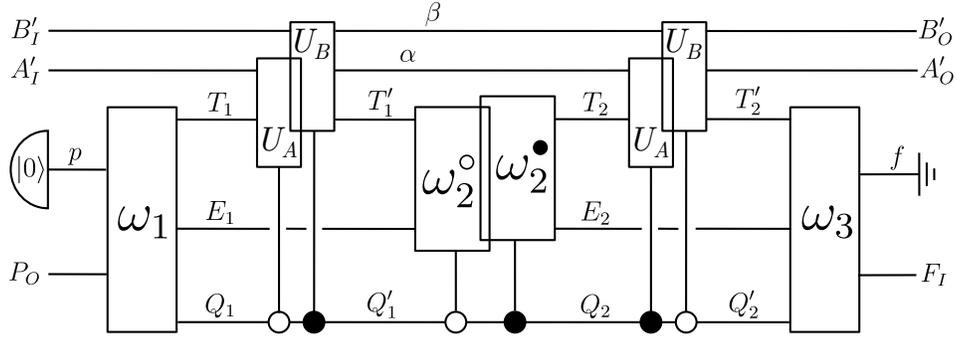}
    \caption{Unitary extensions of bipartite processes can be realised as ``variations of the quantum switch''.}
    \label{fig:bipartite_symm}
\end{figure*}

The circuit receives an input state in the output space $\HS^{P_O}$ of the ``global past'' party, as well as the fixed state $\ket{0}^{p}$ in the space $\HS^p$ (see above). 
A unitary operation $\omega_1: \HS^{P_O p} \to \HS^{T_1 E_1 Q_1}$ given by
\begin{align}
    \label{eq:omega_1}
    \omega_1 = \ & \big[(\id^{E_1} \otimes \id^{A_I \to T_1}) \cdot \tilde \nu_{1}^{A \prec B}\big] \otimes \ket{0}^{Q_1} + \big[(\id^{E_1} \otimes \id^{B_I \to T_1}) \cdot \tilde \nu_{1}^{B \prec A}\big] \otimes \ket{1}^{Q_1} 
\end{align}
(with $\tilde \nu_{1}^{A \prec B}$ from Eq.~\eqref{eq:nu1_elg_AB} and $\tilde \nu_1^{B \prec A}$ from Eq.~\eqref{eq:nu1_elg_BA}, each being implicitly and trivially extended to act on the whole space $\HS^{P_O p} = \HS^{\tilde P_O^\ell} \oplus \HS^{\tilde P_O^r}$ rather than on just one of the subspaces) then takes these systems to a $d$-dimensional ``target system'' $T_1$, a two-dimensional ``control system'' $Q_1$ (with computational basis $\{\ket{0}^{Q_1},\ket{1}^{Q_1}\}$) and a circuit ancilla $E_1$. 
After that, a time-local instance of Alice's operation $U_A$ or Bob's operation $U_B$ is applied to the target system, coherently conditioned on the control system being in the states $\ket{0}^{Q_1}$, or $\ket{1}^{Q_1}$, respectively; that is, formally, the circuit applies the unitary operation 
\begin{align}
& [(\id^{A_O \to T_1'} \otimes \id^{A_O' \to \alpha}) \cdot U_A \cdot (\id^{T_1 \to A_I} \otimes \id^{A_I'})] \otimes \id^{B_I' \to \beta} \otimes \ket{0}^{Q_1'}\bra{0}^{Q_1} \notag \\
& \ + [(\id^{B_O \to T_1'} \otimes \id^{B_O' \to \beta}) \cdot U_B \cdot (\id^{T_1 \to B_I} \otimes \id^{B_I'})] \otimes \id^{A_I' \to \alpha} \otimes \ket{1}^{Q_1'}\bra{1}^{Q_1},
\end{align}
where $\HS^\alpha$ is of the same dimension as Alice's local ancillary spaces $\HS^{A_I'}$ and $\HS^{A_O'}$, and $\HS^\beta$ is of the same dimension as Bob's local ancillary spaces $\HS^{B_I'}$ and $\HS^{B_O'}$. 
The circuit then proceeds with a coherently controlled application of a unitary operation $\omega_2^\circ: \HS^{T_1' E_1} \to \HS^{T_2 E_2}$ or $\omega_2^\bullet: \HS^{T_1' E_1} \to \HS^{T_2 E_2}$. 
That is, it applies the unitary operation
\begin{equation}
\omega_2^\circ \otimes \ket{0}^{Q_2}\bra{0}^{Q_1'} + \omega_2^\bullet \otimes \ket{1}^{Q_2}\bra{1}^{Q_1'}.
\end{equation}
with
\begin{equation}
 \label{eq:omega_2_circ}
    \omega_2^\circ = (\id^{E_2} \otimes \id^{B_I \to T_2}) \cdot \tilde \nu_2^{A \prec B} \cdot (\id^{T_1' \to A_O} \otimes \id^{E_1}), 
\end{equation}
as well as
\begin{equation}
 \label{eq:omega_2_bullet}
    \omega_2^\bullet = (\id^{E_2} \otimes \id^{A_I \to T_2}) \cdot \tilde \nu_2^{B \prec A} \cdot (\id^{T_1' \to B_O} \otimes \id^{E_1}). 
\end{equation}
(with $\tilde \nu_2^{A \prec B}$ from Eq.~\eqref{eq:nu2_elg_AB} and $\tilde \nu_2^{B \prec A}$ from Eq.~\eqref{eq:nu2_elg_BA}).
Then follows the second controlled ``time-local'' application of Alice's and Bob's operations to the target system, which is given by
\begin{align}
    & [(\id^{B_O \to T_2'} \otimes \id^{B_O'}) \cdot U_B \cdot (\id^{T_2 \to B_I} \otimes \id^{\beta \to B_I'})] \otimes \id^{\alpha \to A_O'} \otimes \ket{0}^{Q_2'}\bra{0}^{Q_2} \notag \\
    &\ + [(\id^{A_O \to T_2'} \otimes \id^{A_O'}) \cdot U_A  \cdot (\id^{T_2 \to A_I} \otimes \id^{\alpha \to A_I'})] \otimes \id^{\beta \to B_O'} \otimes \ket{1}^{Q_2'}\bra{1}^{Q_2}.
\end{align}
Finally, the circuit terminates with a unitary operation $\omega_3: \HS^{T_2' E_2 Q_2'} \to \HS^{F_I f}$ given by 
\begin{align}
 \label{eq:omega_3}
    \omega_3 = \ & \big[\tilde \nu_3^{A \prec B} \cdot (\id^{E_2} \otimes \id^{T_2' \to B_O})\big] \otimes \bra{0}^{Q_2'} + \big[\tilde \nu_3^{B \prec A} \cdot (\id^{E_2} \otimes \id^{T_2' \to A_O})\big] \otimes \bra{1}^{Q_2'}
\end{align}
(with $\tilde \nu_3^{A \prec B}$ from Eq.~\eqref{eq:nu3_elg_AB} and $\tilde \nu_3^{B \prec A}$ from Eq.~\eqref{eq:nu3_elg_BA}, again implicitly considering that they each output a state in the whole space $\HS^{F_I f} = \HS^{\tilde F_I^\ell} \oplus \HS^{\tilde F_I^r}$ rather than in just one of the subspaces) that takes the target, control and ancillary systems to the input Hilbert space of the ``global future'' party and the system $\HS^f$, which is then traced out (or equivalently projected onto $\ket{0}^f$, see above).

Such a circuit is an example of a \emph{quantum circuit with quantum control of causal order}~\cite{wechs18}, and similarly to the case of fixed-order quantum circuits, the operations $\omega_1$, $\omega_2^\circ$, $\omega_2^\bullet$ and $\omega_3$ can be constructed explicitly from the process matrix. 

Let us now consider a unitary extension of a tripartite process, described by a process vector $\dket{U} \in \HS^{P_OA_{IO}B_{IO}C_{IO}F_I}$. 
When the process vector is composed with a fixed unitary operation $\dket{U_C}^{C_{IO}C_{IO}'}$ for Charlie, the result $\dket{U_\G(\cdot,\cdot,U_C)} \coloneqq \dket{U} * \dket{U_C}^{C_{IO}C_{IO}'}$ is, mathematically, a process vector describing a unitary extension of a bipartite process (for the remaining parties Alice and Bob) with a ``global past'' party whose output space is $\HS^{P_OC_I'}$ and a ``global future'' party whose input space is $\HS^{F_I C_O'}$, and which can be implemented as a ``variation of the quantum switch'', as presented just above.
For the full tripartite process, one can therefore find a temporal description consisting of a ``variation of the quantum switch'', where the circuit operations $\omega_1$, $\omega_2^\circ$, $\omega_2^\bullet$ and $\omega_3$ and the dimensions of $E_1$ and $E_2$ depend on $U_C$. 
To finally get to the tripartite circuit of Fig.~\jwchanges{3}, for later convenience, we introduce identity channels on the systems $T_1$, $T_2$, $T_1'$, $T_2'$, $Q_1$ and $Q_2'$, which relate the respective system to a copy of it, denoted by the same label with a ``bar'' superscript.

\subsection{Derivation of the isomorphisms $J_{\text{in}}$ and $J_{\text{out}}$}
\label{app:iso_tripartite}

Note that the decomposition in Eq.~\eqref{eq:decomp_bipartite} generalises straightforwardly to an arbitrary numbers of parties. 
That is, for a unitary extension of some multipartite process, it is still true that the unitary maps isomorphically the output system of a given party to a subsystem of the input systems of all other parties and the global future system, and similarly, it maps a subsystem of the output of all but one party and the global past to the input of that remaining party. 
For instance, for a unitary extension of a tripartite process, one can formally decompose its process vector $\dket{U} \in \HS^{P_OA_{IO}B_{IO}C_{IO}F_I}$ as 
\begin{equation}
    \label{eq:decomp_tripartite}
    \dket{U} = \dket{U_1} * \dket{\id}^{Z \bar Z} * \dket{U_2},
\end{equation}
with unitaries $U_1: \HS^{P_O A_O B_O} \to \HS^{C_I Z}$ and $U_2: \HS^{C_O \bar Z} \to \HS^{A_I B_I F_I}$ (cf. Fig.~\ref{fig:tripartite_decomposition}).
This follows from precisely the same argument as in the bipartite case, by considering the reduced process with one operation for Charlie, obtained by transferring Alice's and Bob's input (output) wires to the global future (past) via SWAP operations and identity channels.
\begin{figure}[h]
    \centering
    \includegraphics[width=0.6\textwidth]{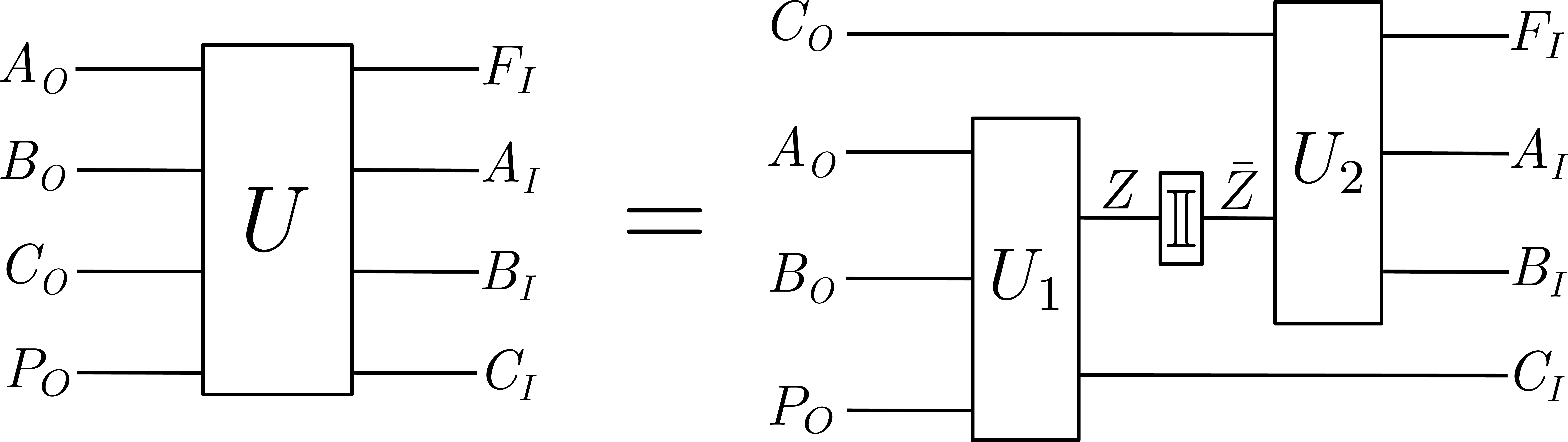}
    \caption{\jwchanges{Graphical representation of Eq.~\eqref{eq:decomp_tripartite}.} The process vector of a tripartite unitary process can be decomposed as $\dket{U} = \dket{U_1} * \dket{\id}^{Z \bar Z} * \dket{U_2}$, with unitaries $U_1: \HS^{A_O B_O P_O} \to \HS^{C_I Z}$ and $U_2: \HS^{C_O \bar Z} \to \HS^{A_I B_I F_I}$.
    The circuit on the right-hand side shows that the unitary $U$ maps some subsystem of $A_O B_O P_O$ to $C_I$, and $C_O$ to some subsystem of $A_I B_I F_I$.}
    \label{fig:tripartite_decomposition}
\end{figure}
(The fact that this abstract, mathematical decomposition generalises to the multipartite case was also noted in Ref.~\cite{guerin18b}.) 
From that decomposition, we then define the isomorphisms
$J_{\textup{in}}: \HS^{A_I B_I C_I Y Z} \to \HS^{T_1 T_2 \bar T_1' \bar T_2' Q_1 P_O}$ and $J_{\textup{out}}: \HS^{ T_1' T_2' \bar T_1 \bar T_2 Q_2' F_I} \to \HS^{A_O B_O C_O \bar Y \bar Z}$ to be
\begin{align}
J_{\textup{in}} = & \id^{A_I \to T_1} \otimes \id^{B_I \to T_2} \otimes [(\id^{A_O \to \bar T_1'} \otimes \id^{B_O \to \bar T_2'} \otimes \id^{P_O}) \, U_1^\dagger ] \otimes \ket{0}^{Q_1}\bra{0}^{Y} \notag \\
& + \id^{B_I \to T_1} \otimes \id^{A_I \to T_2} \otimes [(\id^{B_O \to \bar T_1'} \otimes \id^{A_O \to \bar T_2'} \otimes \id^{P_O}) \, U_1^\dagger ] \otimes \ket{1}^{Q_1}\bra{1}^{Y}
\label{eq:tdls_qcqc_input}
\end{align}
and
\begin{align}
J_{\textup{out}} = & \id^{T_1' \to A_O} \otimes \id^{T_2' \to B_O} \otimes [ U_2^\dagger \, (\id^{\bar T_1 \to A_I} \otimes \id^{\bar T_2 \to B_I} \otimes \id^{F_I}) ] \otimes \ket{0}^{\bar Y}\bra{0}^{Q_2'} \notag \\
& + \id^{T_1' \to B_O} \otimes \id^{T_2' \to A_O} \otimes [ U_2^\dagger \, (\id^{\bar T_1 \to B_I} \otimes \id^{\bar T_2 \to A_I} \otimes \id^{F_I}) ] \otimes \ket{1}^{\bar Y}\bra{1}^{Q_2'}. 
\label{eq:tdls_qcqc_output}
\end{align}

$J_{\textup{in}}$ can be represented graphically as in Fig.~\hyperref[fig:iso_tripartite_general]{15(a)}, and $J_{\textup{out}}$ as in 
Fig.~\hyperref[fig:iso_tripartite_general]{15(b)}.

\begin{figure}[h]
\centering
\begin{minipage}{.5\textwidth}
\flushleft{(a)}\\
  \centering
  \includegraphics[width=.48\linewidth]{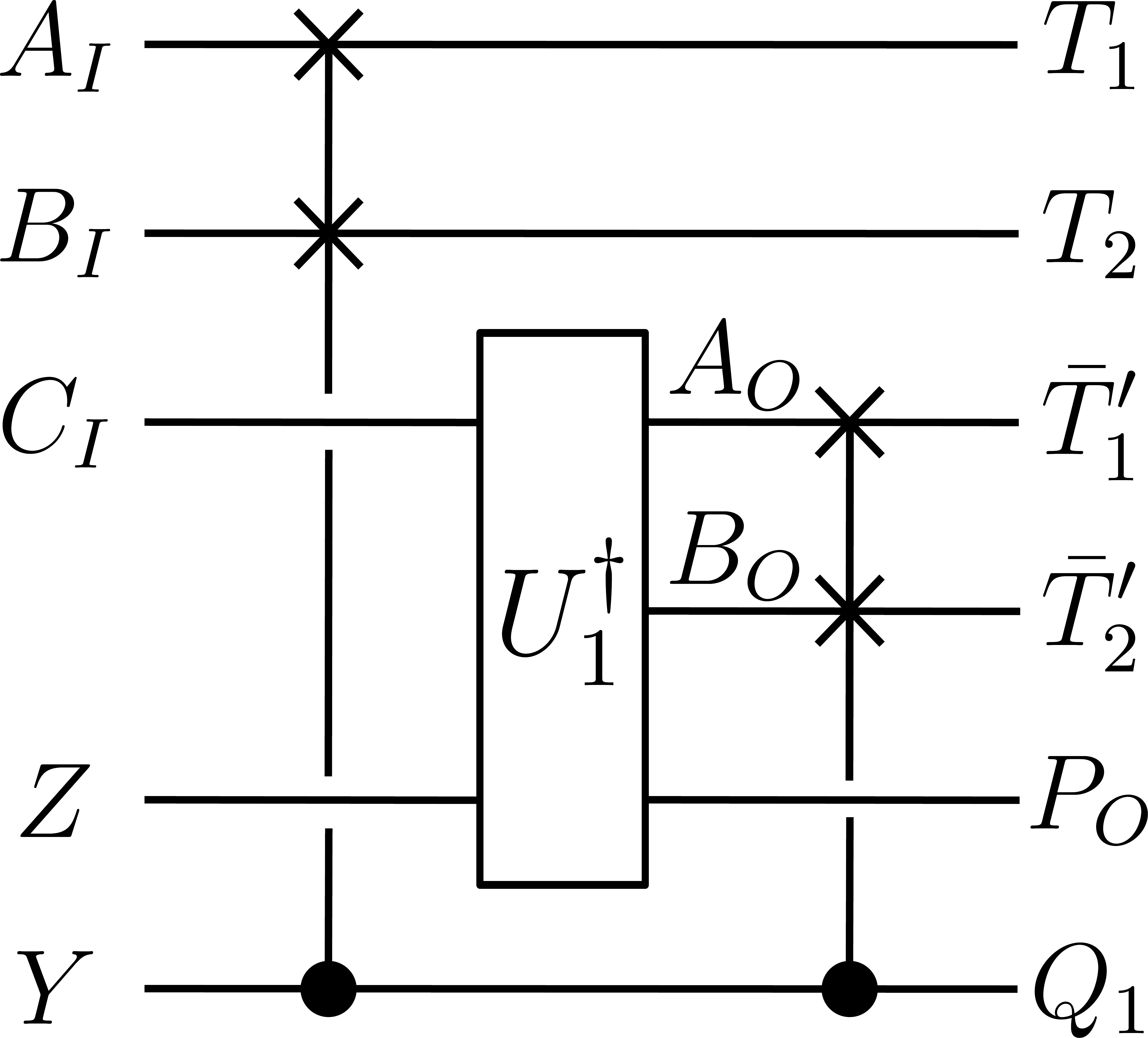}
\end{minipage}%
\begin{minipage}{.5\textwidth}
\flushleft{(b)}\\
  \centering
  \includegraphics[width=.48\linewidth]{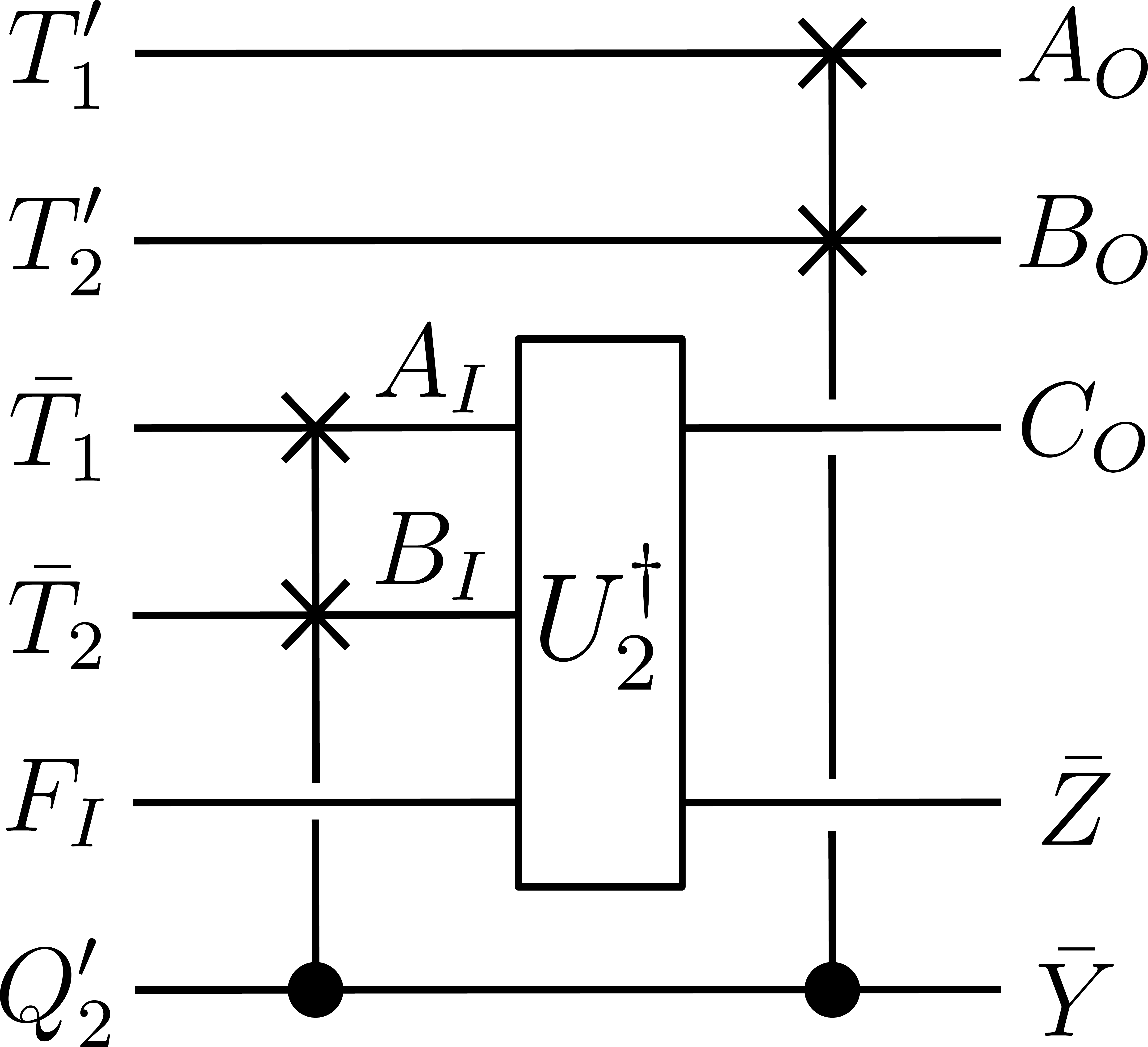}
\end{minipage}
\caption{\jwchanges{Graphical representation of the isomorphisms $J_{\textup{in}}$ and $J_{\textup{out}}$.}
(a) Graphical representation of the isomorphism $J_{\textup{in}}$, with which the incoming systems of the red fragment of Fig.~\hyperref[fig:tripartite_proof]{4(a)} (and the outgoing systems of the blue fragment in Fig.~\hyperref[fig:tripartite_proof]{4(b)}) are composed so as to move to the description in terms of time-delocalised subsystems.
Alice's input system $A_I$ is mapped to the ``target'' incoming system $T_1$, and Bob's input system $B_I$ is mapped to the ``target'' incoming system $T_2$, or vice versa, conditionally on whether the additional system $Y$ (whose computational basis state is transmitted to the control system $Q_1$) is in the state $\ket{0}$ or $\ket{1}$.
The unitary $U_1^\dagger$ from the abstract decomposition in Eq.~\eqref{eq:decomp_tripartite} (see also \jwchanges{Supplementary} Fig.~\ref{fig:tripartite_decomposition}) maps Charlie's input system $C_I$ and the additional system $Z$ to $A_O B_O P_O$, with $A_O$ then being mapped to $\bar T_1'$ and $B_O$ to $\bar{T_2'}$, or vice versa, conditionally on the state of the control system.  
\\
(b) Graphical representation of the isomorphism $J_{\textup{out}}$, with which the outgoing systems in the red fragment of Fig.~\hyperref[fig:tripartite_proof]{4(a)} (and the incoming systems of the blue fragment in Fig.~\hyperref[fig:tripartite_proof]{4(b)}) are composed.
The ``target'' outgoing system $T_1'$ is mapped to Alice's output system $A_O$, and the ``target'' outgoing system $T_2'$ to Bob's output system $B_O$, or vice versa, depending on whether the control system $Q_2'$ (whose computational basis state is transmitted to the additional system $\bar Y$) is in the state $\ket{0}$ or $\ket{1}$.
The ``target'' system $\bar T_1$ is mapped to the system $A_I$, and $\bar T_2$ to $B_I$, or vice versa, conditionally on the state of the control system, and $A_I B_I F_I$ is mapped to Charlie's output system $C_O$ and the additional system $\bar Z$ by the unitary $U_2^\dagger$ from the abstract decomposition in Eq.~\eqref{eq:decomp_tripartite}.
\label{fig:iso_tripartite_general}}
\end{figure} 

\newpage

\subsection{Description of the circuit in terms of time-delocalised subsystems}
\label{app:tripartite_systems}

To describe the circuit in terms of the chosen time-delocalised subsystems, we consider the red circuit fragment in Fig.~\hyperref[fig:tripartite_proof]{4(a)}, which implements a quantum operation $\Phi_1(U_A,U_B,U_C)$ from $\HS^{A_I' B_I' C_I' P_O T_1 T_2 \bar T_1' \bar T_2' Q_1 \bar Q_2'}$ to $\HS^{A_O' B_O' C_O' F_I \bar T_1 \bar T_2 T_1' T_2' \bar Q_1 Q_2'}$, described in the pure Choi representation by
\begin{align}
\label{eq:Alice_systems_1}
& \dket{\Phi_1(U_A,U_B,U_C)} \notag \\[1mm]
& = \dket{\omega_1^{[0]}(U_C)}^{P_O C_I' \bar T_1 E_1 \bar Q_1} * \Big(\dket{U_A^{(1)}}^{T_1 A_I'T_1' \alpha} \otimes \dket{\id}^{B_I' \beta} \otimes \ket{0}^{Q_1} \otimes \ket{0}^{Q_1'} + \dket{U_B^{(1)}}^{T_1 B_I' T_1' \beta} \otimes \dket{\id}^{A_I' \alpha} \otimes \ket{1}^{Q_1} \otimes \ket{1}^{Q_1'}\Big) \notag \\
& \hspace{25mm} * \Big(\dket{\omega_2^\circ(U_C)}^{\bar T_1' E_1 \bar T_2 E_2} \otimes \ket{0}^{Q_1'} \otimes \ket{0}^{Q_2} + \dket{\omega_2^\bullet(U_C)}^{\bar T_1' E_1 \bar T_2 E_2} \otimes \ket{1}^{Q_1'} \otimes \ket{1}^{Q_2}\Big) \notag \\
& \hspace{8mm} * \Big(\dket{U_B^{(2)}}^{T_2 \beta T_2' B_O'} \!\otimes \dket{\id}^{\alpha A_O'} \otimes \ket{0}^{Q_2} \otimes \ket{0}^{Q_2'} + \dket{U_A^{(2)}}^{T_2 \alpha T_2' A_O'} \!\otimes \dket{\id}^{\beta B_O'} \otimes \ket{1}^{Q_2} \otimes \ket{1}^{Q_2'}\Big) \!* \dket{\omega_3^{[0]}(U_C)}^{\bar T_2' E_2 \bar Q_2' F_I  C_O'} \notag \\[1mm]
& = \dket{\omega_1^{[0]}(U_C)}^{P_O C_I' \bar T_1 E_1 \bar Q_1} * \Big(\big[\dket{U_A^{(1)}}^{T_1 A_I'T_1' \alpha}\! * \dket{\id}^{\alpha A_O'}\big] \otimes \dket{\omega_2^\circ(U_C)}^{\bar T_1' E_1 \bar T_2 E_2} \otimes \big[\dket{\id}^{B_I' \beta} \!* \dket{U_B^{(2)}}^{T_2 \beta T_2' B_O'}\big] \otimes \ket{0}^{Q_1} \!\otimes\! \ket{0}^{Q_2'} \notag \\
&+\big[\dket{U_B^{(1)}}^{T_1 B_I'T_1' \beta} \!* \dket{\id}^{\beta B_O'}\big] \otimes \dket{\omega_2^\bullet(U_C)}^{\bar T_1' E_1 \bar T_2 E_2} \otimes \big[\dket{\id}^{A_I' \alpha} \!* \dket{U_A^{(2)}}^{T_2 \alpha T_2' A_O'}\big] \otimes \ket{1}^{Q_1} \!\otimes\! \ket{1}^{Q_2'}\Big) * \dket{\omega_3^{[0]}(U_C)}^{\bar T_2' E_2 \bar Q_2' F_I  C_O'}
\end{align}
with the shorthand notations $\dket{\omega_1^{[0]}(U_C)}^{P_O C_I' \bar T_1 E_1 \bar Q_1} \coloneqq \ket{0}^p * \dket{\omega_1(U_C)}^{P_O p C_I' \bar T_1 E_1 \bar Q_1}$, $\dket{U_A^{(1)}}^{T_1 A_I'T_1' \alpha} \coloneqq \dket{\id}^{T_1 A_I} * \dket{U_A}^{A_{IO} A_{IO}'} * (\dket{\id}^{A_O T_1'} \otimes \dket{\id}^{A_O' \alpha})$,  $\dket{U_B^{(1)}}^{T_1 B_I'T_1' \beta} \coloneqq \dket{\id}^{T_1 B_I} * \dket{U_B}^{B_{IO} B_{IO}'} * (\dket{\id}^{B_O T_1'} \otimes \dket{\id}^{B_O' \beta})$,  $\dket{U_A^{(2)}}^{T_2 \alpha T_2' A_O'} \coloneqq (\dket{\id}^{T_2 A_I} \otimes \dket{\id}^{\alpha A_I'}) * \dket{U_A}^{A_{IO} A_{IO}'} * \dket{\id}^{A_O T_2'}$,  $\dket{U_B^{(2)}}^{T_2 \beta T_2' B_O'} \coloneqq (\dket{\id}^{T_2 B_I} \otimes \dket{\id}^{\beta B_I'}) * \dket{U_B}^{B_{IO} B_{IO}'} * \dket{\id}^{B_O T_2'}$ and $\dket{\omega_3^{[0]}(U_C)}^{\bar T_2' E_2 \bar Q_2' F_I  C_O'} \coloneqq \dket{\omega_3(U_C)}^{\bar T_2' E_2 \bar Q_2' F_I f C_O'} * \ket{0}^f$.
In terms of the time-delocalised systems defined above, the operation implemented by the circuit fragment reads
\begin{align}
\label{eq:Alice_systems_2}
& \dket{J_{\textup{in}}} * \dket{\Phi_1(U_A,U_B,U_C)} * \dket{J_{\textup{out}}}  = \dket{U_A}^{A_{IO} A_{IO}'} \otimes \dket{U_B}^{B_{IO}B_{IO}'} \otimes \dket{R(U_C)}^{C_{IO}C_{IO}' Y \bar Y Z \bar Z \bar Q_1 \bar Q_2'} 
\end{align}
(see Fig.~\hyperref[fig:tripartite_proof]{4(a)}), where we denote by $\dket{R(U_C)}$ the pure Choi representation of the operation $R(U_C): \HS^{C_I' C_I Y Z \bar Q_2'} \to \HS^{C_O' C_O \bar Y \bar Z \bar Q_1}$ happening in parallel to $U_A$ and $U_B$, which is given by
\begin{align}
\label{eq:R_UC}
&\dket{R(U_C)}^{C_{IO}C_{IO}' Y \bar Y Z \bar Z \bar Q_1 \bar Q_2'} \notag \\
& = \big(\dket{U_1^\dagger}^{C_I Z P_O A_O B_O} * [\dket{\id}^{A_O \bar T_1'} \otimes \dket{\id}^{B_O \bar T_2'}]\big) * \dket{\omega_1^{[0]}(U_C)}^{P_O C_I' \bar T_1 E_1 \bar Q_1} * \dket{\omega_2^\circ(U_C)}^{\bar T_1' E_1 \bar T_2 E_2} * \dket{\omega_3^{[0]}(U_C)}^{\bar T_2' E_2 \bar Q_2' F_I C_O'} \notag \\
&\hspace{60mm} * \big([\dket{\id}^{\bar T_1 A_I} \otimes \dket{\id}^{\bar T_2 B_I}] * \dket{U_2^\dagger}^{A_I B_I F_I C_O \bar Z} \big) \otimes \ket{0}^{\bar Y} \otimes \ket{0}^{Y} \notag \\[1mm]
& \ \ + \big(\dket{U_1^\dagger}^{C_I Z P_O A_O B_O} * [\dket{\id}^{B_O \bar T_1'} \otimes \dket{\id}^{A_O \bar T_2'}]\big) 
* \dket{\omega_1^{[0]}(U_C)}^{P_O C_I' \bar T_1 E_1 \bar Q_1} * \dket{\omega_2^\bullet(U_C)}^{\bar T_1' E_1 \bar T_2 E_2} * \dket{\omega_3^{[0]}(U_C)}^{\bar T_2' E_2 \bar Q_2' F_I C_O'} \notag \\
&\hspace{60mm}* \big([\dket{\id}^{\bar T_1 B_I} \otimes \dket{\id}^{\bar T_2 A_I}] * \dket{U_2^\dagger}^{A_I B_I F_I C_O \bar Z} \big) \otimes \ket{1}^{\bar Y} \otimes \ket{1}^{Y}. 
\end{align}
$R(U_C)$ is represented graphically in Fig.~\ref{fig:R_UC}.

\begin{figure}[h]
    \centering
    \includegraphics[width=0.9\textwidth]{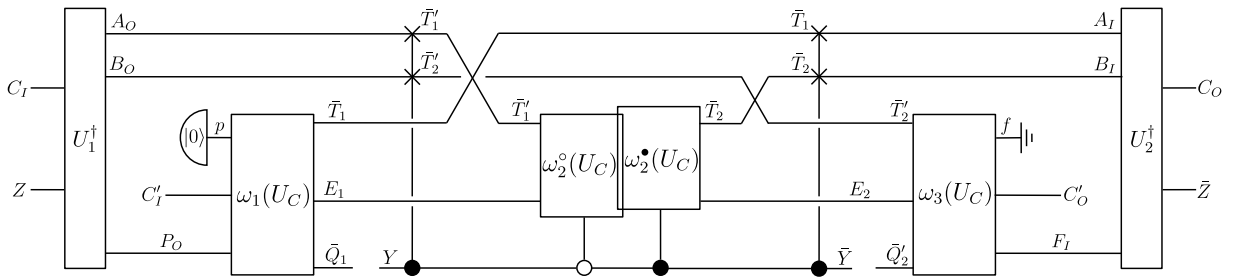}
    \caption{Circuit representation of $R(U_C)$, obtained by linking the red circuit fragment in Fig.~\hyperref[fig:tripartite_proof]{4(a)} with the circuits of $J_{\text{in}}$ and $J_{\text{out}}$ from Fig.~\ref{fig:iso_tripartite_general}, and ``factorising out'' $U_A$ and $U_B$.}
    \label{fig:R_UC}
\end{figure}

To complete the rewriting of the circuit in terms of these time-delocalised subsystems, we also need to compose the operation implemented by the complementary blue fragment, which consists of the identity channels relating the systems with and without the ``bar'' superscripts, with precisely the inverse of the isomorphisms $J_{\text{in}}$ and $J_{\text{out}}$.
We obtain
\begin{align}
\label{eq:R_prime}
\dket{R'} \coloneqq & \dket{J_{\textup{out}}^\dagger} * \big[\dket{\id}^{\bar T_1 T_1} \otimes \dket{\id}^{\bar T_2 T_2} \otimes \dket{\id}^{T_1' \bar T_1'} \otimes \dket{\id}^{T_2' \bar T_2'} \otimes \dket{\id}^{\bar Q_1 Q_1} \otimes \dket{\id}^{Q_2' \bar Q_2'}\big] * \dket{J_{\textup{in}}^\dagger} \notag\\[1mm]
= &\dket{U_2}^{C_O \bar Z A_I B_I F_I}  \otimes \dket{U_1}^{P_O A_O B_O C_I Z} \otimes \big[\ket{0}^{\bar Y} \otimes \ket{0}^{\bar Q_2'}  \otimes \ket{0}^{\bar Q_1} \otimes \ket{0}^{Y} + \ket{1}^{\bar Y} \otimes \ket{1}^{\bar Q_2'}  \otimes \ket{1}^{\bar Q_1} \otimes \ket{1}^{Y}\big] \notag \\
& + \dket{(\mathfrak{S}^{A_I B_I}\otimes\id^{F_I})U_2}^{C_O \bar Z A_I B_I F_I} \otimes \dket{U_1(\mathfrak{S}^{A_O B_O}\otimes\id^{P_O})}^{P_O A_O B_O C_I Z} \notag \\
&\otimes \big[\ket{0}^{\bar Y} \otimes \ket{0}^{\bar Q_2'}  \otimes \ket{1}^{\bar Q_1} \otimes \ket{1}^{Y} + \ket{1}^{\bar Y} \otimes \ket{1}^{\bar Q_2'}  \otimes \ket{0}^{\bar Q_1} \otimes \ket{0}^{Y}\big]
\end{align}
(see Fig.~\hyperref[fig:tripartite_proof]{4(b)}), where $\mathfrak{S}^{A_I B_I} \coloneqq \id^{A_I \to B_I} \otimes \id^{B_I \to A_I}$ and $\mathfrak{S}^{A_O B_O} \coloneqq \id^{A_O \to B_O} \otimes \id^{B_O \to A_O}$ are SWAP operations.
The operation $R'$ is represented graphically in \jwchanges{Supplementary} Fig.~\ref{fig:Rprime}.

\begin{figure}[h]
    \centering
    \includegraphics[width=0.9\textwidth]{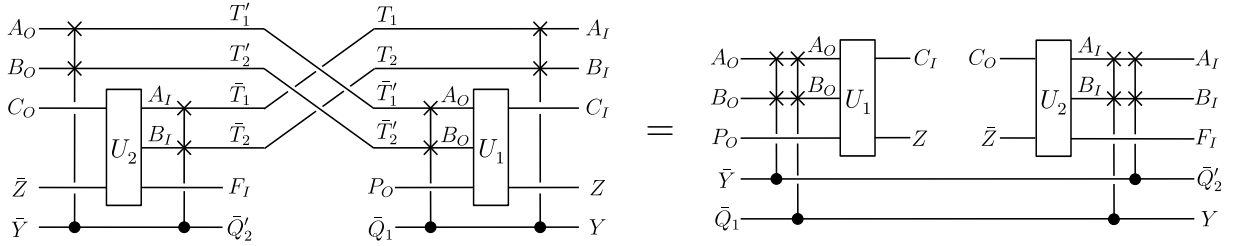}
    \caption{Circuit representation of $R'$, obtained by linking the reversed circuits of $J_{\text{in}}$ and $J_{\text{out}}$ from Fig.~\ref{fig:iso_tripartite_general} with the blue circuit fragment in Fig.~\hyperref[fig:tripartite_proof]{4(b)}. 
    One easily sees, as in Eq.~\eqref{eq:R_prime}, that if the control states in systems $\bar Y\bar Q_2'$ and $\bar Q_1Y$ are the same, then the controlled SWAP operations don't have any effect, while if the control states are different, these effectively swap the $A$ and $B$ systems. (Note that the latter situation, and the corresponding terms in Eq.~\eqref{eq:R_prime}, get cancelled when linking $R'$ to $R(U_C)$.)}
    \label{fig:Rprime}
\end{figure}

Finally, we want to recompose the operations $\dket{R(U_C)}$ and $\dket{R'}$ over the systems $Y$, $\bar Y$, $Z$, $\bar Z$, $\bar Q_1$, $\bar Q_2'$ (but not over the systems $C_I$ and $C_O$).
Inserting the explicit expressions for $\omega_1$, $\omega_2^\circ$, $\omega_2^\bullet$ and $\omega_3$ (Eqs.~\eqref{eq:omega_1},~\eqref{eq:omega_2_circ},~\eqref{eq:omega_2_bullet} and~\eqref{eq:omega_3}, respectively), taking the link product of the right-hand sides of Eq.~\eqref{eq:R_UC} and Eq.~\eqref{eq:R_prime} and evaluating the link products over $Y$, $\bar Y$, $\bar Q_1$, $\bar Q_2'$, $\bar T_1$, $\bar T_2$, $\bar T_1'$, $\bar T_2'$, yields 
\begin{align}
\label{eq:comp_R}
&\dket{R(U_C)} * \dket{R'} \notag \\
& = \dket{U_1}^{P_O A_O B_O C_I Z} 
* \Big(\dket{U_1^\dagger}^{C_I Z P_O A_O B_O} 
 * \big[\dket{\nu_1^{A \prec B}(U_C)}^{P_O C_I' A_I E_1} * \dket{\nu_2^{A \prec B}(U_C)}^{A_O E_1 B_I E_2} * \dket{\nu_3^{A \prec B}(U_C)}^{B_O E_2 F_I C_O'} \notag \\
& \ \ + \dket{\nu_1^{B \prec A}(U_C)}^{P_O C_I' B_I E_1} * \dket{\nu_2^{B \prec A}(U_C)}^{B_O E_1 A_I E_2} * \dket{\nu_3^{B \prec A}(U_C)}^{A_O E_2 F_I C_O'}\big] * \dket{U_2^\dagger}^{A_I B_I F_I C_O \bar Z} \Big) * \dket{U_2}^{C_O \bar Z A_I B_I F_I}. 
\end{align}

We then note that the term in the square bracket on the right-hand side of Eq.~\eqref{eq:comp_R} is precisely $\dket{U_\G(\cdot,\cdot,U_C)} = \dket{U} * \dket{U_C}$ (cf. Eq.~\eqref{eq:BLO_decomp}), which, with the decomposition in Eq.~\eqref{eq:decomp_tripartite}, can be written as $\dket{U_1} * \dket{\id}^{Z \bar Z} * \dket{U_2} * \dket{U_C}$.
Therefore, with Eq.~\eqref{eq:unitary_canceling_pure}, the term in the round bracket on the right-hand side of Eq.~\eqref{eq:comp_R} is precisely $\dket{U_C} * \dket{\id}^{Z \bar Z}$.
Finally, evaluating the link product over $Z$ and $\bar Z$ yields 
\begin{align}
\label{eq:tripartite_final}
\dket{R(U_C)} * \dket{R'} = & \dket{U_1}^{P_O A_O B_O C_I Z} * (\dket{U_C}^{C_{IO} C_{IO}'} \otimes \dket{\id}^{Z \bar Z}) * \dket{U_2}^{C_O \bar Z A_I B_I F_I} \notag \\
&= (\dket{U_1}^{P_O A_O B_O C_I Z} * \dket{\id}^{Z \bar Z} * \dket{U_2}^{C_O \bar Z A_I B_I F_I}) * \dket{U_C}^{C_{IO} C_{IO}'} = \dket{U}^{P_O A_{IO} B_{IO} C_{IO} F} * \dket{U_C}^{C_{IO} C_{IO}'}. 
\end{align}
(see Fig.~\hyperref[fig:tripartite_proof]{4(c)}).

\newpage

\section*{
Supplementary Note 4---Example of a process that violates causal inequalities on time-delocalised subsystems
}
\label{app:tripartite_example}
\setcounter{subsection}{0}

\subsection{Description in terms of quantum operations and time-delocalised subsystems}
\label{app:tripartite_example_quantum}

The first operation in the circuit of Fig.~\ref{fig:tripartite_example}, $\omega_1(U_C): \HS^{P_1 P_2 P_3 C_I' \to \bar T_1 E_1 \bar Q_1 \gamma}$ is given by 
\begin{align}
\label{eq:omega1_bw}
&\omega_1(U_C) = \id^{P_1 \to \bar T_1} \otimes \id^{P_2 \to E_1} \otimes [(\ket{0}^{\bar Q_1}\bra{0}^{C_O} \otimes \id^{C_O' \to \gamma}) \cdot U_C \cdot (\id^{P_3 \to C_I} \otimes \id^{C_I'})] \notag \\
&\hspace{45mm} + \id^{P_1 \to E_1} \otimes \id^{P_2 \to \bar T_1} \otimes [(\ket{1}^{\bar Q_1}\bra{1}^{C_O} \otimes \id^{C_O' \to \gamma}) \cdot U_C \cdot (\id^{P_3 \to C_I} \otimes \id^{C_I'})], 
\end{align}
where $\gamma$ is an ancillary space with dimension $d_\gamma = d_{C_I'}$.
The two circuit operations $\omega_2^\circ: \HS^{\bar T_1' E_1} \to \HS^{\bar T_2 E_2}$ and $\omega_2^\bullet: \HS^{\bar T_1' E_1} \to \HS^{\bar T_2 E_2}$ do not depend on $U_C$ for this particular process. 
They are
\begin{align}
\omega_2^\circ = \ket{0}^{E_2}\bra{0}^{\bar T_1'} \otimes \id^{E_1 \to \bar T_2} + \ket{1}^{E_2}\bra{1}^{\bar T_1'} \otimes \sigma_\textsc{x}^{E_1 \to \bar T_2}
\end{align}
(where $\sigma_\textsc{x}^{Y \to Z} \coloneqq \ket{0}^Z\bra{1}^Y + \ket{1}^Z\bra{0}^Y$ denotes a NOT gate) and
\begin{align}
\omega_2^\bullet = \ket{1}^{E_2}\bra{1}^{\bar T_1'} \otimes \id^{E_1 \to \bar T_2} + \ket{0}^{E_2}\bra{0}^{\bar T_1'} \otimes \sigma_\textsc{x}^{E_1 \to \bar T_2}.
\end{align}

The final circuit operation $\omega_3(U_C): \HS^{\bar T_2' E_2 \bar Q_2' \gamma} \to \HS^{F_1 F_2 F_3 C_O'}$ is
\begin{align}
\label{eq:omega3_bw}
\omega_3(U_C) &= \big[\ket{000}^{ F_1 F_2 F_3}\bra{000}^{\bar T_2' E_2 \bar Q_2'} + \ket{001}^{F_1 F_2 F_3}\bra{001}^{\bar T_2' E_2 \bar Q_2'} + \ket{100}^{F_1 F_2  F_3}\bra{010}^{\bar T_2' E_2 \bar Q_2'} \notag \\
& +\ket{101}^{F_1 F_2 F_3}\bra{101}^{\bar T_2' E_2 \bar Q_2'} + \ket{110}^{F_1  F_2 F_3}\bra{110}^{\bar T_2' E_2 \bar Q_2'} +\ket{111}^{F_1 F_2 F_3}\bra{111}^{\bar T_2' E_2 \bar Q_2'}\big] \otimes \id^{\gamma \to C_O'}\notag \\
& + \ket{01}^{F_1 F_2 }\bra{100}^{\bar T_2' E_2 \bar Q_2'} \otimes \big[\id^{C_O \to F_3} \otimes \id^{C_O'})\cdot (U_C \cdot (\sigma_\textsc{x}^{C_I \to C_I} \otimes \id^{C_I'}) \cdot U_C^\dagger)\cdot (\ket{0}^{C_O} \otimes \id^{\gamma \to C_O'}\big] \notag \\
& + \ket{01}^{F_1 F_2 }\bra{011}^{\bar T_2' E_2 \bar Q_2'} \otimes \big[\id^{C_O \to F_3} \otimes \id^{C_O'})\cdot (U_C \cdot (\sigma_\textsc{x}^{C_I \to C_I} \otimes \id^{C_I'}) \cdot U_C^\dagger)\cdot (\ket{1}^{C_O} \otimes \id^{\gamma \to C_O'}\big]. 
\end{align}

For this particular process, a decomposition as in Eq.~\eqref{eq:decomp_tripartite} is given by 
$\dket{U_{\text{BW}}} = \dket{U_1} * \dket{\id}^{Z \bar Z} * \dket{U_2}$, with the unitaries $U_1: \HS^{P_1 P_2 P_3 A_O B_O} \to \HS^{C_I Z}$ and $U_2: \HS^{C_O \bar Z} \to \HS^{A_I B_I F_1 F_2 F_3}$
\begin{align}
\label{eq:charlieSystems}
U_1 &= \sum_{p_1 p_2 p_3 a_O b_O} \ket{p_3 \oplus \lnot \, a_O \land b_O}^{C_I}\ket{p_1,p_2,a_O,b_O}^{Z}\bra{p_1,p_2,p_3}^{P_1 P_2 P_3}\bra{a_O,b_O}^{A_O B_O} \notag\\
U_2 &= \sum_{a_O b_O c_O p_1 p_2}  \ket{p_1 \oplus \lnot \, b_O \land c_O, p_2 \oplus \lnot \, c_O \land a_O}^{A_I B_I}\ket{a_O,b_O,c_O}^{F_1 F_2 F_3} \bra{c_O}^{C_O}\bra{p_1,p_2,a_O,b_O}^{\bar Z},  
\end{align}
from which the isomorphisms $J_{\textup{in}}: \HS^{A_I B_I C_I Y Z} \to \HS^{T_1 T_2 \bar T_1' \bar T_2' Q_1 P_O}$ and $J_{\textup{out}}: \HS^{ T_1' T_2' \bar T_1 \bar T_2 Q_2' F_I} \to \HS^{A_O B_O C_O \bar Y \bar Z}$ that define the decomposition into time-delocalised subsystems are obtained through Eqs.~\eqref{eq:tdls_qcqc_input} and~\eqref{eq:tdls_qcqc_output}.

Specifically, these are given by
\begin{align}
\label{eq:iso_incoming_bw}
J_{\textup{in}} = & \id^{A_I \to T_1} \otimes \id^{B_I \to T_2} \otimes [ \sum_{\substack{p_1 p_2 p_3\\ a_O b_O}} \ket{p_1,p_2,p_3}^{P_1 P_2 P_3}\ket{a_O,b_O}^{\bar T_1' \bar T_2'} \bra{p_3 \oplus \lnot \, a_O \land b_O}^{C_I}\bra{p_1,p_2,a_O,b_O}^{Z}] \otimes \ket{0}^{Q_1}\bra{0}^{Y} \notag \\
 &+ \id^{B_I \to T_1} \otimes \id^{A_I \to T_2} \otimes [\sum_{\substack{p_1 p_2 p_3\\ a_O b_O}} \ket{p_1,p_2,p_3}^{P_1 P_2 P_3}\ket{a_O,b_O}^{\bar T_2' \bar T_1'} \bra{p_3 \oplus \lnot \, a_O \land b_O}^{C_I}\bra{p_1,p_2,a_O,b_O}^{Z}] \otimes \ket{1}^{Q_1}\bra{1}^{Y},
\end{align}
and
\begin{align}
\label{eq:iso_outgoing_bw}
J_{\textup{out}} = & \id^{T_1' \to A_O} \otimes \id^{T_2' \to B_O} \notag \\ &\hspace{10mm}\otimes [ \sum_{\substack{a_O b_O c_O\\ p_1 p_2}} \ket{c_O}^{C_O}\ket{p_1,p_2,a_O,b_O}^{\bar Z} \bra{p_1 \oplus \lnot \, b_O \land c_O, p_2 \oplus \lnot \, c_O \land a_O}^{\bar T_1 \bar T_2}\bra{a_O,b_O,c_O}^{F_1 F_2 F_3} ] \otimes \ket{0}^{\bar Y}\bra{0}^{Q_2'} \notag \\
&+ \id^{T_1' \to B_O} \otimes \id^{T_2' \to A_O} \notag \\
&\hspace{10mm}\otimes [ \sum_{\substack{a_O b_O c_O\\ p_1 p_2}} \ket{c_O}^{C_O}\ket{p_1,p_2,a_O,b_O}^{\bar Z} \bra{p_1 \oplus \lnot \, b_O \land c_O, p_2 \oplus \lnot \, c_O \land a_O}^{\bar T_2 \bar T_1}\bra{a_O,b_O,c_O}^{F_1 F_2 F_3} ] \otimes \ket{1}^{\bar Y}\bra{1}^{Q_2'}. 
\end{align}
These isomorphisms are represented graphically in Fig.~\ref{fig:iso_bw}.

\begin{figure}[H]
\centering
\begin{minipage}{.5\textwidth}
\flushleft{(a)}\\
  \centering
  \includegraphics[width=.48\linewidth]{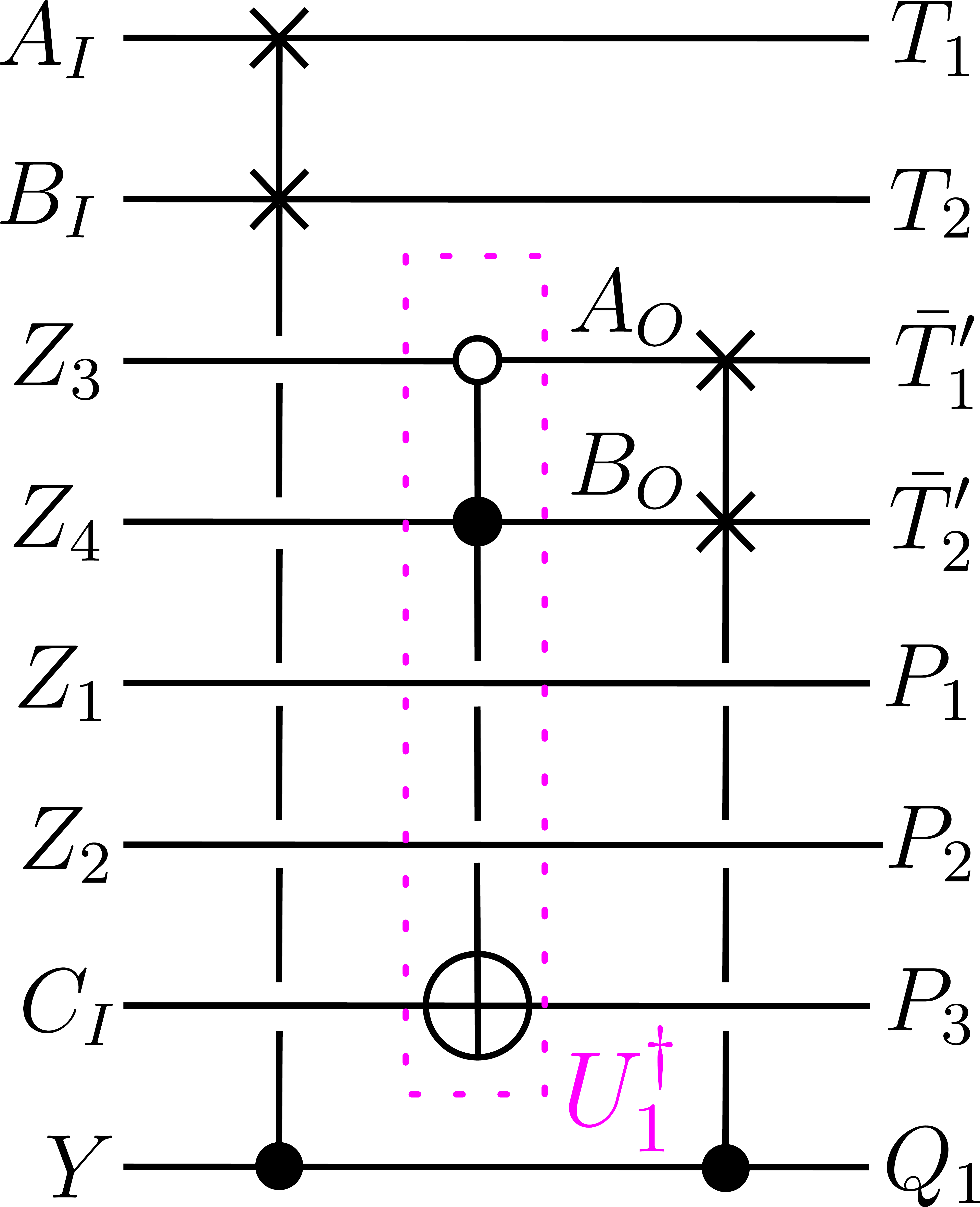}
\end{minipage}%
\begin{minipage}{.5\textwidth}
\flushleft{(b)}\\
  \centering
  \includegraphics[width=.48\linewidth]{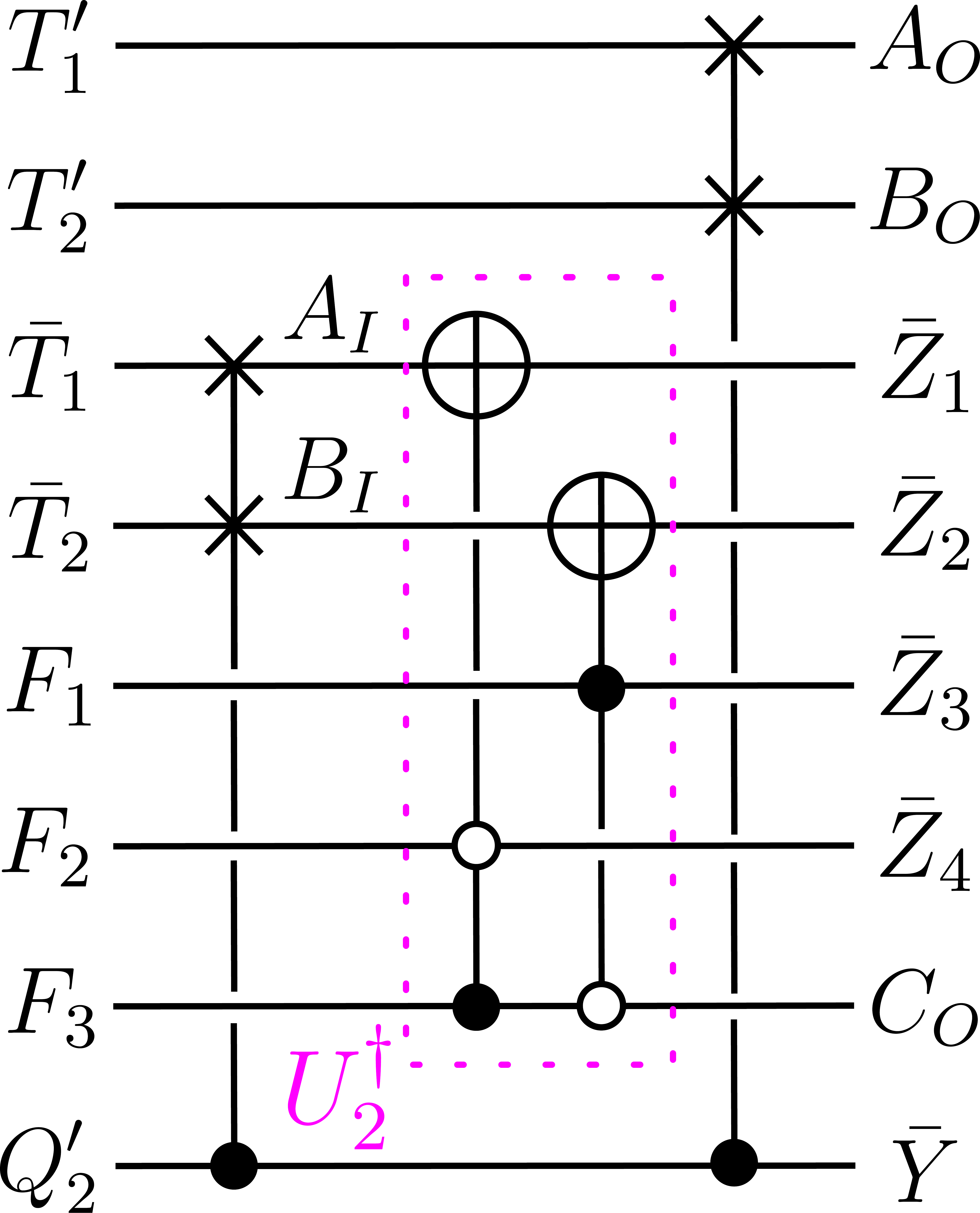}
\end{minipage}
\caption{\jwchanges{Graphical representation of the isomorphisms $J_{\textup{in}}$ and $J_{\textup{out}}$ for the BW process.}
(a) Isomorphism $J_\textup{in}$ that defines the time-delocalised subsystems $A_I$, $B_I$, and $C_I$ for the particular example of the BW process from Eq.~\eqref{eq:pm_bw} and Fig.~\ref{fig:tripartite_example} (with the four qubits constituting $\HS^Z \coloneqq \HS^{Z_1} \otimes \HS^{Z_2} \otimes \HS^{Z_3} \otimes \HS^{Z_4}$ shown as separate systems). 
\\
(b) Isomorphism $J_\textup{out}$ that defines the time-delocalised subsystems $A_O$, $B_O$, and $C_O$ (with again the four qubits constituting $\HS^{\bar Z} \coloneqq \HS^{\bar Z_1} \otimes \HS^{\bar Z_2} \otimes \HS^{\bar Z_3} \otimes \HS^{\bar Z_4}$ shown as separate systems).
  \label{fig:iso_bw}
}
\end{figure}

One can check that by applying the general tripartite ``circuit rewriting'' procedure (represented graphically in Fig.~\ref{fig:tripartite_proof} and detailled mathematically in Supplementary Note~\hyperref[app:tripartite_systems]{3~C}) to the particular temporal circuit of Fig.~\ref{fig:tripartite_example}, with the specific circuit operations $\omega_1$, $\omega_2^\circ$, $\omega_2^\bullet$, $\omega_3$ in Eqs.~\eqref{eq:omega1_bw}--\eqref{eq:omega3_bw} and the specific isomorphisms $J_{\text{in}}$ and $J_{\text{out}}$ defined in Eqs.~\eqref{eq:iso_incoming_bw} and~\eqref{eq:iso_outgoing_bw}, one indeed ends up with the process vector $\dket{U_{\text{BW}}}$ in Eq.~\eqref{eq:pm_bw}, composed with $\dket{U_A}$, $\dket{U_B}$ and $\dket{U_C}$. 
(And where, in the calculation, one replaces $\dket{\omega_1^{[0]}(U_C)}^{P_O C_I' \bar T_1 E_1 \bar Q_1}$ by $\dket{\omega_1(U_C)}^{P_1 P_2 P_3 C_I' \bar T_1 E_1 \bar Q_1 \gamma}$, and $\dket{\omega_3^{[0]}(U_C)}^{\bar T_2' E_2 \bar Q_2' F_I  C_O'}$ by  $\dket{\omega_3(U_C)}^{\bar T_2' E_2 \bar Q_2' \gamma F_1 F_2 F_3 C_O'}$, in order to account for the simplifications we made in the treatment of the circuit ancillas, see Fig.~\ref{fig:tripartite_example}).
The operations $R(U_C)$ and $R'$ for this specific example, as well as their composition, are shown in Figs.~\ref{fig:R_UC_BW} to~\ref{fig:R_Rprime_BW}.

\begin{figure}[h]
    \centering
    \includegraphics[width=0.9\textwidth]{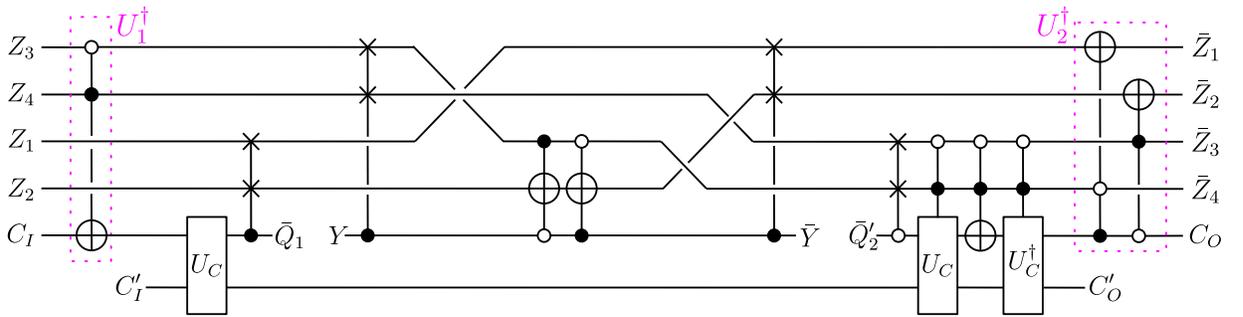}
    \caption{Circuit representation of $R(U_C)$ for the BW process. }
    \label{fig:R_UC_BW}
\end{figure}

\begin{figure}[h]
    \centering
    \includegraphics[width=0.6\textwidth]{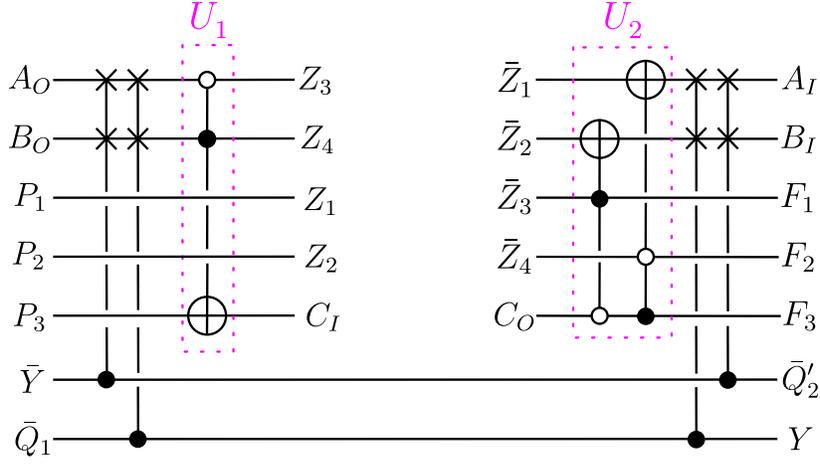}
    \caption{Circuit representation of $R'$ for the BW process. }
    \label{fig:Rprime_BW}
\end{figure}

\begin{figure}[h]
    \centering
    \includegraphics[width=0.9\textwidth]{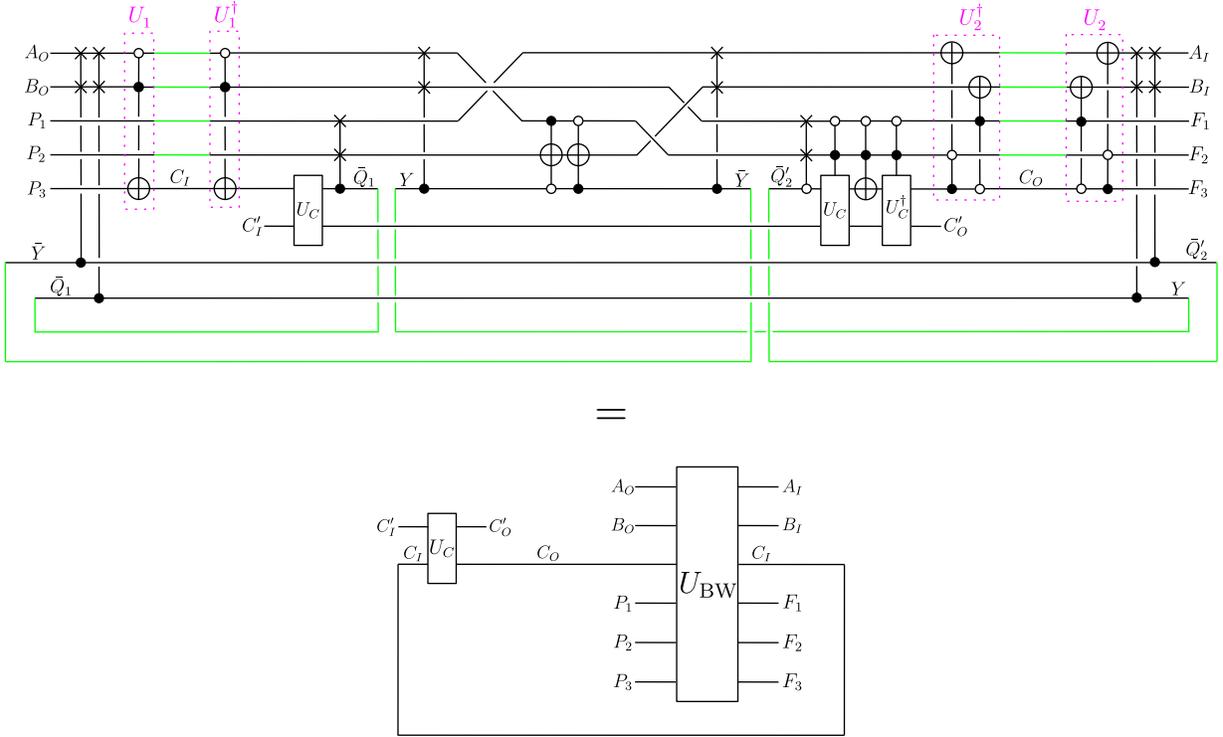}
    \caption{Circuit representation of the composition of $R(U_C)$ and $R'$ for the BW process. Here we see in particular how the time-delocalised systems (or classical variables, here) $C_I, C_O$ are identified. To verify that we indeed get the composition of $U_{\text{BW}}$ with $U_C$, we should evaluate the composition over $C_I, C_O$, as we do on the next figure.
    }
    \label{fig:R_Rprime_BW_partial}
\end{figure}

\begin{figure}[H]
    \centering
    \includegraphics[width=0.7\textwidth]{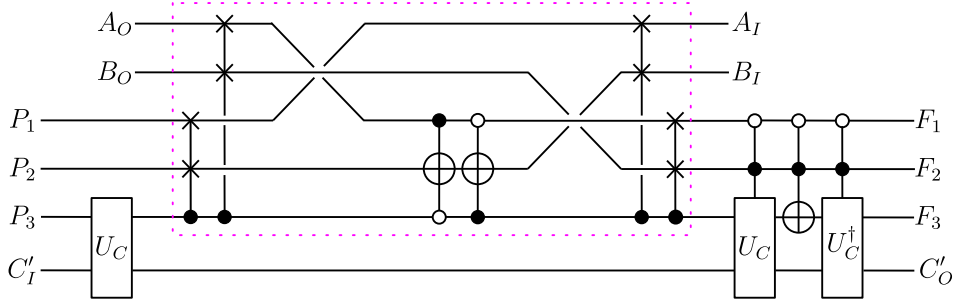}
    \caption{Circuit representation of the composition of $R(U_C)$ and $R'$ for the BW process, after simplification of $U_1^\dagger U_1$ and $U_2 U_2^\dagger$ and other further simplifications in the previous figure. It is easily checked that the circuit fragment in the dashed box realises the unitary operation $\sum_{a_O,b_O,c_O,p_1,p_2} \ketbra{p_1 \oplus \neg b_O \land c_O, p_2 \oplus \neg c_O \land a_O,a_O,b_O,c_O}{a_O,b_O,p_1,p_2,c_O}$, so that the whole fragment shown above realises $\sum_{a_O,b_O,c_O,p_1,p_2} \ketbra{p_1 \oplus \neg b_O \land c_O, p_2 \oplus \neg c_O \land a_O,a_O,b_O}{a_O,b_O,p_1,p_2}\otimes [U_C(\sigma_\textsc{x}\otimes\id)U_C^\dagger]^{\neg a_O \land b_O}(\ketbra{c_O}{c_O}\otimes\id)U_C = \sum_{a_O,b_O,c_O,p_1,p_2,p_3} \ketbra{p_1 \oplus \neg b_O \land c_O, p_2 \oplus \neg c_O \land a_O,a_O,b_O}{a_O,b_O,p_1,p_2}\otimes (\ketbra{c_O}{c_O}\otimes\id)U_C(\ketbra{p_3 \oplus \neg a_O \land b_O}{p_3}\otimes\id)$ (where the equality is obtained rather trivially for the terms with $\neg a_O \land b_O = 0$; for the terms with $\neg a_O \land b_O = 1$, note that $c_O$ only appears in $\sum_{c_O} \ketbra{c_O}{c_O} = \id$; we can then simplify this partial sum, simplify $U_C^\dagger$ together with $U_C$, re-introduce a similar $\sum_{c_O} \ketbra{c_O}{c_O}$ to the left of the remaining $U_C$, and finally write $\sigma_\textsc{x}^{\neg a_O \land b_O} = \sum_{p_3} \ketbra{p_3 \oplus \neg a_O \land b_O}{p_3}$). From this expression it is then easy to verify that the above circuit indeed realises precisely the composition of the process $U_{\text{BW}}$, as defined from Eq.~\eqref{eq:pm_bw}, with $U_C$ (see Fig.~\hyperref[fig:tripartite_proof]{4(c)}). 
    }
    \label{fig:R_Rprime_BW}
\end{figure}


\subsection{Description in terms of classical operations and time-delocalised variables}
\label{app:tripartite_example_classical}

In this Supplementary Note, we explain in more detail that the example in \hyperref[sec:td_bw]{``A process that violates causal inequalities on time-delocalised subsystems''} of the main text describes a \emph{classical} noncausal process, with classical operations that take place on \emph{time-delocalised variables}.
The classical counterpart of a quantum system $X$ with a Hilbert space $\HS^{X}$ is a random variable $X$ with values in a set $\mathcal{S}^X \coloneqq \{0, \ldots, d_X - 1\}$. 
The classical counterpart of a quantum operation (i.e., most generally, of a quantum instrument) from an incoming quantum system $X$ to an outgoing quantum system $Y$ is a ``classical instrument'', that is, a conditional probability distribution $P(r,y|x)$ which specifies the probability that the outgoing variable $Y$ takes the value $y \in \mathcal{S}^Y$ and the ``outcome'' of the operation (which can also be described by a random variable) is $r$, given that the incoming variable $X$ has the value $x \in \mathcal{S}^X$.

If the outgoing variable $Y$ is then passed on to a subsequent classical operation (say, an operation from incoming random variables $S$ and $Y$ to outgoing random variables $T$ and $Z$, specified by a conditional probability distribution $P(t,z|s,y)$), the two operations compose to a new classical operation $P(r,t,z|x,s) = \sum_y P(r,y|x) P(t,z|s,y)$. 
This is the classical counterpart of the link product.

If one then has an acyclic network composed of classical operations, it can be described in terms of \emph{time-delocalised variables} by decomposing it into fragments and composing these with deterministic reversible operations that take the incoming and outgoing ``time-local'' variables to new ones, analogously to the quantum case.
For instance, if the operations in the circuit of Fig.~\ref{fig:circuit_generic} are all classical, the red fragment defines a classical operation $P_{\text{red}}(k,d,h,i|a,f)$ from incoming time-local variables $A$, $F$ to outgoing time-local variables $D$, $H$, $I$ and with an outcome $k$, and the blue circuit fragment an operation $P_{\text{blue}}(j,l,m,a,f|d,h,i)$, from the incoming variables $D$, $H$, $I$ to the outgoing variables $A$, $F$ and with outcomes $j$, $l$, $m$.
We can then change to a description in terms of time-delocalised variables $V$, $W$, $X$, $Y$ defined by any bijective functions $J_{\text{in}}: \mathcal{S}^V \times \mathcal{S}^W \to \mathcal{S}^A \times \mathcal{S}^F$ and $J_{\text{out}}: \mathcal{S}^D \times \mathcal{S}^H \times \mathcal{S}^I \to \mathcal{S}^X \times \mathcal{S}^Y$, by taking
\begin{equation}
P'_{\text{red}}(k,x,y|v,w) = \sum_{a,d,f,h,i} P_{\text{red}}(k,d,h,i|a,f) \delta_{(a,f),J_{\text{in}}(v,w)} \delta_{(x,y),J_{\text{out}}(d,h,i)}
\end{equation}
and 
\begin{equation}
P'_{\text{blue}}(j,l,m,v,w|x,y) = \sum_{a,d,f,h,i} P_{\text{blue}}(j,l,m,a,f|d,h,i) \delta_{(d,h,i),J_{\text{out}}^{-1}(x,y)} \delta_{(v,w),J_{\text{in}}^{-1}(a,f)}.
\end{equation}
where $\delta_{(...),(...)}$ denotes the Kronecker delta between tuples.
It is straightforward to check that, analogously to the quantum case, the composition of the two fragments over the time-local and time-delocalised variables is indeed the same---i.e., that $\sum_{v,w,x,y} P'_{\text{red}}(k,x,y|v,w) P'_{\text{blue}}(j,l,m,v,w|x,y) = \sum_{a,d,f,h,i} P_{\text{red}}(k,d,h,i|a,f) P_{\text{blue}}(j,l,m,a,f|d,h,i) = P(j,k,l,m)$. 

If the state evolving through a quantum circuit is diagonal in the computational basis at any (relevant) time, the situation is effectively classical (i.e., when one identifies each quantum system $X$ with a classical, random variable $X$ and the computational basis states $\ket{i}^X$ of $\HS^X$ with the elements of $\mathcal{S}^X$, one obtains a probability distribution evolving through a circuit consisting of classical operations in the above sense). 
In the circuit of Fig.~\ref{fig:tripartite_example} for instance, when an input state in the incoming systems $\HS^{P_1 P_2 P_3 A_I' B_I' C_I'}$ that is diagonal in the computational basis is prepared, and when $U_A$, $U_B$ and $U_C$ are restricted to unitaries that map computational basis states to computational basis states, we indeed have such an effectively classical circuit composed of deterministic operations.

The quantum isomorphisms $J_{\text{in}}: \HS^{A_I B_I C_I Y Z} \to \HS^{T_1 T_2 \bar T_1' \bar T_2' Q_1 P_O}$ and $J_{\textup{out}}: \HS^{ T_1' T_2' \bar T_1 \bar T_2 Q_2' F_I} \to \HS^{A_O B_O C_O \bar Y \bar Z}$ given in Supplementary Note~\hyperref[app:tripartite_example_quantum]{4~A} above, which also map computational basis states to computational basis states, then translate into bijective functions $J_{\text{in}}: \mathcal{S}^{A_I} \times \mathcal{S}^{B_I} \times \mathcal{S}^{C_I} \times \mathcal{S}^{Y} \times \mathcal{S}^{Z} \to \mathcal{S}^{T_1} \times \mathcal{S}^{T_2} \times \mathcal{S}^{\bar T_1'} \times \mathcal{S}^{\bar T_2'} \times \mathcal{S}^{Q_1} \times \mathcal{S}^{P_O}$ and $J_{\text{out}}: \mathcal{S}^{T_1'} \times \mathcal{S}^{T_2'} \times \mathcal{S}^{\bar T_1} \times \mathcal{S}^{\bar T_2} \times \mathcal{S}^{Q_2'} \times \mathcal{S}^{F_I} \to \mathcal{S}^{A_O} \times \mathcal{S}^{B_O} \times \mathcal{S}^{C_O} \times \mathcal{S}^{\bar Y} \times \mathcal{S}^{\bar Z}$ that define the time-delocalised variables $A_I$, $B_I$, $C_I$, $A_O$, $B_O$, $C_O$. Namely, we obtain $A_I = (\neg Q_1 \land T_1) \oplus (Q_1 \land T_2)$, $B_I = (\neg Q_1 \land T_2) \oplus (Q_1 \land T_1)$, $C_I = P_3 \oplus (\neg Q_1 \land \neg \bar T_1' \land \bar T_2') \oplus (Q_1 \land \bar T_1' \land \neg \bar T_2')$, $A_O = (\neg Q_2' \land T_1') \oplus (Q_2' \land T_2')$, $B_O = (\neg Q_2' \land T_2') \oplus (Q_2' \land T_1')$ and $C_O = F_3$ (see Fig.~\ref{fig:iso_bw}).
With respect to these variables, the circuit of Fig.~\ref{fig:tripartite_example} then corresponds to three classical local operations $P_A(a_O,a_O'|a_I,a_I')$, $P_B(b_O,b_O'|b_I,b_I')$, $P_C(c_O,c_O'|c_I,c_I')$, that are composed with a classical, deterministic channel $P_{\text{BW}}(a_I,b_I,c_I,f_1,f_2,f_3|a_O,b_O,c_O,p_1,p_2,p_3) = \delta_{a_I,p_1 \oplus \neg b_O \land c_O} \delta_{b_I,p_2 \oplus \neg c_O \land a_O} \delta_{c_I,p_3 \oplus \neg a_O \land b_O}\delta_{f_1,a_O}\delta_{f_2,b_O}\delta_{f_3,c_O}$ that sends the outputs $A_O$, $B_O$, $C_O$ of the local operations, as well as the outputs $P_1$, $P_2$ and $P_3$ of the ``global past'' party, back to their inputs $A_I$, $B_I$, $C_I$ and the inputs $F_1$, $F_2$ and $F_3$ of the ``global future'' party. When the classical input state $\delta_{0,p_1}\delta_{0,p_2}\delta_{0,p_3}$ is prepared by the ``global past'' party, and the input variables of the ``global future'' party are discarded, the operation realised by the circuit can be written as
\begin{equation} \label{eq:Poooiii}
    P(a_O',b_O',c_O'|a_I',b_I',c_I') = \sum_{\substack{a_I, b_I, c_I,\\ a_O, b_O, c_O}} P_{\text{AF}}(a_I,b_I,c_I|a_O,b_O,c_O)P_A(a_O,a_O'|a_I,a_I')P_B(b_O,b_O'|b_I,b_I')P_C(c_O,c_O'|c_I,c_I'),
\end{equation} i.e., the circuit corresponds to $P_A(a_O,a_O'|a_I,a_I')$, $P_B(b_O,b_O'|b_I,b_I')$, $P_C(c_O,c_O'|c_I,c_I')$, composed with a deterministic, classical channel $P_{\text{AF}}(a_I,b_I,c_I|a_O,b_O,c_O) = \delta_{a_I,\neg b_O \land c_O} \delta_{b_I,\neg c_O \land a_O} \delta_{c_I,\neg a_O \land b_O}$ (which indeed corresponds to the ``classical process matrix'' $W_{\text{AF}}$ of Eq.~\eqref{eq:pm_bw_trivialPF}). 
In this classical description, for the operations that the parties apply to violate causal inequalities (see \hyperref[sec:td_bw]{``A process that violates causal inequalities on time-delocalised subsystems''} in the main text), we may identify the incoming ancillary variables $A_I'$, $B_I'$ and $C_I'$ with the classical variables that describe the local classical input variables, or ``settings'', that the parties receive (denoted by $I_A$, $I_B$ and $I_C$, respectively, with values $i_A$, $i_B$ and $i_C$), and their outgoing ancillary variables $A_O'$, $B_O'$ and $C_O'$ with their classical outputs (described by classical variables $O_A$, $O_B$ and $O_C$, respectively, with values $o_A$, $o_B$ and $o_C$).
The operations are then given by $P(a_O,o_A|a_I,i_A) = \delta_{o_A,a_I}\delta_{a_O,i_A}$, $P(b_O,o_B|b_I,i_B) = \delta_{o_B,b_I}\delta_{b_O,i_B}$ and $P(c_O,o_C|c_I,i_C) = \delta_{o_C,c_I}\delta_{c_O,i_C}$, and through Eq.~\eqref{eq:Poooiii} we obtain the correlation $P(o_A,o_B,o_C|i_A,i_B,i_C) = \delta_{o_A,\neg i_B \land i_C}\delta_{o_B,\neg i_C \land i_A}\delta_{o_C,\neg i_A \land i_B}$, which violates causal inequalities. 


\section*{Supplementary Note 5---Causal inequality assumptions imply causal correlations}
\label{app:causal_correlations}

In this section, we are going to prove that the conditions~\eqref{eq:sum_causalStructure}--\eqref{eq:closedLab3} from the section \hyperref[app:methods_CI]{``Causal inequality assumptions''} imply that $P(o_A,o_B,o_C|i_A,i_B,i_C)$ must respect causal inequalities.
For most generality, we consider the general multipartite case, that is, we consider a set of variables $\Gamma \coloneqq \{X_I,X_O,I_X,O_X\}_{X = A,B,C,\ldots}$ for an arbitrary number of parties $\{A,B,C,\ldots\}$, and a probability distribution $P(o_A,o_B,o_C,\ldots,\kappa(\Gamma)|i_A,i_B,i_C,\ldots)$, with $\kappa(\Gamma)$ the possible SPOs on the set $\Gamma$, that satisfies the straightforward multipartite generalisation of the constraints~\eqref{eq:sum_causalStructure}--\eqref{eq:closedLab3}.
As in the main text, when we write probabilities for some constraint on the causal order, this is to be understood as the probability obtained by summing over all $\kappa(\Gamma)$ that satisfy the respective constraint.
(For instance, $P(I_{\Lambda^{(0)}} \prec \Lambda^{(0)}_I|i_A,i_B,i_C,\ldots)$ below is obtained by summing $P(o_A,o_B,o_C,\ldots,\kappa(\Gamma)|i_A,i_B,i_C,\ldots)$ over all $\kappa(\Gamma)$ that satisfy $I_{\Lambda^{(0)}} \prec \Lambda^{(0)}_I$, as well as over all outcomes.)

The idea of the proof is that, to each $\kappa(\Gamma)$, we associate a ``coarse-grained'' SPO $\C_{\kappa(\Gamma)}$ on the set of parties $\{A,B,C,\ldots \}$, and show that, with respect to this coarse-grained SPO, 
the correlation $P(o_A,o_B,o_C,\ldots|i_A,i_B,i_C,\ldots)$ satisfies the definition of ``causal correlation'' as introduced in Ref.~\cite{oreshkov16} (also referred to as ``causal process'' in Ref.~\cite{oreshkov16}), which means that it must respect causal inequalities.
Let us first recall this definition.

\medskip

\textbf{Causal correlation (Definition from Ref.~\cite{oreshkov16}).} \textit{A correlation $P(o_A,o_B,o_C,\ldots|i_A,i_B,i_C,\ldots)$ for a set of ``local experiments'' $\{A,B,C,\ldots\} \eqqcolon \Delta$, is \emph{causal} if and only if there exists a probability distribution 
\begin{equation}
P(o_A,o_B,o_C,\ldots,\xi(\Delta)|i_A,i_B,i_C,\ldots),
\end{equation}
where $\xi(\Delta)$ are the possible SPOs on the set $\Delta$, such that 
     \begin{equation}
     \label{eq:causal_process_def_1}
        \sum_{\xi(\Delta)} P(o_A,o_B,o_C,\ldots,\xi(\Delta)|i_A,i_B,i_C,\ldots) = P(o_A,o_B,o_C,\ldots|i_A,i_B,i_C, \ldots) 
    \end{equation}
and such that, for every local experiment (e.g., $A$), for every subset $\X$ of the rest of the local experiments, and for every SPO $\xi(\{A\} \cup \X)$ on the local experiment in question and that subset,
\begin{equation}
\label{eq:causal_process_def_2}
    P(o^\X, A \npreceq \X, \xi(\{A\} \cup \X)|i_A,i_B,i_C,\ldots) =  P(o^\X, A \npreceq \X, \xi(\{A\} \cup \X)|i_B,i_C,\ldots).
\end{equation}
}

That is, in Eq.~\eqref{eq:causal_process_def_2}, one considers the probability 
for all local experiments in $\X$ to be outside of the causal future of $A$, to display specific outcomes $o^{\X}$ and to have a specific causal order $\xi(\{A\} \cup \X)$ together with $A$.  
This probability is required to be independent of the setting of the local experiment $A$.

\medskip

We thus consider the coarse-grained set $\{A,B,C,\ldots \} = \Delta$, and first define a binary relation $\C_{\kappa(\Gamma)}$ on this set, for any given $\kappa(\Gamma)$, as follows: $X \C_{\kappa(\Gamma)} Y$ if and only if there exist some parties $\Lambda^{(1)}, \Lambda^{(2)}$, $\ldots$, $\Lambda^{(M)} \in \Delta$, $M \ge 0$, such that  $X_O \prec \Lambda_I^{(1)}$,  $\Lambda_O^{(i)} \prec \Lambda_I^{(i+1)}$ for $1 \le i \le M-1$, and $\Lambda_O^{(M)} \prec Y_I$. 
(We also allow for $M = 0$, which corresponds to the case where $X_O \prec Y_I$, without any ``intermediate'' $\Lambda$.)

As a first step, we show that, for all $\kappa(\Gamma)$ that occur with non-zero probability, this coarse-grained relation $\C_{\kappa(\Gamma)}$ is indeed a SPO on the set $\{A,B,C,\ldots \}$.
That is, we prove the following.

\begin{proposition}
$P(o_A,o_B,o_C,\ldots,\kappa(\Gamma)|i_A,i_B,i_C,\ldots) > 0$ only if $\C_{\kappa(\Gamma)}$ is a SPO on the set of coarse-grained variables $\Delta$.
\end{proposition}
\begin{proof}
It is straightforward to check that $\C_{\kappa(\Gamma)}$ is transitive.
We thus need to show its irreflexivity for all $\kappa(\Gamma)$ associated with a non-zero probability.
$\C_{\kappa(\Gamma)}$ is reflexive if and only if there exist $\Lambda^{(0)}, \Lambda^{(1)}$, $\ldots$, $\Lambda^{(M)} \in \Delta$, $M \ge 0$, such that $\Lambda_O^{(i)} \prec \Lambda_I^{(i+1)}$ for $0 \le i \le M - 1$ and $\Lambda^{(M)}_O \prec \Lambda_I^{(0)}$.
We prove by induction over $M$ that, for any $\kappa(\Gamma)$ for which such $\Lambda^{(0)}, \Lambda^{(1)}$, $\ldots$, $\Lambda^{(M)}$ exist, the probabilities $P(o_A,o_B,o_C,\ldots,\kappa(\Gamma)|i_A,i_B,i_C,\ldots)$ must be zero. 

We start by proving the claim for $M = 0$. 
Assume thus that, for the $\kappa(\Gamma)$ under consideration, we have $\Lambda^{(0)}_O \prec \Lambda^{(0)}_I$.
Since $I_{\Lambda^{(0)}} \prec \Lambda^{(0)}_O$ (as per condition~\eqref{eq:CL1} of the closed laboratory assumption in the main text), transitivity implies that, for this $\kappa(\Gamma)$, we have $I_{\Lambda^{(0)}} \prec \Lambda^{(0)}_I$. 
However, the probability $P(I_{\Lambda^{(0)}} \prec \Lambda^{(0)}_I|i_A,i_B,i_C,\ldots)$ (and therefore any probability $P(o_A,o_B,o_C,\ldots,\kappa(\Gamma)|i_A,i_B,i_C,\ldots)$ with $\kappa(\Gamma)$ satisfying $I_{\Lambda^{(0)}} \prec \Lambda^{(0)}_I$ that contributes to it) must be zero.
This follows from the fact that $P(I_{\Lambda^{(0)}} \prec \Lambda^{(0)}_I|i_A,i_B,i_C,\ldots)$ is independent of $i_{\Lambda^{(0)}}$ (since it is obtained by summing over various probabilities that are each individually independent of $i_{\Lambda^{(0)}}$ according to Eq.~\eqref{eq:free_choice} in the main text; namely, to obtain $P(I_{\Lambda^{(0)}} \prec \Lambda^{(0)}_I|i_A,i_B,i_C,\ldots)$ one sums the constraint of Eq.~\eqref{eq:free_choice} corresponding to $I_{\Lambda^{(0)}}$ over all $\Y,\Z$ such that $\Lambda^{(0)}_I \notin \Y \cup \Z$, all causal orders on the respective $\Y \cup \Z$, and all outcomes in the respective $\Y \cup \Z$), and $P(I_{\Lambda^{(0)}} \prec \Lambda^{(0)}_I|i_A,i_B,i_C,\ldots) = 0$ for the value  $i^*_{\Lambda^{(0)}}$ of $I_{\Lambda^{(0)}}$ (since, for $i^*_{\Lambda^{(0)}}$, we have $\Lambda^{(0)}_I \prec \Lambda^{(0)}_O$ with certainty, which is incompatible with $I_{\Lambda^{(0)}} \prec \Lambda^{(0)}_I$, as the latter would imply $\Lambda^{(0)}_O \prec \Lambda^{(0)}_I$ due to condition~\eqref{eq:CL1} in the main text).

Assume then that the claim holds up to some given $M \ge 0$---i.e., that whenever $\kappa(\Gamma)$ is such that some $\Lambda^{(0)}$, $\ldots$, $\Lambda^{(M)}$ with the prescribed properties exist (starting from any element of $\Delta$), then $P(o_A,o_B,o_C,\ldots,\kappa(\Gamma)|i_A,i_B,i_C,\ldots) = 0$. 
We will prove that the claim then also holds for $M+1$. 
We thus consider a $\kappa(\Gamma)$ such that there exist $\Lambda^{(0)}$, $\ldots$, $\Lambda^{(M + 1)}$, for which $\Lambda_O^{(i)} \prec \Lambda_I^{(i+1)}$ for $0 \le i \le M$ and $\Lambda^{(M + 1)}_O \prec \Lambda^{(0)}_I$.
We consider the sets of variables $\{\Lambda_I^{(i)}| 1 \le i \le M + 1\} \eqqcolon \Omega_I$ and $\{\Lambda_O^{(i)}| 0 \le i \le M\} \eqqcolon \Omega_O$, and distinguish between two possible cases.
First, consider the case where any of the variables in $\Omega_I \cup \Omega_O$ is in the causal future of $I_{\Lambda^{(M+1)}}$.
If $I_{\Lambda^{(M+1)}} \prec \Lambda_I^{(j)}$ for some $\Lambda_I^{(j)} \in \Omega_I$ (i.e. for some $j$ with $1 \le j \le M+1$), the closed laboratory assumption implies that $\Lambda^{(M+1)}_O \prec \Lambda_I^{(j)}$. Therefore, we have $\Lambda_O^{(j)} \prec \Lambda_I^{(j+1)}$, $\Lambda_O^{(j+1)} \prec \Lambda_I^{(j+2)}$, $\ldots$, $\Lambda^{(M+1)}_O \prec \Lambda_I^{(j)}$, and thus $\Lambda^{(j)} \C_{\kappa(\Gamma)} \Lambda^{(j)}$ with a number of intermediate $\Lambda$ strictly smaller than $M+1$, which is associated to a probability of zero by assumption.
And if $I_{\Lambda^{(M+1)}} \prec \Lambda_O^{(j)}$ for some $\Lambda_O^{(j)} \in \Omega_O$ (i.e. for some $j$ with $0 \le j \le M$), we have that $I_{\Lambda^{(M+1)}} \prec \Lambda_I^{(j+1)}$ because of transitivity.

In the second case, all variables in the sets $\Omega_I \cup \Omega_O$ are in the causal past or elsewhere of $I_{\Lambda^{(M+1)}}$. Therefore, $\kappa(\Gamma)$ contributes to the probability 
\begin{align}
\label{eq:pCycle}
&P(\Omega_I \cup \Omega_O \subseteq \P_{I_{\Lambda^{(M+1)}}} \cup \E_{I_{\Lambda^{(M+1)}}}, \Lambda_O^{(i)} \prec \Lambda_I^{(i+1)} \ \text{for} \ 0 \le i \le M, I_{\Lambda^{(M+1)}} \prec \Lambda^{(0)}_I|i_A,i_B,i_C,\ldots).
\end{align}
This probability is again independent of $i_{\Lambda^{(M+1)}}$, and zero for $i^*_{\Lambda^{(M+1)}}$, implying that it (and therefore also the probability for any $\kappa(\Gamma)$ contributing to it) is always zero.
Therefore, it follows that $P(o_A,o_B,o_C,\ldots,\kappa(\Gamma)|i_A,i_B,i_C,\ldots)$ is zero also in this case with $M + 1$.

To see that the probability in Eq.~\eqref{eq:pCycle} is indeed independent of $i_{\Lambda^{(M+1)}}$, note that it is again obtained by summing over various probabilities that are independent of $i_{\Lambda^{(M+1)}}$ according to Eq.~\eqref{eq:free_choice} in the main text. 
Namely, one considers the constraint of Eq.~\eqref{eq:free_choice}	 corresponding to $I_{\Lambda^{(M+1)}}$, and sums it over all $\Y,\Z$ such that $\Omega_I \cup \Omega_O \subseteq \Y \cup \Z$ and $\Lambda^{(0)}_I \notin \Y \cup \Z$, as well as over all causal orders on the respective $\Y \cup \Z$ that satisfy $\Lambda_O^{(i)} \prec \Lambda_I^{(i+1)} \ \text{for} \ 0 \le i \le M$, and over all outcomes in the respective $\Y \cup \Z$.

To see that the probability in Eq.~\eqref{eq:pCycle} is indeed zero for $i_{\Lambda^{(M+1)}}^*$, note that, for $i_{\Lambda^{(M+1)}}^*$, we have $\Lambda_I^{(M+1)} \prec \Lambda_O^{(M+1)}$, and thus, because of $\Lambda_O^{(M)} \prec \Lambda_I^{(M+1)}$, $\Lambda_O^{(M+1)} \prec \Lambda^{(0)}_I$ and transitivity, $\Lambda_O^{(M)} \prec \Lambda^{(0)}_I$.
Therefore, we have $\Lambda^{(0)} \C_{\kappa(\Gamma)} \Lambda^{(0)}$ with a number $M$ of intermediate $\Lambda$, which by assumption, implies that the probability associated to $\kappa(\Gamma)$ is zero.
\end{proof}

\medskip

The second step is then to prove that the probability $P(o_A,o_B,o_C,\ldots|i_A,i_B,i_C,\ldots)$ satisfies the definition of \emph{causal correlation} from Ref.~\cite{oreshkov16}, which we recalled above.

\begin{proposition}
$P(o_A,o_B,o_C,\ldots|i_A,i_B,i_C,\ldots)$ is a causal correlation, with the underlying probability distribution $P(o_A,o_B,o_C,\ldots,\xi(\Delta)|i_A,i_B,i_C,\ldots)$ associated to each SPO $\xi(\Delta)$ obtained by summing over all $\kappa(\Gamma)$ for which $\C_{\kappa(\Gamma)} = \xi(\Delta)$:
    \begin{equation}
        P(o_A,o_B,o_C,\ldots,\xi(\Delta)|i_A,i_B,i_C,\ldots) = \sum_{\{\kappa(\Gamma)|\C_{\kappa(\Gamma)} = \xi(\Delta)\}} P(o_A,o_B,o_C,\ldots,\kappa(\Gamma)|i_A,i_B,i_C,\ldots).
    \end{equation} 
\end{proposition}

\begin{proof}
It is clear that this distribution satisfies Eq.~\eqref{eq:causal_process_def_1}, since summing $P(o_A,o_B,o_C,\ldots,\xi(\Delta)|i_A,i_B,i_C,\ldots)$ over all $\xi(\Delta)$ is equivalent to summing $P(o_A,o_B,o_C,\ldots,\kappa(\Gamma)|i_A,i_B,i_C,\ldots)$ over all $\kappa(\Gamma)$.
It therefore remains to prove that this distribution satisfies Eq.~\eqref{eq:causal_process_def_2}.

In order to obtain the probability in Eq.~\eqref{eq:causal_process_def_2}, for a given $A$, $\X$ and $\xi(\{A\} \cup \X)$, one needs to sum up the probabilities $P(o_A,o_B,o_C,\ldots,\kappa(\Gamma)|i_A,i_B,i_C,\ldots)$ over all $o_X$ with $X \notin \X$, as well as over all ``fine-grained'' orders $\kappa(\Gamma)$ for which the ``coarse-grained'' order $\C_{\kappa(\Gamma)}$ satisfies the respective conditions---that is, with respect to $\C_{\kappa(\Gamma)}$, 
$(i)$ $A \npreceq \X$, and $(ii)$ the restriction of the coarse-grained order to $\{A\} \cup \X$ is precisely $\xi(\{A\} \cup \X)$. 
Expressed as conditions on the fine-grained $\kappa(\Gamma)$, this means that $\kappa(\Gamma)$ has to satisfy the following two constraints. 
\begin{align}
\label{eq:kappa_constraints}
&(i) \ \forall \ X \in \X: (A_O \npreceq X_I) \text{ and } (\nexists \ \Lambda^{(1)},\ldots,\Lambda^{(M)}  \text{ s.t. }  A_O \prec \Lambda^{(1)}_I, \Lambda^{(i)}_O \prec \Lambda^{(i+1)}_I \text{ for } 1 \le i \le M - 1, \text{ and } \Lambda^{(M)}_O \prec X_I) \notag \\
&(ii) \ \forall \ X, Y \in \{A\} \cup \X: [X \prec Y \text{ w.r.t } \xi(\{A\} \cup \X)] \notag \\
&\hspace{10mm} \Leftrightarrow [(X_O \prec Y_I) \text{ or }  (\exists \ \Lambda^{(1)}, \ldots, \Lambda^{(M)} \text{ s.t. } X_O \prec \Lambda^{(1)}_I, \Lambda^{(i)}_O \prec \Lambda^{(i+1)}_I \text{ for } 1 \le i \le M - 1, \text{ and } \Lambda^{(M)}_O \prec Y_I)].
\end{align}

To proceed, we note the following two observations. 
\begin{observation}
\label{obs:observation1}
For any given $\kappa(\Gamma)$,  whether or not the two constraints $(i)$ and $(ii)$ in Eq.~\eqref{eq:kappa_constraints} are satisfied is completely determined by the sets $\P_{I_A}$ and $\E_{I_A}$, and the causal order on $\P_{I_A} \cup \E_{I_A}$.
\end{observation}
\begin{observation}
\label{obs:observation2}
For all $\kappa(\Gamma)$ that satisfy the constraints $(i)$ and $(ii)$ in Eq.~\eqref{eq:kappa_constraints}, all $O_X$, $X \in \X$, are in the causal past or elsewhere of $I_A$.
\end{observation}

The proof of observation~\ref{obs:observation1} is given below.
Observation~\ref{obs:observation2} follows from $A_O \npreceq X_I$ and the closed laboratory assumption, which imply that $I_A \npreceq O_X$.

From the two observations~\ref{obs:observation1} and~\ref{obs:observation2}, it follows that the probability in Eq.~\eqref{eq:causal_process_def_2} can be obtained by summing over various probabilities which are each independent of $i_A$ according to Eq.~\eqref{eq:free_choice} in the main text.
Namely, to obtain the probability in Eq.~\eqref{eq:causal_process_def_2}, one sums the probability in Eq.~\eqref{eq:free_choice} over all $\Y$, $\Z$ and $\kappa(\Y \cup \Z)$ which are such that the constraints $(i)$ and $(ii)$ in Eq.~\eqref{eq:kappa_constraints} are indeed satisfied, and marginalises over all $O_X \in \Y \cup \Z$ with $X \notin \X$ (so that precisely the $O_X$ with $X \in \X$ remain).
\end{proof}
In the following, we prove Observation~\ref{obs:observation1} used above.

\begin{proof}[Proof of Observation~\ref{obs:observation1}]
In the constraint $(i)$, $A_O \npreceq X_I$ can equivalently be replaced by $I_A \npreceq X_I$, and $A_O \prec \Lambda^{(1)}_I$ can equivalently be replaced by $I_A \prec \Lambda^{(1)}_I$ due to the closed laboratory assumption. 
Also, in the constraint $(i)$, we may without loss of generality take all $\Lambda_O^{(i)}$ and $\Lambda_I^{(i+1)}$ with $1 \le i \le M - 1$ (and $\Lambda_O^{(M)}$) to be in the causal past or elsewhere of $I_A$. Namely, whenever constraint $(i)$ is violated for some $\Lambda^{(1)},\ldots,\Lambda^{(M)}$, and $I_A \prec \Lambda_I^{(j+1)}$ for some $j$ with $1 \le j \le M - 1$ (or $I_A \prec \Lambda_O^{(j)}$, which implies $I_A \prec \Lambda_I^{(j+1)}$ by transitivity), the constraint is also violated for $\Lambda^{(j+1)}$, $\ldots$, $\Lambda^{(M)}$, i.e., we can ``skip'' $\Lambda^{(1)}$, $\ldots$, $\Lambda^{(j)}$ (and if $I_A \prec \Lambda_O^{(M)}$, we would have $I_A \prec X_I$ and thus $A_O \prec X_I$ due to the closed laboratory assumption).
Therefore, whether $\kappa(\Gamma)$ satisfies the constraint $(i)$ or not is completely determined by what variables are contained in  $\P_{I_A}$ and $\E_{I_A}$, and the causal order on the subset $\P_{I_A} \cup \E_{I_A}$.

Furthermore, if constraint $(i)$ holds, the right-hand side of the constraint $(ii)$ (i.e. the part in the second square bracket) can only be true if all variables $X_O$, $Y_I$ and $\Lambda_I^{(i)}$, $\Lambda_O^{(i)}$ for $1 \le i \le M$ considered there are in the causal past or elsewhere of $I_A$ (otherwise, we would have $A \prec Y$).
Therefore, whether this right-hand side is true or not---and thus, whether $\kappa(\Gamma)$ satisfies the constraint $(ii)$ or not---is also completely determined by $\P_{I_A}$ and $\E_{I_A}$, and the causal order on $\P_{I_A} \cup \E_{I_A}$.
\end{proof}

\medskip

For completeness, we recall here a simpler characterisation of causal correlations, which is as follows~\cite{oreshkov16,abbott16}.

\medskip

\textbf{Causal correlation (alternative characterisation~\cite{oreshkov16,abbott16}).} \textit{For one local experiment (i.e. $\Delta = \{A\}$), any correlation $P(o_A|i_A)$ is causal. 
For multiple local experiments, a correlation $P(o_A,o_B,o_C,\ldots|i_A,i_B,i_C,\ldots)$ is causal if and only if it can be decomposed as 
\begin{equation}
P(o_A,o_B,o_C,\ldots|i_A,i_B,i_C,\ldots) = \sum_{X \in \Delta} q_X \, P_X(o_X|i_X) \, P_{X,i_X,o_X}(o_{\Delta \backslash \{X\}} | i_{\Delta \backslash \{X\}}), \label{eq:causal_correl}
\end{equation}
with $q_X\ge 0$, $\sum_{X \in \Delta} q_X = 1$, where (for each $X$) $P_X(o_X|i_X)$ is a single-partite (and hence causal) correlation and (for each $X,i_X,o_X$) $P_{X,i_X,o_X}(o_{\Delta \backslash \{X\}} | i_{\Delta \backslash \{X\}})$ is a causal correlation for the parties in $\Delta \backslash \{X\}$.\\
}

This characterisation can be intuitively interpreted as describing an iterative ``unravelling'' of the local experiments
in a sequence, where, with some probability, one local experiment occurs first, then, according to some probability which depends on the setting and outcome 
of this first local experiment, another local experiment occurs second, and so on~\cite{oreshkov16,abbott16}.


\section*{Supplementary Note 6---Which assumptions are violated?}
\label{app:CI_assumptions}
The conclusion that in an experiment that admits a description in terms of standard causal evolution in spacetime there exist physical variables which violate causal inequalities without manifestly violating the closed laboratory and free choice assumptions naturally raises the question of whether we could gain further insights into the way the causal inequality assumptions are violated. Is it possible that, in spite of the outlined considerations about the observable causal relations between the concerned variables, upon a closer inspection of the circuit in Fig.~\ref{fig:tripartite_example} we would find a sense in which the free choice or closed laboratory assumption is violated for these variables, rather than the existence of causal order \textit{per se}? In particular, the circuit describes a sequence of interactions between sets of systems that one may intuitively associate with the three different parties. This seems to violate the closed laboratory assumption, which essentially stipulates that each party is involved in a single round of information exchange, where they receive information about the past through the input system $X_I$ and subsequently send out information into the future through the output system $X_O$. 

It is crucial to realise that the causal inequality assumptions concern concrete variables, which in our case we have explicitly specified. As we will see, these variables are not the same as what one might intuitively assume if one thinks of this experiment as involving three laboratories existing through time that exchange information with each other (the parties Alice, Bob, and Charlie that operate on the time-delocalised process must be understood abstractly as agents who control the parameters that describe the operations taking place on the time-delocalised systems). The possibility of understanding the experiment in terms of other variables for which a causal order exists and for which the closed laboratory assumption might appear to be violated is not in contradiction with the claim that no such interpretation is available for the variables of interest.  
To investigate whether such an interpretation is available for these variables, first note that any claimed causal order on the variables should have operational grounds---otherwise it is always possible to imagine some fictitious causal order and a violation of the free choice or closed laboratory assumptions relative to it so as to ``explain'' the observed correlations. In the present scenario, we are unaware of any other operationally grounded notion of causal order that one could invoke apart from the one imposed by spacetime, which could be further constrained by the lack of physical interaction between specific variables at different times (as in when a laboratory is kept ``closed'' between the time of input and the time of output). We will therefore focus on the question of whether there could exist a compelling interpretation in which the variables we have identified can be seen as taking place at definite, although possibly random and dependent on other variables, spacetime locations such that the free choice or closed laboratory assumption is violated. We will argue that no such interpretation is supported by the spatiotemporal description of the experiment (Fig.~\ref{fig:tripartite_example}). On the contrary, a careful analysis of the link between the time-delocalised variables of interest and the time-local variables in the circuit leads to the conclusion that the time-delocalised variables we have identified cannot be interpreted as taking place at definite locations in the background spacetime, which is what the causal inequality violation witnesses. 

First, observe that the free choice assumption is trivially compatible with the structure of the circuit in Fig.~\ref{fig:tripartite_example}: the variables $I_X$ can all be defined at the initial time, so irrespectively of what spacetime locations we may attribute to the remaining variables (which would all be in the future), there is no reason why these variables could not be chosen freely. (Of course, no matter how we choose these variables in practice, one can never exclude the in-principle possibility that there is a hidden common cause for these variables and some of the other relevant variables, but this conspiratorial possibility is trivially always present and is obviously not suggested by the circuit.) 

To discuss the plausibility of the closed laboratory assumption, we need to first identify a reasonable candidate causal order on the variables, since this assumption is formulated in terms of a causal order. For the operations of Alice and Bob, there is a natural interpretation of the experiment as describing the occurrence of these operations at definite spacetime locations since these operations effectively take place within a ``classical switch''---each can be thought of as taking place at one of two possible times determined by the state of a control bit. This is compatible with the time-delocalised input and output variables $A_I$, $A_O$, $B_I$, $B_O$ (see Supplementary Note~\hyperref[app:tripartite_example_classical]{4~B}), which effectively reduce to the respective time-local input and output variables conditionally on the state of the control bit. As each of these effective operations is a standard time-localised operation from an input to an output system, which involves the ancillas of the respective party and constitutes the only interaction with those ancillas, the closed laboratory assumption for the respective party is manifestly respected. 

The operation $U_C$ of Charlie, however, is delocalised in time in a different way, which can intuitively be described as follows: the operation either happens at the beginning of the circuit (in the case when the controlled operation in the last stage is not triggered and hence that controlled operation can be effectively omitted), or it happens at the very end of the circuit (when the controlled operation in the last stage is triggered and it effectively results in ``undoing'' the effect of the first $U_C$ (by the action of $U_C^{\dagger}$) followed by applying the NOT gate $\sigma_\textsc{x}$ and then $U_C$ anew). Notice that if we interpret the operation of Charlie in this way, it is again a standard operation taking place at one of two possible times that manifestly respects the closed laboratory assumption. How is it then possible that we obtain a causal inequality violation? 

The answer is that the interpretation just outlined makes no operational sense, which also transpires in the fact that it contradicts the causal order of events in spacetime. Indeed, if we think that the operation of Charlie takes place at one of the two possible times, we must conclude that whether it happens at the earlier time or not (which is a variable associated with that time) depends on the state of the controlled bits at the end, which itself can be influenced by the operations of Alice and Bob. This would amount to influence by Alice and Bob on the past, which is in contradiction with spacetime causality. The error leading to this contradiction is in not recognizing that in order to say that a given variable such as $C_I$ or $C_O$ ``takes place'' at a given time, it must in principle be possible for an agent at the same time to know with certainty whether this is the case. This is clearly not possible at the first of the two candidate temporal locations for the operation $U_C$ since the state of the control bits is not yet known at that time. This is in contrast to the situation for $U_A$ or $U_B$ whose time depends on a variable in the past that can be known at the time of the operation. 

An analogous problem arises if one attempts any other obvious ``localisation'' of these variables, such as for instance the one suggested by the isomorphism in Fig.~\ref{fig:iso_bw}. In that case, $C_I$ could be thought of as effectively corresponding to $P_3$ at the initial time but up to a NOT gate that is controlled by variables in the future. Associating $C_I$ with the initial time would, however, be in contradiction with spacetime causality and, as before, makes no sense since $C_I$ cannot be known at that time. 

We now provide a general proof in the case when the process is treated as quantum, that the quantum input system $C_I$ cannot be effectively associated with definite times in principle. The operations of Alice and Bob will be interpreted as localised in time conditionally on the computational basis of $Q_1$ as discussed previously. 

Consider the case where each of the unitaries $U_A$ and $U_B$ is a SWAP operation (sending $X_I'$ to $X_O$ and $X_I$ to $X_O'$, $X=A,B$), and $U_C$ is the identity channel from $C_I$ to $C_O$ (hence no explicit ancilla for Charlie needs to be considered). This situation is depicted in Fig.~\ref{fig:AppF_Fig1}. 

\begin{figure}[h]
    \centering
    \includegraphics[width=0.7\textwidth]{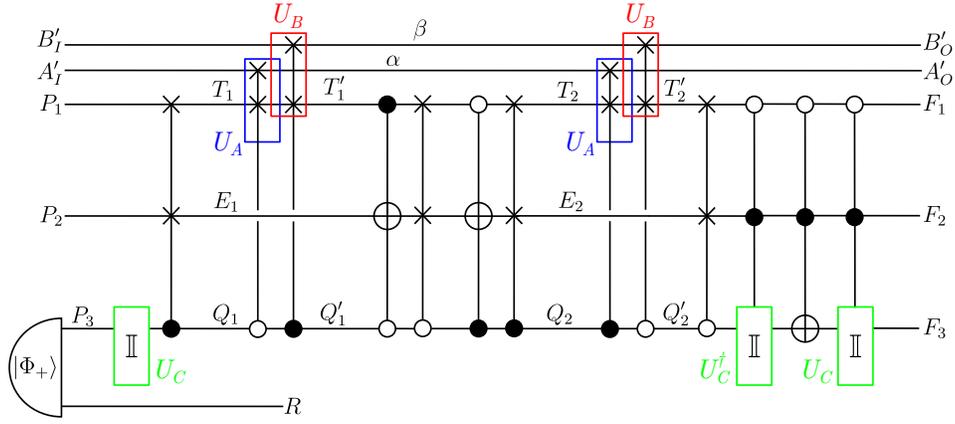}
    \caption{Temporal circuit of Fig.~\ref{fig:tripartite_example}, with SWAP operations for Alice and Bob, an identity channel for Charlie and an additional, time-local ``reference'' system $R$. The systems $P_3$ and $R$ are prepared in a maximally entangled state $\ket{\Phi_+}$ at the beginning of the circuit.}
    \label{fig:AppF_Fig1}
\end{figure}

As depicted in the figure, we also introduce an additional, time-local, ``reference'' system $R$ at the beginning of the circuit. Since $R$ is separate from $C_I$, the joint Hilbert space of $R$ and $C_I$ is $\mathcal{H}^{C_I R} = \mathcal{H}^{C_I} \otimes \mathcal{H}^{R}$. (Note that arbitrary preparations and measurements could be applied on the system $\mathcal{H}^{C_I} \otimes \mathcal{H}^{R}$ by trivially extending the previously described procedures for $C_I$ onto $R$, hence this is an operationally meaningful Hilbert space.)  Let the joint system $P_3 R$ be initially prepared in the state $\ket{\Phi_+} = (\ket{00} + \ket{11})/\sqrt{2}$ (Fig.~\ref{fig:AppF_Fig1}). Considering the precise definition of $C_I$ through the isomorphism in Fig.~\hyperref[fig:iso_bw]{18(a)}, we will now argue that $C_I$ cannot be consistently associated with any definite time since no quantum system at a definite time contains the correct information implied by this definition. 

Let us restrict our attention to the cases where the systems $A_I'$ and $B_I'$ are each prepared in one of the computational basis states $\ket{0}$ or $\ket{1}$. Through the SWAP operations performed by Alice and Bob, these states are transferred onto $A_O$ and $B_O$, respectively. Now, according to the isomorphism in Fig.~\hyperref[fig:iso_bw]{18(a)}, for any such combination of these states, which could be freely chosen by Alice and Bob, the state on the system $C_I$ reduces to the state of $P_3$, or to the state of $P_3$ up to a NOT gate $\sigma_\textsc{x}$. This property holds for the full operator algebras on the respective systems and thus remains true also when we regard $C_I$ and $P_3$ as parts of larger systems that include $R$: depending on the computational basis states output by Alice and Bob, $C_IR$ reduces to $P_3R$ or $P_3 R$ up to a NOT gate on $P_3$. In our case, this means that the state on $C_IR$ is either $(\ket{00} + \ket{11})/\sqrt{2}$ or $(\ket{10} + \ket{01})/\sqrt{2}$. Furthermore, by choosing the combination of basis states that they output, Alice and Bob can fully determine which of these two orthogonal states is received on $C_I R$. But since these two states are perfectly distinguishable, this means that we have perfect signalling from $A_I'$ and $B_I'$ to $C_IR$. This, in turn, means that $C_I$ could not possibly exist prior to the time of the controlled operations of Alice and Bob in the circuit, even with small nonzero probability (unless it is defined as a subsystem that overlaps with $A_I'$ and $B_I'$, which however would contradict the premise that it is a subsystem separate from $A_I'$ and $B_I'$). Indeed, let us assume that with some nonzero probability $C_I$ can be associated with such prior times. For the same reasons as explained in the previous examples, whether this is the case or not cannot depend on the choices of Alice and Bob that are only used in the future, since whether the variables ``take place'' at a given time must be possible to know at the respective time. But then in the hypothetical cases in which $C_I$ takes place prior to the operations of Alice and Bob, Alice and Bob could make different choices altering the state of $C_I R$, which is in contradiction with spacetime causality. Therefore, $C_I$ could not be associated with any time prior to the operations of Alice and Bob. 

We will now show that $C_I$ cannot be consistently associated with only later times either. Let us consider the case in which $A_I'$ is prepared in the state $\ket{1}$, $B_I'$ in the state $\ket{0}$, $P_1$ in the state $\ket{0}$ and $P_2$ in the state $\ket{1}$. In this case, it is straightforward to verify that the controlled operations at the end of the circuit are not triggered and the state on $C_I R$ is $(\ket{00} + \ket{11})/\sqrt{2}$. Since the state of $R$ is purified on $C_I$, if $C_I$ exists at any time after the beginning of the operations of Alice and Bob, it must be possible to find the purification of $R$ at this time. However, if we track this information in time, we see that it is ``lost'' as soon as the first controlled operation of Alice or Bob happens. Indeed, following the evolution of the computational basis it is straightforward to verify that the controlled-SWAP gate immediately after the first gate $U_C$ correlates $R$ with $T_1$ in the computational basis, and this correlation is propagated onto the ancillas of Alice and Bob by the first controlled operations of Alice and Bob, remaining there until the end of the circuit. (Note that for the natural time localisation of the operations of Alice and Bob assumed here, the ancillary wire $\alpha$ coincides with $A_I'$ or $A_O'$ depending on the time at which $U_A$ takes place, and similarly $\beta$ coincides with $B_I'$ or $B_O'$. Since by definition $C_I$ is separate from these ancillary systems, it cannot overlap with them.)

To summarise, we have shown that the time-delocalised quantum subsystem $C_I$ that we have identified cannot be effectively localised in time, since there exists no time-local subsystem that contains the correct information as per the definition of $C_I$. For this argument, it was essential that we treated the system as a quantum system as we used the property that quantum information cannot be copied in order to argue that no subsystems could exists that contains the required quantum information at any given time. This argument does not automatically imply that the classical variable corresponding to the computational basis of $C_I$ could not be effectively localised. Indeed, in the example we considered, a copy of this variable remains available through time and thus this variable in principle could be declared associated with the future. (As in the quantum case, we can rule out the possibility that this variable takes place in the past even with small probability since Alice and Bob can fully determine its value.) However, if we assume that the classical variable $C_I$, however localised in time, should not overlap with any ancillary systems that could be introduced for Charlie, we can easily perform a version of the quantum argument leading to the same conclusion: we now introduce an ancilla for Charlie and let Charlie perform a SWAP operation, which already at the start of the circuit transfers the value of $P_3$ onto this ancilla (see Fig.~\ref{fig:AppF_Fig2}). This makes it impossible to find a system that contains the correct information expected to live on $C_I$ at later times only, except if that subsystem is defined to overlap with the ancillas.  

\begin{figure}[h]
    \centering
    \includegraphics[width=0.7\textwidth]{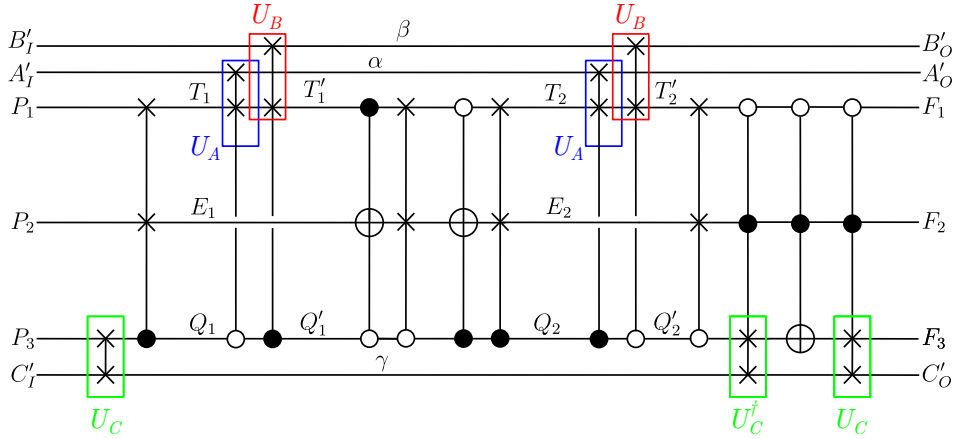}
        \caption{Temporal circuit of Fig.~\ref{fig:tripartite_example}, with SWAP operations for Alice, Bob and Charlie.}
    \label{fig:AppF_Fig2}
\end{figure}

In conclusion, we have shown that some of the time-delocalised variables that we have identified do not admit an effective localisation in time by showing that there do not exist time-local variables that take their value. This prevents us from assigning a definite causal order on all variables, which is arguably what the violation of the causal inequality witnesses. Our argument for the purely classical case is somewhat weaker than the quantum case as it assumes that the variable $C_I$ does not live on the intermediate ancillary wire with which Charlie is made to interact. While this seems intuitive, this assumption could in principle be debated since the intermediate ancillary wire of Charlie need not coincide with $C_I'$ or $C_O'$, which must be separate from $C_I$ by definition. It would be interesting to investigate whether this assumption could be relaxed as in the quantum case. 

\bibliography{literatur}
\end{document}